\documentclass[aps,prx,twocolumn,10pt,notitlepage,longbibliography,superscriptaddress]{revtex4-2}

\usepackage[utf8]{inputenc}
\usepackage[T1]{fontenc}
\usepackage{lmodern}
\usepackage{microtype}

\usepackage{newunicodechar}
\newunicodechar{Ș}{\textcommabelow{S}}
\newunicodechar{ș}{\textcommabelow{s}}
\newunicodechar{Ț}{\textcommabelow{T}}
\newunicodechar{ț}{\textcommabelow{t}}

\usepackage{amsmath,amssymb,amsthm}
\usepackage{mathtools}
\usepackage{tikz-cd}      
\usepackage{quantikz}     

\usepackage{enumitem}
\usepackage{booktabs}
\usepackage{graphicx}
\usepackage{manfnt}       
\usepackage{marginnote}   
\usepackage{subcaption}
\usepackage{booktabs}
\usepackage{multirow}
\usepackage{siunitx}
\usepackage{xfp}
\usepackage[colorlinks=true,linkcolor=blue,citecolor=blue,urlcolor=blue]{hyperref}
\usepackage[capitalize]{cleveref}   

\ExplSyntaxOn
\AddToHook{env/thebibliography/begin}{
  \providecommand\bibinfo[2]{#2}
  \cs_set_eq:NN \paper_bibinfo:nn \bibinfo
  \RenewDocumentCommand\bibinfo{m m}{
    \str_if_eq:nnTF {#1} {note}
      {
        \regex_match:nnTF {\A https?://} {#2}
          {\begingroup\urlstyle{same}\url{#2}\endgroup}
          {\paper_bibinfo:nn{#1}{#2}}
      }
      {\paper_bibinfo:nn{#1}{#2}}
  }
}
\ExplSyntaxOff

\newtheorem{theorem}{Theorem}
\newtheorem{proposition}{Proposition}
\newtheorem{lemma}{Lemma}
\newtheorem{corollary}{Corollary}

\newtheorem*{xmeasurementrestatement}{Theorem~\ref{thm:x-measurements}}
\newtheorem*{efficientsimulationrestatement}{Theorem~\ref{thm:efficient-simulation}}
\newtheorem*{dcpnormalformrestatement}{Proposition~\ref{prop:dcp-normal-form}}
\newtheorem*{dcpzmeasurementrestatement}{Proposition~\ref{prop:z-update}}

\theoremstyle{definition}
\newtheorem{definition}{Definition}
\newtheorem{example}{Example}

\theoremstyle{remark}
\newtheorem{remark}{Remark}

\numberwithin{equation}{section}
\renewcommand{\theequation}{\arabic{section}.\arabic{equation}}

\newcommand{\Z}{\mathbb{Z}}

\newcommand{\C}{\mathbb{C}}
\newcommand{\F}{\mathbb{F}}

\DeclareMathOperator{\supp}{supp}
\DeclareMathOperator{\Fix}{Fix}
\DeclareMathOperator{\rank}{rank}

\DeclareMathOperator{\image}{Im}
\renewcommand{\Im}{\image}
\renewcommand{\S}{S}
\newcommand{\T}{T}
\newcommand{\Hgate}{H}
\DeclareMathOperator{\CNOT}{CNOT}
\DeclareMathOperator{\CCZ}{CCZ}
\DeclareMathOperator{\CS}{CS}
\DeclareMathOperator{\CZ}{CZ}
\DeclareMathOperator{\Span}{span}

\newcommand{\gcomm}[2]{\mathopen{[\![}#1,#2\mathclose{]\!]}}

\renewcommand{\ket}[1]{|#1\rangle}
\renewcommand{\bra}[1]{\langle#1|}

\renewcommand{\vec}[1]{#1}

\begin{document}

\title{Polynomial-time simulation of non-Clifford quantum error correction}
\author{Șerban Cercelescu}
\email{serban-ion.cercelescu@cs.ox.ac.uk}
\affiliation{Department of Computer Science, University of Oxford, Oxford, UK}
\author{Mark Koch}
\email{mark.koch@cs.ox.ac.uk}
\affiliation{Department of Computer Science, University of Oxford, Oxford, UK}
\author{Arthur Pesah}
\email{arthur.pesah@fu-berlin.de}
\affiliation{Dahlem Center for Complex Quantum Systems, Freie Universität Berlin, 14195 Berlin, Germany}

\begin{abstract}
We show that every intermediate state of a broad class of non-Clifford quantum error-correction circuits is a third-order phase-polynomial state, even under stochastic Pauli noise. The class includes magic state distillation and cultivation, code switching and gauge fixing, transversal non-Clifford gates with syndrome extraction, and injection of diagonal magic states. Third-order phase-polynomial states strictly generalize stabilizer states, and we prove that they are exactly the stabilizer states of the diagonal-Clifford-and-Pauli (DCP) stabilizer formalism that we introduce. We characterize these states and show that their representations can be updated in polynomial time. This yields an exact polynomial-time simulation algorithm for these circuits and, more broadly, a framework for reasoning about their internal states and about magic states in general. We provide an open-source implementation, \texttt{merlin}, and benchmark it against existing non-Clifford simulators on distillation, cultivation, and code switching circuits. The benchmarks demonstrate improved runtime and memory scaling as the number of logical outputs grows in Bravyi-Haah distillation, and simulation of a code switching circuit beyond the reach of all other tested simulators.
\end{abstract}
\maketitle

\begingroup
\makeatletter
\let\l@subsection\@gobbletwo
\let\l@subsubsection\@gobbletwo
\makeatother
\tableofcontents
\endgroup

\section{Introduction}

Designing a fault-tolerant quantum architecture requires choosing an error-correcting code, a syndrome-extraction procedure, and protocols for implementing logical gates. Comparing these choices often requires circuit-level simulation under physically motivated noise models, since their logical error rates and thresholds are generally difficult to determine analytically. Efficient Clifford simulation tools, such as \texttt{stim} \cite{gidney2021stim}, have made it possible to evaluate and compare a wide range of quantum memory experiments and logical Clifford operations.

Extending this capability to non-Clifford operations, which are required for universal quantum computation, remains a challenge. Numerous fault-tolerant protocols have been proposed, either to implement such operations directly or to prepare encoded \emph{magic states} that enable them through state injection. These include magic state distillation \cite{bravyi2005universal,bravyi2012magic,campbell2017unified,litinski2019magic}, cultivation \cite{gidney2024magic,vaknin2026high}, code switching \cite{beverland2021cost,bombin2014gauge,kubica2015universal,tan2025singleshot}, and pieceable fault tolerance \cite{yoder2016universal}. Recent constructions have strengthened the prospects for preparing magic states at high rates, including constant-overhead distillation protocols based on codes with nonvanishing encoding rate \cite{golowich2024asymptotically, golowich2024quantum, he2025asymptotically, li2026transversal, nguyen2024good, scruby2025quantum, steinberg2025universal, tan2025singleshot, wills2024constantoverhead, wills2026improved, xu2025batched, zhu2025nonclifford}. Such approaches reduce the overhead per output by processing many magic states within a growing code block, making their performance at large numbers of logical qubits particularly important to assess.

The need to evaluate these protocols has motivated recent simulation algorithms and software libraries tailored to non-Clifford quantum error correction \cite{haenel2026tsim,chase2026clifft,umbrarescu2026syqma,fang2026symft,surti2026efficient}. These methods substantially extend the range of accessible circuits, but generally retain potentially exponential runtime dependencies. In particular, Surti, Daguerre, and Kim \cite{surti2026efficient} establish polynomial scaling in the number of physical qubits for broad families of magic state preparation protocols when the number $k$ of logical qubits is fixed. Their complexity also depends on a stabilizer or Pauli decomposition of the target logical state, whose size can grow exponentially with $k$. Consequently, efficient simulation at fixed $k$ does not guarantee efficient simulation of high-rate protocols, where the number of logical outputs grows with the block size.

This motivates simulation methods that scale polynomially in both the physical and logical system sizes, even as the number of non-Clifford operations grows. In this work, we identify a broad class of circuits relevant to fault-tolerant magic state preparation that admit such simulation under stochastic Pauli noise.
\subsection{Summary of results}

We identify a class of non-Clifford quantum circuits that admit classical simulation in time polynomial in the number of qubits and the circuit size. The supported operations are:
\begin{itemize}
    \item Preparation of $|0\rangle$ and $|+\rangle$ states;
    \item Pauli, $\CNOT$, and $\T$ gates;
    \item Arbitrary $Z$-type Pauli measurements and compatible $X$-type Pauli measurements.
\end{itemize}
To describe the measurement restriction, we first identify the states generated by the unitary operations above. A \emph{third-order phase polynomial (3PP) state} is a uniform-magnitude superposition over an affine subspace of computational-basis labels (that is, the solution set of a system of linear equations over $\F_2$), with phases specified by a third-order phase polynomial over $\mathbb{Z}_8$. These are exactly the states obtainable from the specified preparations using Pauli, $\CNOT$, and $\T$ gates. A projective Pauli measurement is \emph{compatible} with a 3PP state if the post-measurement state is a 3PP state for every outcome having nonzero probability. We call circuits composed of the operations above, with each measurement compatible on every reachable branch, \emph{3PP circuits}.

The simulation maintains an exact, efficient representation of the conditional state throughout the circuit and samples intermediate measurement outcomes according to their probabilities. It also supports exact evaluation of computational-basis amplitudes in a fixed global-phase convention, exact Pauli expectation values, computational-basis sampling, and equality testing between represented states up to global phase. Stochastic Pauli noise is incorporated by sampling faults and evolving the corresponding conditional states, provided that measurements remain compatible along the noisy trajectories.

Alongside the phase-polynomial description, we develop an equivalent \emph{diagonal-Clifford-and-Pauli (DCP) stabilizer formalism}, in which states are uniquely specified, up to global phase, by stabilizing operators generated by Pauli gates and diagonal Clifford gates. We prove that a state is a DCP stabilizer state if and only if it is a 3PP state. We establish structural properties of DCP stabilizer groups, construct efficient tableau representations, and develop efficient update rules for both the tableau and phase-polynomial descriptions.

We give necessary and sufficient conditions for an $X$-type Pauli measurement to be compatible, together with polynomial-time compatibility tests and conditional-state updates. These results extend efficient simulation beyond unitary CNOT-dihedral circuits---circuits made of $\T$, $\CNOT$, and Pauli gates \cite{cross2016scalable}---to circuits with intermediate measurements and polynomial-time classical feedforward.

Our main result for quantum error correction reduces compatibility under noise to a condition on the ideal circuit.
\Cref{thm:iterative-measurements} shows that, if every ideal $X$ measurement is deterministic or anticommutes with a current $Z$-type stabilizer, then every measurement remains compatible under arbitrary inserted DCP faults, provided that outcome-dependent operations are restricted to DCP corrections.
Thus, these ideal-circuit conditions certify exact polynomial-time simulation of the entire noisy protocol, without a separate analysis of noisy branches.
This covers broad families of magic state preparation protocols, including distillation, cultivation, and code switching, as well as non-Clifford logical operations interleaved with syndrome extraction.

Beyond circuits that preserve the DCP state class, we introduce \emph{DCP rank}: the minimum number of DCP stabilizer states needed to express a state as a linear combination. Such decompositions extend the framework to universal circuit simulation, with a cost that depends on the number of terms. Since every Pauli stabilizer state is a DCP state, DCP rank is no greater than stabilizer rank, although it can still grow exponentially.

Finally, we provide \texttt{merlin}~\cite{koch2026merlin}, an open-source implementation of our simulation methods. We compare its runtime and memory usage with existing non-Clifford simulation libraries on distillation, cultivation, and code switching circuits.
The benchmarks demonstrate improved scaling as the number of logical outputs grows in Bravyi-Haah distillation \cite{bravyi2012magic}, and simulation of a recent code switching circuit based on a $[[81,3,5]]$ code \cite{li2026transversal} on which all other tested libraries exhaust memory or fail.
They also illustrate the complementary strengths of existing methods, which achieve lower per-shot runtimes on the single-output distillation and cultivation examples.

\subsection{Related work}

\subsubsection{Non-Clifford simulation algorithms}

Classical simulation beyond Clifford circuits includes stabilizer decompositions \cite{bravyi2016improved,bravyi2019simulation}, quasiprobability and Pauli-propagation methods \cite{pashayan2015estimating,rall2019simulation}, and hybrid stabilizer--tensor-network representations \cite{masot-llima2024stabilizer,nakhl2025stabilizer}.
Recent work has adapted and optimized these approaches for quantum error correction.
Tsim \cite{haenel2026tsim} combines stabilizer decompositions with ZX-calculus simplification and parametric compilation \cite{kissinger2022classical,sutcliffe2025fast}, while cultivation-specific decompositions demonstrate that circuit structure can substantially reduce the number of required Clifford terms \cite{wan2025cutting,wan2026simulating}.
SOFT \cite{li2025soft}, Clifft \cite{chase2026clifft}, and SymFT \cite{fang2026symft} combine stabilizer or Clifford-frame descriptions with explicit coefficients for non-stabilizer degrees of freedom.
Other approaches include sparse Pauli-frame representations with truncation \cite{tuloup2026computing} and polynomial-memory symbolic simulation through auxiliary qubits, as implemented in SyQMA \cite{umbrarescu2026syqma}.

More specifically exploiting magic-state preparation, Surti, Daguerre, and Kim \cite{surti2026efficient} propagate circuit-level Pauli faults into Clifford errors, with complexity depending on a decomposition of the target magic state.
Magdalena de la Fuente et al.~\cite{fuente2026highthreshold} develop a complementary spacetime approach for efficiently sampling syndrome statistics in a class of non-Clifford CSS-type circuits.
Our approach establishes polynomial-time simulation for circuits preserving the DCP state class, including compatible intermediate measurements, without requiring a small $\T$ count or a small number of logical magic states.

\subsubsection{Simulating \texorpdfstring{$\CNOT+\T$}{CNOT+T} circuits}

The unitary part of 3PP circuits consists of $\CNOT$, $\T$, and Pauli gates acting on qubits initialized in $\ket{0}$ or $\ket{+}$. Circuits with this gate set have been studied in several contexts in the quantum information literature.

This gate set is closely related to the Z-network formulation of Instantaneous Quantum Polynomial-time (IQP) circuits \cite[Section 5]{shepherd2009instantaneous}. With $\ket{+}$ inputs and $X$-basis measurements, Z-networks generated by $\CNOT$, $X$, and $e^{i\theta Z}$ are equivalent to IQP sampling up to efficient classical post-processing. Restricting the rotation angles to integer multiples of $\pi/8$ gives our gate set, up to global phase. For this choice of angles, Shepherd subsequently showed that marginal distributions on logarithmically many output bits can be computed efficiently \cite[Theorem 4]{shepherd2010binary}.

The unitary group generated by $\CNOT$, $\T$, and Pauli gates was subsequently named the CNOT-dihedral group and studied in the contexts of randomized benchmarking and circuit synthesis \cite{cross2016scalable, amy2018finite, garion2020synthesis}. Efficient polynomial representations and composition algorithms were established in \cite{cross2016scalable}, while \cite{amy2018finite} developed a finite presentation using phase polynomials.

Here, we characterize when $X$ measurements preserve the class of 3PP states and derive efficient conditional-state update rules, enabling simulation of circuits with intermediate measurements satisfying these conditions.

\subsubsection{Generalized stabilizer formalisms}

Our DCP stabilizer formalism can be seen as a special case of the \textit{monomial stabilizer formalism} introduced by Van den Nest \cite{nest2011monomial}, whose stabilizers are products of diagonal unitaries and permutation matrices. DCP operators have this form because Pauli $X$ operators are permutation matrices. In this framework, each stabilizer state is an equal-magnitude superposition of computational-basis states, connected by the stabilizers' permutation actions. For DCP states, these actions are binary translations, so the support is an affine subspace. Our phase-polynomial representation satisfies the monomial framework's sufficient conditions for efficient computational-basis sampling and approximate local expectation values. We additionally provide exact Pauli expectation values and efficient gate and compatible-measurement updates.

More directly, the XS formalism \cite{ni2015noncommuting} uses tensor products of $X$ and $\S$ operators with eighth-root global phases. XS states already admit affine-support, third-order phase-polynomial descriptions, but with additional constraints linking the polynomials to the support. By allowing $\CZ$ factors in stabilizers, DCP removes these constraints and captures all third-order phase-polynomial states.

The XP formalism \cite{webster2022xp} further generalizes XS to arbitrary phase precision. Its measurement analysis \cite[Section 7 and Appendix F]{webster2022xp} provides outcome probabilities for XP measurements and updates for diagonal Pauli measurements. For Pauli $X$ measurements, its relative-phase calculation corresponds to the discrete derivative of the phase polynomial in our formulation. However, it does not classify when non-diagonal measurements preserve the XP state class. We provide such a classification for DCP states, together with polynomial-time compatibility tests and post-measurement updates.

\vspace{\baselineskip}
\paragraph*{Note added.}
While completing this manuscript, we became aware of two independent related works on efficient simulation of noisy non-Clifford quantum error correction.
Takada, Bartlett, and Williamson \cite{takada2026exact} develop a Clifford-stabilizer method closely related to our tableau approach.
Here, we also describe the states using phase polynomials, establish the equivalence of the two representations, and classify the measurements preserving this state class.
Magdalena de la Fuente and Scruby \cite{fuente2026efficiently} instead use a spacetime framework to sample measurement outcomes and bound logical failure rates without explicitly tracking the logical state, whereas our method maintains the conditional state throughout the circuit.
Additionally, we provide an open-source implementation and extend our framework to universal simulation through DCP state decompositions (\Cref{sec:universal-simulation}).
Exploring the precise relationships and practical tradeoffs between these approaches is an interesting direction for future work.

\subsection{Reading guide}

We develop the simulation framework through two complementary descriptions of the same states and circuits.
We begin with third-order phase-polynomial states and compatible circuits in \Cref{sec:phase-polynomial-states}, using this representation to construct a simulator in \Cref{sec:simulator-phase-polynomial-representation}.
We then recast these states in the DCP stabilizer formalism in \Cref{sec:dcp-stabilizer-formalism}, leading to an alternative, tableau-based simulator in \Cref{sec:simulator-dcp-representation}.
With both descriptions in place, \Cref{sec:noisy-simulations} establishes their applicability to noisy quantum error-correction circuits, and \Cref{sec:benchmarking-results} benchmarks our implementation, \texttt{merlin}, against existing simulators.
Finally, \Cref{sec:universal-simulation} extends the approach beyond the efficiently simulable class through DCP decompositions, before we discuss further directions in \Cref{sec:discussion}.
Technical proofs and algorithmic details are collected in the appendices.

Readers primarily interested in understanding the simulator can focus on \Cref{sec:phase-polynomial-states,sec:simulator-phase-polynomial-representation} and defer the DCP formalism on a first reading. The stabilizer picture in \Cref{sec:dcp-stabilizer-formalism,sec:simulator-dcp-representation} provides further structural insights and an alternative way to understand the same states and circuits, as well as an alternative implementation of the simulator.

\section{Phase-polynomial states and circuits}
\label{sec:phase-polynomial-states}

Our simulation framework rests on a simple description of states prepared by Pauli, $\CNOT$, and $\T$ gates: a uniform superposition over an affine space, with phases specified by a third-order polynomial.
We first introduce the CNOT-dihedral group and its phase-polynomial description, then define the corresponding class of third-order phase-polynomial (3PP) states.
Finally, we introduce compatible measurements, which preserve this state class, and use them to specify the circuits considered throughout the paper.

\subsection{The CNOT-dihedral group and phase polynomials}

We consider circuits generated by Pauli, $\CNOT$, and $\T$ gates, whose unitary operations form the third-order CNOT-dihedral group \cite{cross2016scalable}.

\begin{definition}[Third-order CNOT-dihedral group]
    The third-order CNOT-dihedral group on $n$ qubits is
    \begin{align}
        G_3 = \langle X_j, \CNOT_{ij}, \T_j : i \neq j \rangle / \langle \omega_8 I \rangle,
    \end{align}
    where $\omega_8 = e^{i\pi/4}$ and $\T = \operatorname{diag}(1,\omega_8)$.
    The quotient identifies operators that differ by a power of the global phase $\omega_8$.
\end{definition}

Throughout this work, we refer to $G_3$ simply as the \emph{CNOT-dihedral group}.
The name comes from the fact that $X$ and $\T$ generate a dihedral group up to global phase, to which $\CNOT$ gates are added.
Since $Z=\T^4$ and $Y=iXZ$, this group contains all Pauli operators up to global phase.

CNOT-dihedral operators admit an efficient description through their action on computational-basis states.
The $X$ and $\CNOT$ gates implement invertible affine transformations of the bit strings, while $\T$ gates contribute phases that are powers of $\omega_8$.
Consequently, every $U\in G_3$ has a unitary representative satisfying
\begin{align}
    U\ket{x} = \omega_8^{p(x)}\ket{Ax\oplus b},
    \qquad x\in\mathbb{F}_2^n,
\end{align}
where $A\in\mathrm{GL}(n,\mathbb{F}_2)$, $b\in\mathbb{F}_2^n$, and $p:\mathbb{F}_2^n\to\mathbb{Z}_8$.
The phase functions $p$ arising in this way are precisely the third-order phase polynomials defined below \cite{amy2013meetinthemiddle,campbell2017unified,amy2019tcount,amy2018finite}.

\begin{definition}[Third-order phase polynomial]\label{def:third-order-phase-polynomial}
A function $p:\F_2^n\to\Z_8$ is a \emph{third-order phase polynomial} if it can be written as
\begin{align}
    p(x) ={}& d + \sum_{i=1}^n a_i x_i
        + 2 \sum_{1 \le i < j \le n} b_{ij} x_i x_j \notag \\
    &+ 4 \sum_{1 \le i < j < k \le n} c_{ijk} x_i x_j x_k,
\end{align}
with $d,a_i \in \Z_8$, $b_{ij} \in \Z_4$, and $c_{ijk} \in \Z_2$, and where arithmetic is modulo $8$.
\end{definition}

This expression is sometimes called a \emph{weighted polynomial}~\cite{campbell2017unified,vuillot2022quantum}, while the term ``phase polynomial'' is often reserved for expressions of the form
\begin{align} \label{eqn:phase-polynomial-def}
    p(x)=a_0+\sum_{u\ne0}\lambda_u\langle u,x\rangle\pmod8,
\end{align}
where the inner product is taken modulo $2$, and the rest is taken modulo $8$. We show in \Cref{thm:phase-polynomial-equivalence} that these two forms are equivalent, and we therefore use the terms interchangeably.

Within this class, diagonal Clifford operators correspond to a more restricted form.

\begin{definition}[Clifford phase polynomial]\label{def:clifford-phase-polynomial}
A function $d:\F_2^n\to\Z_8$ is a \emph{Clifford phase polynomial} if it has the form
\begin{align}
    d(x)=d_0+2\sum_i \alpha_i x_i
        +4\sum_{i<j}\beta_{ij}x_ix_j \pmod 8,
\end{align}
where $d_0\in\Z_8$, $\alpha_i\in\Z_4$, and $\beta_{ij}\in\F_2$.
\end{definition}

These are precisely the phase functions of diagonal Clifford operators, up to global phase. \Cref{thm:clifford-phase-polynomial} proves this characterization and gives the equivalent parity-sum form.

\subsection{Third-order phase-polynomial states}

The efficient representation of CNOT-dihedral operators extends to the states they prepare. Indeed, for $U\in G_3$ on $n=m+\ell$ qubits,
\begin{align}
    U\ket{+}^{\otimes m} \ket{0}^{\otimes \ell} = \frac{1}{\sqrt{2^m}}\sum_{x \in \F_2^m} \omega_8^{p((x,0))} \ket{A(x,0) \oplus b}
\end{align}
where $(x,0)$ is the concatenation of $x \in \F_2^m$ and $0 \in \F_2^\ell$. Writing $\tilde{p}(x):= p((x,0))$ and $\tilde{A}$ for the first $m$ columns of $A$, we can rewrite the above equation as
\begin{align}
    U\ket{+}^{\otimes m} \ket{0}^{\otimes \ell} = \frac{1}{\sqrt{2^m}}\sum_{x \in \F_2^m} \omega_8^{\tilde{p}(x)} \ket{\tilde{A}x \oplus b}.
\end{align}
Since $\tilde{A}$ has full column rank, choose a left inverse $L$ satisfying $L\tilde{A}=I_m$ and define $q(y):=\tilde p(L(y\oplus b))$ for every $y\in\F_2^n$. Then $q(\tilde{A}x\oplus b)=\tilde p(x)$, and $q$ is a third-order phase polynomial by \Cref{cor:phase-polynomial-affine-pullback}. Changing variables therefore gives
\begin{align}
    U\ket{+}^{\otimes m} \ket{0}^{\otimes \ell} = \frac{1}{\sqrt{|S|}}\sum_{x \in S} \omega_8^{q(x)} \ket{x},
\end{align}
where $S = b \oplus \Im(\tilde{A})$. This motivates the following state class.

\begin{definition}[3PP states]\label{def:3pp}
    A \emph{third-order phase-polynomial state} (3PP state) is a state that is, up to global phase, of the form
    \begin{align}
        \ket{\psi} = \frac{1}{\sqrt{|S|}} \sum_{x \in S} \omega_8^{q(x)} \ket{x},
    \end{align}
    where $S \subseteq \F_2^n$ is a nonempty affine subspace and $q:\F_2^n\to\Z_8$ is a third-order phase polynomial.
\end{definition}

We note that the polynomial associated to a 3PP state is not unique, since its value only matters on the affine support $S$. We present later in \cref{thm:efficient-simulation} an algorithm to efficiently determine whether two 3PP states, represented by their phase polynomial and support, are equal up to a global phase.

The preparation argument above also has a converse.

\begin{proposition}[Preparation of 3PP states]\label{prop:3pp-preparation}
    A state is a 3PP state if and only if, up to global phase, it can be prepared from $\ket{0}$ and $\ket{+}$ states by a circuit of Pauli, $\CNOT$, and $\T$ gates.
\end{proposition}
\begin{proof}
    The calculation above proves that every state prepared by such a circuit is a 3PP state. Conversely, let $\ket{\psi}$ have the form in \Cref{def:3pp}. A uniform superposition over the affine space $S$ can be prepared from $\ket{0}$ and $\ket{+}$ states using $X$ and $\CNOT$ gates: choose an affine parametrization $S=b\oplus\Im(B)$ with $B$ of full column rank, extend the columns of $B$ to an invertible binary matrix, and implement the resulting affine transformation. Writing $q$ in the weighted form above and applying
    \[
        \prod_i \T_i^{a_i}
        \prod_{i<j}\CS_{ij}^{b_{ij}}
        \prod_{i<j<k}\CCZ_{ijk}^{c_{ijk}}
    \]
    then imprints the phase $\omega_8^{q(x)}$ up to the global phase $\omega_8^d$.
    Since $\CS$ and $\CCZ$ can themselves be synthesized from $\T$ and $\CNOT$ gates, this gives the required preparation.
\end{proof}

Keeping track of the affine space $S$ and the $O(n^3)$ coefficients of $q$ is therefore enough to represent any 3PP state.

\subsection{Compatible measurements and 3PP circuits}\label{sec:compatible-measurements}

To extend the gate set to measurements, we require each measurement branch to remain within the 3PP state class.

\begin{definition}[Compatible measurement]\label{def:compatible-measurement}
    Given a 3PP state $\ket{\psi}$, we say that a projective measurement of a Pauli operator is \emph{compatible} with it if the post-measurement state is a 3PP state for every outcome with nonzero probability.
\end{definition}

The need for introducing the notion of compatibility comes from the observation that not every $X$ measurement is compatible, as shown in the following simple example.
\begin{example}\label{ex:incompatible-x-measurement}
    Consider the 3PP state $\ket{\CCZ}=\CCZ\ket{+}^{\otimes 3}$. For an outcome $\varepsilon\in\{+1,-1\}$ of measuring the first qubit in the $X$ basis, the unnormalized post-measurement vector is
    \begin{align}
        \ket{\widetilde\psi_\varepsilon}=\frac{1}{2}(1+\varepsilon X_1)\ket{\CCZ}.
    \end{align}
    The outcome $\varepsilon=-1$ leaves the 3PP state $\ket{\psi_-}\propto\ket{-}\otimes\ket{11}$. The outcome $\varepsilon=+1$, which also has nonzero probability, leaves
    \begin{align}
        \ket{\psi_+}\propto\ket{+}\otimes
        \bigl(\ket{00}+\ket{01}+\ket{10}\bigr).
    \end{align}
    Its computational-basis support has six elements, whereas every affine subspace over $\F_2$ has cardinality a power of two. Thus its support is not affine, so $\ket{\psi_+}$ is not a 3PP state and the measurement is incompatible.
\end{example}

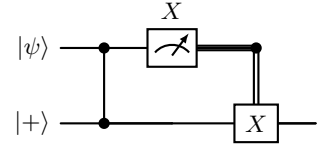
\begin{figure}[tbp]
    \centering
    \begin{quantikz}
        \lstick{$\ket{\psi}$} & \ctrl{1} & \meter{X} & \cwbend{1} \\
        \lstick{$\ket{+}$}    & \control{} & \qw & \gate{X} & \qw
    \end{quantikz}
    \caption{Hadamard gadget. A $\CZ$ gate with a $\ket{+}$ ancilla, followed by an $X$ measurement of the input qubit and an outcome-dependent $X$ correction, leaves the ancilla in the state $\Hgate\ket{\psi}$.}
    \label{fig:hadamard-gadget}
\end{figure}

The fact that our framework cannot simulate arbitrary $X$ measurements is not surprising.
Indeed, preparing a $|+\rangle$ ancilla state, applying $\CZ$, measuring the input qubit in the $X$ basis, and applying an outcome-dependent $X$ correction to the ancilla, implements a Hadamard gate (see \cref{fig:hadamard-gadget}).
Allowing arbitrary $X$ measurements therefore gives universal quantum computation. Consequently, a polynomial-time classical sampler for these circuits would imply $\mathrm{BQP}\subseteq\mathrm{BPP}$.

By contrast, as we discuss in \cref{sec:z-measurements-3pp}, a $Z$-type Pauli measurement automatically preserves the 3PP form: conditioning on a nonzero-probability outcome intersects the affine support with an affine hyperplane and leaves the phase polynomial unchanged.

Together, these closure properties specify the circuits considered in our simulation algorithms.

\begin{definition}[3PP circuit]\label{def:3pp-circuit}
A \emph{3PP circuit} is composed of $\ket{0}$ and $\ket{+}$ preparations, Pauli, $\CNOT$, and $\T$ gates, arbitrary $Z$-type Pauli measurements, and $X$-type Pauli measurements that are compatible on every state reachable with nonzero probability. The operations may depend on earlier measurement outcomes through classical feedforward.
\end{definition}

Thus every measurement branch of a 3PP circuit can be described using an affine support and a third-order phase polynomial.

\section{Simulation in the phase-polynomial representation}
\label{sec:simulator-phase-polynomial-representation}

We now describe how to update the phase-polynomial representation under the
operations in a 3PP circuit.  We represent the current state as
\begin{equation}
    \ket{\psi}
    = \frac{1}{\sqrt{|S|}}
      \sum_{x\in S}\omega_8^{q(x)}\ket{x},
    \qquad S=x_0\oplus V,
    \label{eq:phase-polynomial-state-update}
\end{equation}
where $S\subseteq\F_2^n$ is affine and $q$ is a third-order phase
polynomial.  For the initial state
$\ket{0}^{\otimes a}\ket{+}^{\otimes b}$, one takes
$S=\{0\}^a\times\F_2^b$ and $q=0$.
Thus, the representation can be initialized and all the updates below can be carried out in polynomial time by linear algebra over $\F_2$ and arithmetic on the coefficients of $q$.

\subsection{Unitary operations}

Every CNOT-dihedral unitary has an affine permutation and a third-order
phase polynomial $p$ such that
\begin{equation}
    U\ket{x}=\omega_8^{p(x)}\ket{\phi(x)},
    \qquad \phi(x)=Ax\oplus b,
    \label{eq:cnot-dihedral-action}
\end{equation}
where $A\in\F_2^{n\times n}$ is invertible.  Applying $U$ to
\cref{eq:phase-polynomial-state-update} gives the updated data
\begin{equation}
    S'=\phi(S),
    \qquad
    q'(y)=q\bigl(\phi^{-1}(y)\bigr)
           +p\bigl(\phi^{-1}(y)\bigr)\pmod 8.
    \label{eq:phase-polynomial-unitary-update}
\end{equation}
The set $S'$ is affine, and $q'$ is again a third-order phase polynomial by
\Cref{cor:phase-polynomial-affine-pullback}.  In particular, the elementary
updates are as follows, where
$\tau_{a\to b}(x):=x\oplus x_a e_b$:
\begin{align}
    \T_a^{\pm1}:&\quad
        S'=S,
        &q'(y)&=q(y)\pm y_a, \\
    Z_a:&\quad
        S'=S,
        &q'(y)&=q(y)+4y_a, \\
    X_a:&\quad
        S'=S\oplus e_a,
        &q'(y)&=q(y\oplus e_a), \\
    \CNOT_{a\to b}:&\quad
        S'=\tau_{a\to b}(S),
        &q'(y)&=q\bigl(\tau_{a\to b}(y)\bigr).
\end{align}
Updates for arbitrary Pauli operators follow by composing the $X_a$ and
$Z_a$ rules. Global phases may be discarded.

\subsection{Z-measurements}
\label{sec:z-measurements-3pp}

Let us consider a measurement of a $Z$-type Pauli operator $Z(v)=\prod_{a=1}^n Z_a^{v_a}$, and encode its measurement outcome
as $(-1)^\mu$, with $\mu\in\F_2$.
The projection onto that outcome can be written
\begin{align}
    \ket{\psi'_\mu} &= \frac{1}{2}\left( 1 + (-1)^\mu Z(v) \right)\ket{\psi} \\
    &= \frac{1}{2\sqrt{|S|}} \sum_{x\in S}
    \left( 1 + (-1)^{\mu+\langle v,x\rangle} \right)\ket{x} \\
\end{align}
where we used $Z(v)\ket{x}=(-1)^{\langle v,x\rangle}\ket{x}$.
Therefore, the phase is only nonzero when $\langle v,x\rangle=\mu$, and the affine support is restricted to that hyperplane:
\begin{equation} \label{eq:phase-polynomial-z-measurement-update}
    S'_\mu
    =S\cap\{x\in\F_2^n:\langle v,x\rangle=\mu\},
    \qquad
    q'_\mu=q\big|_{S'_\mu}.
\end{equation}
The probability of the outcome is
\begin{equation}
    \Pr(\mu)=\frac{|S'_\mu|}{|S|}.
\end{equation}
The linear function $x\mapsto\langle v,x\rangle$ is either constant on
$S$ or balanced on $S$.  Consequently, the two outcome probabilities are
either $0$ and $1$, or both $1/2$.  Gaussian elimination determines which
case applies and produces an affine description of $S'_\mu$ in polynomial
time.  A single-qubit $Z$-measurement is the special case $v=e_a$.

\subsection{X-measurements}

Recall from \Cref{sec:compatible-measurements} that an $X$-type Pauli measurement is compatible with a 3PP state when every outcome of nonzero probability leaves a 3PP state. We first introduce the discrete derivative and classify single-qubit $X$ measurements with their phase-polynomial updates, then extend the results to arbitrary $X$-type Pauli measurements.

\begin{definition}[Discrete derivative]\label{def:discrete-derivative}
    Given a function $f: V \to \mathbb Z_8$, we define its discrete derivative $\Delta_v f$ w.r.t. $v \in V$ by $\Delta_v f(x) := f(x \oplus v) - f(x) \pmod 8.$
\end{definition}

\newcommand{\xMeasurementClassificationStatement}{%
    Consider a 3PP state
    \[ \ket{\psi} = \frac{1}{\sqrt{|S|}} \sum_{x\in S} \omega_8^{q(x)}\ket{x} \]
    over an affine subspace $S = x_0 \oplus V$ with phase polynomial $q$.
    An $X$ measurement on qubit $a$ is compatible with $\ket{\psi}$ iff one of the following cases apply.
    In each case, we get a post-measurement affine subspace $S'_\mu \subseteq \mathbb F_2^n$ and phase polynomial $q'_\mu: S'_\mu \to \mathbb Z_8$ depending on the measurement outcome $\mu \in \{0, 1\}$, which occurs with some probability $Pr(\mu)$.
    \begin{enumerate}
        \item $e_a \not\in V$. In that case, $Pr(\mu) = \frac{1}{2}$ and the post-measurement state is defined by $S'_\mu = S \cup (S \oplus e_a)$ and the unique phase polynomial $q'_\mu$ with $q'_\mu(x) = q(x)$ and $q'_\mu(x \oplus e_a) = q(x) + 4\mu$ for all $x \in S$.

        \item
        $e_a \in V$ and one of the following apply:
        \begin{enumerate}
            \item \label{thm:x-measurements:even-0-4}
            $\Delta_{e_a} q(x) = 4f(x) \pmod 8$ for some affine function $f: S \to \mathbb F_2$. In that case
            \[
                Pr(\mu) = \begin{cases}
                    0 & \text{if } f(x) \neq \mu \text{ for all } x \in S, \\
                    1 & \text{if } f(x) = \mu \text{ for all } x \in S, \\
                    \frac{1}{2} & \text{otherwise},
                \end{cases}
            \]
            and the post-measurement state is defined by $S'_\mu = S \cap \{ x ~|~ f(x) = \mu \pmod 2 \}$ and $q'_\mu(x) = q(x)$ for all $x \in S'_\mu$.

            \item \label{thm:x-measurements:even-2-6}
            $\Delta_{e_a} q(x) = 2 + 4g(x) \pmod 8$ for some function $g: S \to \mathbb F_2$ such that $2g$ is a third-order phase polynomial.
            In that case, $Pr(\mu) = \frac{1}{2}$ and the post-measurement state is defined by $S'_\mu = S$ and $q'(x) = q(x) - 2(-1)^\mu g(x)$

            \item \label{thm:x-measurements:odd}
            $\Delta_{e_a} q(x) = (-1)^{h(x)}d \pmod 8$ for some affine function $h: S \to \mathbb F_2$ and $d \in \{ 1, 3, 5, 7 \}$.
            In that case, $Pr(\mu) = \frac{1}{2}(1 + (-1)^\mu\cos(d\pi/4))$ and the post-measurement state is defined by $S'_\mu = S$ and $q'(x) = q(x) + (4\mu - d)h(x)$.
        \end{enumerate}
    \end{enumerate}
}
\begin{theorem}\label{thm:x-measurements}
    \xMeasurementClassificationStatement
\end{theorem}
The proof can be found in \cref{appendix:x-measurements}.

For a general $X$-type Pauli operator $X(m)=\prod_j X_j^{m_j}$ with $m\in\F_2^n\setminus\{0\}$, choose $a$ with $m_a=1$ and set $C_m=\prod_{j\ne a:\,m_j=1}\CNOT_{a\to j}$. Then $C_mX(m)C_m^\dagger=X_a$, so
\begin{align}
    \frac{I+(-1)^\mu X(m)}{2}
    =C_m^\dagger\frac{I+(-1)^\mu X_a}{2}C_m.
    \label{eq:x-measurement-reduction}
\end{align}
Since both $C_m$ and $C_m^\dagger$ preserve 3PP states, measuring $X(m)$ on $\ket{\psi}$ has the same compatibility and outcome probabilities as measuring $X_a$ on $C_m\ket{\psi}$. Applying $C_m^\dagger$ to the conditional state gives the original measurement's conditional state. In the original coordinates, the same cases and update formulas apply with $e_a$ replaced by $m$ and $\Delta_{e_a}q$ replaced by $\Delta_mq$; in particular, case~(1) means $m\notin V$. We use these case labels for arbitrary $X$-type Pauli measurements throughout the paper. For $m=0$, the operator is the identity and the outcome is deterministically $+1$.

\subsection{Efficient simulation}

\newcommand{\efficientSimulationStatement}[1]{%
An $n$-qubit 3PP state admits a representation
\begin{align}
    \ket{\psi}
    = 2^{-r/2}\sum_{y\in\mathbb{F}_2^r}
    \omega_8^{p(y)}\ket{x_0\oplus By},
    #1
\end{align}
where $B\in\mathbb{F}_2^{n\times r}$ has full column rank and $p:\mathbb{F}_2^r\to\mathbb{Z}_8$ is a third-order phase polynomial.
This representation requires $O(n^3)$ bits and supports the following operations.

\begin{enumerate}
    \item \textbf{Circuit simulation.}
    Let $C$ be a 3PP circuit made of $m$ operations acting on at most $n$ qubits.
    There exists a classical algorithm that samples the measurement record according to its exact distribution and maintains the corresponding conditional state, up to global phase, in the form \eqref{eq:simulation-representation}.
    Each operation takes $O(n^3)$ expected time, giving a total expected runtime of
    \begin{align}
        O(mn^3)+T_{\mathrm{classical}},
    \end{align}
    where $T_{\mathrm{classical}}$ is the total time required for classical feedforward.
    The maintained state representation uses $O(n^3)$ bits, excluding storage for the measurement record and classical processing.

    \item \textbf{Queries on represented states.}
    Given a state in the form \eqref{eq:simulation-representation}, the following tasks can be performed:
    \begin{enumerate}
        \item Compute $\langle x|\psi \rangle$ exactly for any $x\in\mathbb{F}_2^n$ in $O(n^3)$ time.
        \item Sample $x\in\mathbb{F}_2^n$ according to $|\langle x|\psi \rangle|^2$ in $O(n^2)$ time.
        \item Compute $\bra{\psi}P\ket{\psi}$ exactly for any specified Pauli operator $P$ in $O(n^3)$ time.
        \item Given a second represented state on the same ordered qubits, decide whether the two states are equal up to global phase in $O(n^4)$ time.
    \end{enumerate}
\end{enumerate}
}

\begin{theorem}[Efficient simulation of 3PP circuits]
\label{thm:efficient-simulation}
\efficientSimulationStatement{\label{eq:simulation-representation}}
\end{theorem}
The proof can be found in \cref{appendix:efficient-simulation}.

The Pauli-expectation claim does not require the
corresponding measurement to be compatible: evaluating
an expectation value does not require representing the
post-measurement state. By linearity, it also gives
efficient exact expectations for observables supplied
as polynomial-size sums of Pauli operators.

Equality testing should be distinguished from evaluating
an arbitrary overlap between two DCP states. Such an
overlap generally involves a third-order phase sum,
rather than the quadratic sum appearing in
\cref{eq:pauli-expectation-gauss-sum}, and is not
covered by the theorem.

For circuits with stochastic Pauli noise, the theorem
applies to each sampled fault trajectory, provided that
all measurements remain compatible along that trajectory.
The represented state and its exact expectation values
are conditional on the sampled faults and measurement
outcomes.

Our simulation theorem recovers efficient simulation of Clifford circuits. Every Pauli stabilizer state is a 3PP state, and Pauli measurements preserve the stabilizer-state class. While $\CNOT$ and $\S=\T^2$ are already supported, a Hadamard gate can be implemented by the measurement-based gadget shown in \cref{fig:hadamard-gadget}.

For stabilizer inputs, every intermediate state remains a stabilizer state, so this measurement is compatible. Thus any Clifford circuit can be compiled into a 3PP circuit with constant overhead per gate.

Finally, we note that alternative sparse representations of the phase polynomial (for example in the form of Eq. \eqref{eqn:phase-polynomial-def}), can yield better complexity for circuits with few non-Clifford gates, which we have exploited in our implementation.

\section{The DCP stabilizer formalism}
\label{sec:dcp-stabilizer-formalism}

We now develop an equivalent stabilizer description of 3PP states, generalizing the Pauli stabilizer formalism by allowing products of diagonal Clifford and Pauli operators as stabilizers. They provide a Heisenberg-picture representation suited to tracking non-Clifford circuits. We establish the structure and canonical representation of those stabilizer groups, then prove that the states they uniquely specify are exactly the 3PP states.

\subsection{DCP operators and normal form}

The stabilizers corresponding to 3PP states belong to the following operator class.

\begin{definition}[DCP operator]\label{def:dcp}
    A \textit{DCP operator} is a unitary operator $U$ on $n$ qubits which can be expressed as a product of diagonal Clifford and Pauli operators. Explicitly, the group of DCP
    operators on $n$ qubits is
    \begin{align}
        \mathcal{D}^n
        = \langle \omega_8 I,\; X_i,\; \S_i,\; \CZ_{ij}
        \mid 1\leq i,j\leq n,\; i \neq j \rangle.
    \end{align}
\end{definition}

DCP operators can be efficiently represented using the following normal form.

\newcommand{\dcpNormalFormStatement}{%
Every DCP operator $U$ on $n$ qubits corresponds to a unique quadruple $(k,\vec{s},M,\vec{x})$ where $k \in \Z_8, \vec{s} \in \Z_4^{n}, \vec{x} \in \F_2^n$, and $M$ is an $n\times n$ symmetric matrix over $\F_2$ that is zero on the diagonal, such that
\begin{align}
    U = \omega_8^{k}\, D(\vec{s},M)\, X(\vec{x}),
\end{align}
where $D(\vec{s},M)$ is the diagonal Clifford operator defined by
\begin{align}
    D(\vec{s},M) = \S^{s_1} \otimes \cdots \otimes \S^{s_n} \; \cdot \; \prod_{1\le i < j \le n} \CZ_{ij}^{M_{ij}},
\end{align}
and $X(\vec{x})$ is the Pauli operator defined by
\begin{align}
    X(\vec{x}) = X^{x_1} \otimes \cdots \otimes X^{x_n}.
\end{align}
}

\begin{proposition}[Normal form for DCP operators]\label{prop:dcp-normal-form}
\dcpNormalFormStatement
\end{proposition}
We prove this proposition in \Cref{appendix:dcp-algebra}, where we also give explicit formulas to compute certain operations on DCP operators (multiplication, inverse, conjugation and commutator) using the normal form.

\subsection{Stabilizer groups and stabilizing closure}

We now use DCP operators to specify states by their common fixed space. Write
\begin{align}
    \Fix(U) = \{ v \in \C^{2^n} \mid Uv = v \}
\end{align}
for the space fixed by a unitary operator $U$, and
\begin{align}
    \Fix(S)=\bigcap_{U \in S} \Fix(U)
\end{align}
for the space fixed by a set of operators $S$.

\begin{definition}[Maximal DCP stabilizer group and DCP state]\label{def:dcp-state}
    Given a state $\ket{\psi}$ on $n$ qubits, its \textit{maximal DCP stabilizer group}, $\mathcal{D}(\psi)$, is the subgroup of $\mathcal{D}^n$ defined by
    \begin{align}
        \mathcal{D}(\psi) &= \{ g \in \mathcal{D}^n \mid \ket{\psi} \in \Fix(g) \}
    \end{align}
    We call a state $\ket{\psi}$ a \textit{DCP-stabilizer state} or for short a \textit{DCP state} if the space fixed by its DCP-stabilizer group is spanned by $\ket{\psi}$ itself, that is:
    \[
    \dim \Fix(\mathcal{D}(\psi)) = 1
    \]
    In other words, a state is a DCP state if it is determined up to global phase by the DCP operators that fix it.
\end{definition}

\begin{remark}
    This is analogous to the fact that a state is a (Pauli) stabilizer state if it is determined by its Pauli stabilizer group up to a global phase. Moreover, note that every stabilizer state is a DCP state, but the converse is not true: for example, the state $|T\rangle=\T\ket{+}$ is a DCP state but not a stabilizer state.
\end{remark}

In contrast to the Pauli case, DCP stabilizer groups need not be abelian. For instance, the state $\ket{0+}$ is stabilized by both $X_2$ and $\CZ_{12}$, but
\begin{align}
    \gcomm{X_2}{\CZ_{12}}=Z_1\neq I,
\end{align}
where $\gcomm{g}{h}:=ghg^{-1}h^{-1}$ denotes the group commutator.
More generally, the group commutator of two DCP stabilizers is characterized by the following proposition:

\begin{proposition}[Commutators of DCP stabilizers]
    Let $\ket{\psi}$ be a nonzero state and
    $g,h\in\mathcal{D}(\psi)$. Then
    \begin{align}
        \gcomm{g}{h}=\varepsilon Z(\vec{z})
        \in\mathcal{D}(\psi)
    \end{align}
    for some $\varepsilon\in\{+1,-1\}$ and
    $\vec{z}\in\F_2^n$.
\end{proposition}

\begin{proof}
    By \cref{prop:dcp-group-ops}, we have
    $\gcomm{g}{h}=i^c Z(\vec{z})$ for some $c\in\Z_4$
    and $\vec{z}\in\F_2^n$.
    Since $g$ and $h$ fix $\ket{\psi}$, so does their
    commutator. Therefore
    \begin{align}
        Z(\vec{z})\ket{\psi}=i^{-c}\ket{\psi}.
    \end{align}
    As $\ket{\psi}\neq 0$ and $Z(\vec{z})$ has
    eigenvalues in $\{+1,-1\}$, we obtain
    $i^c\in\{+1,-1\}$.
\end{proof}

Another difference from the Pauli case is that distinct DCP groups may fix the same space.
For example, the distinct groups
\begin{align}
    G=\langle Z_1,\S_2\rangle,
    \qquad
    H=\langle \S_1,Z_2\rangle
\end{align}
both have fixed space $\C\ket{00}$: the $+1$ eigenspace
of either $Z$ or $\S$ on a single qubit is $\C\ket{0}$.
Consequently, equality of fixed spaces does not imply
equality of groups. This motivates identifying groups
through their \emph{stabilizing closure}.

\begin{definition}[Stabilizing closure]\label{def:stabilizing-closure}
    Given a DCP subgroup $G \le \mathcal{D}^n$, its \emph{stabilizing closure} $\overline{G}$ is the subgroup of $\mathcal{D}^n$ defined by
    \begin{align}
        \overline{G} = \{ g \in \mathcal{D}^n \mid \Fix(G) \subseteq \Fix(g) \}.
    \end{align}
    That is, $\overline{G}$ is the largest subgroup of $\mathcal{D}^n$ that fixes the same space as $G$.
\end{definition}

The stabilizing closure therefore captures exactly the fixed space, as expressed by the following equivalence.

\begin{proposition}
    For any groups $G, H \le \mathcal{D}^n$, we have $\Fix(G) = \Fix(H)$ if and only if $\overline{G} = \overline{H}$.
\end{proposition}
\begin{proof}
    We first show that taking the stabilizing closure
    preserves the fixed space:
    \begin{align}
        \Fix(\overline{G})=\Fix(G).
    \end{align}
    Since $G\subseteq\overline{G}$, we have
    $\Fix(\overline{G})\subseteq\Fix(G)$.
    Conversely, every element of $\overline{G}$ fixes every vector in $\Fix(G)$ by definition, so $\Fix(G)\subseteq\Fix(\overline{G})$.

    If $\Fix(G)=\Fix(H)$, then their stabilizing closures
    coincide by definition.
    Conversely, if
    $\overline{G}=\overline{H}$, the identity above gives
    \begin{align}
        \Fix(G)=\Fix(\overline{G})
        =\Fix(\overline{H})=\Fix(H).
    \end{align}
\end{proof}

In our simulation procedures, any DCP group uniquely fixing the current state can serve as its representation.

\begin{definition}[DCP-stabilizer group]\label{def:dcp-stabilizer-group}
    Given a DCP state $\ket{\psi}$ and a group $G \le \mathcal{D}^n$, we say that $G$ is a \textit{DCP-stabilizer group} of $\ket{\psi}$ if $\Fix(G) = \C\ket{\psi}$, or equivalently if $\overline{G} = \mathcal{D}(\psi)$.
\end{definition}

\subsection{Canonical representation}

To manipulate DCP subgroups, compute their fixed spaces, and use them in a tableau-based simulation algorithm, we first need to develop a canonical representation: a generating set which uniquely determines a DCP subgroup.

We start by separating its purely diagonal elements from the ones that contain $X$ operators.
By \Cref{prop:dcp-group-ops}, the map $\pi:\mathcal{D}^n\to\F_2^n$ defined by
\begin{align}
    \pi\left(\omega_8^kD(\vec{s},M)X(\vec{x})\right)=\vec{x},
\end{align}
is a group homomorphism. For a subgroup
$G\leq\mathcal{D}^n$, define
\begin{align}
    \Delta(G):=\ker\pi|_G,
    \qquad
    X(G):=\pi(G).
\end{align}
Thus, $\Delta(G)$ is the subgroup of diagonal elements
of $G$, while $X(G)$ is a subspace of $\F_2^n$ representing the $X$ part of each element of $G$.
In other words, $G$ fits into the short exact sequence
\begin{equation}\label{eq:dcp-ses}
    1 \rightarrow \Delta(G) \rightarrow G \xrightarrow{\pi} X(G) \rightarrow 1,
\end{equation}
and the first isomorphism theorem gives
\begin{align}
    G/\Delta(G)\cong X(G).
\end{align}
More explicitly, if $h\in G$ has $\pi(h)=\vec{x}$, then
\begin{align}
    G\cap\pi^{-1}(\vec{x})
    =h\Delta(G)=\Delta(G)h.
\end{align}
This yields the following generating set for $G$.

\begin{proposition}[Generating sets for DCP subgroups]
    \label{prop:dcp-generators}
    Let $G\leq\mathcal{D}^n$, let
    $\vec{x}^{(1)},\ldots,\vec{x}^{(r)}$ be a basis of $X(G)$,
    and choose $h_j\in G$ with $\pi(h_j)=\vec{x}^{(j)}$.
    If $\delta_1,\ldots,\delta_m$ generate $\Delta(G)$,
    then
    \begin{align}
        G=\langle
            \delta_1,\ldots,\delta_m,h_1,\ldots,h_r
        \rangle.
    \end{align}
\end{proposition}

\begin{proof}
    We show that every $g\in G$ is generated by these
    elements. Write
    \begin{align}
        \pi(g)=\sum_{j=1}^r a_j\vec{x}^{(j)},
    \end{align}
    for $a_j\in\F_2$, and set $h=h_1^{a_1}\cdots h_r^{a_r}$.
    Since $\pi(h)=\pi(g)$, the element $gh^{-1}$ belongs to $\Delta(G)$. Thus, $g=(gh^{-1})h$ belongs to the claimed generated subgroup.
    The reverse inclusion holds because all the generators belong to $G$.
\end{proof}

We next choose canonical coordinates for the diagonal
subgroup. Set $d=1+n+\binom{n}{2}$ and define $\iota:\ker\pi\hookrightarrow\Z_8^d$ by
\begin{align}
    \iota \left( \omega_8^kD(\vec{s},M) \right) =
    \big(k,2\vec{s},(4M_{ab})_{a<b}\big).
\end{align}
We fix the coordinate order, taking the pairs $(a,b)$
in lexicographic order. By the normal form and product formula for DCP operators, $\iota$ is an injective group homomorphism into the additive group $\Z_8^d$.
Consequently, $\iota(\Delta(G))$ can be seen as a $\Z_8$-submodule of the module $\Z_8^d$.

\begin{lemma}\label{lem:diagonal-generators}
Given generators $g_1,\ldots,g_m$ of a DCP group $G\leq\mathcal D^n$, generators of $\Delta(G)$ and lifts of a basis of $X(G)$ can be computed in time polynomial in $n$ and $m$.
\end{lemma}
\begin{proof}
We first construct lifts of a basis of $X(G)$, then recover $\Delta(G)$ by closing a diagonal subgroup under conjugation.
Gaussian elimination on $\pi(g_1),\ldots,\pi(g_m)$, tracking the corresponding operator products, gives elements $h_1,\ldots,h_r\in G$ whose images form a basis of $X(G)$.
For each $a$, write
\begin{align}
    \pi(g_a)=\bigoplus_{i=1}^r c_{ai}\pi(h_i),
    \qquad
    \delta_a=g_a\left(h_1^{c_{a1}}\cdots h_r^{c_{ar}}\right)^{-1}.
\end{align}
Then each $\delta_a$ is diagonal and
\begin{align}
    G=\langle\delta_1,\ldots,\delta_m,h_1,\ldots,h_r\rangle.
\end{align}

To recover $\Delta(G)$, start with
\begin{align}
    N_0=\left\langle\delta_a,\ h_i^2,\ \gcomm{h_i}{h_j}\ :\ 1\leq a\leq m,\ 1\leq i,j\leq r\right\rangle.
\end{align}
All these generators are diagonal, since $\pi$ takes values in an elementary abelian $2$-group.
Given a generating set $\mathcal S_k$ of $N_k$, define
\begin{align}
    N_{k+1}=\left\langle N_k,\ h_i d h_i^{-1}\ :\ d\in\mathcal S_k,\ 1\leq i\leq r\right\rangle,
\end{align}
using Howell normal form in the coordinates $\iota$ to reduce the generating set after each step.
Every $N_k$ is contained in $\Delta(G)$.
Stop when $N_{k+1}=N_k$ and denote the resulting subgroup by $N$.

We now show that $N=\Delta(G)$.
At termination, $h_iNh_i^{-1}\subseteq N$ for every $i$; since $N$ is finite, these inclusions are equalities.
Moreover, the diagonal generators $\delta_a$ commute with $N$, so $N$ is normal in $G$.
The quotient $G/N$ is generated by the commuting elements $h_iN$, each of order at most two, and therefore has at most $2^r$ elements.
Since $N\subseteq\Delta(G)$, the map $\pi$ induces a surjection
\begin{align}
    G/N\twoheadrightarrow X(G).
\end{align}
The target has $2^r$ elements, so this map is an isomorphism.
Its kernel is $\Delta(G)/N$, proving that $N=\Delta(G)$.

Finally, we verify the polynomial runtime.
Set $d=1+n+\binom{n}{2}$.
Since $\iota(\Delta(G))\leq\mathbb Z_8^d$, every strict enlargement of $N_k$ at least doubles its order, and there can be at most $3d$ such enlargements.
Howell normal form supplies at most $d$ generators after each reduction, so each round involves polynomially many DCP operations and a polynomial-time Howell reduction.
The initial Gaussian elimination and construction of $N_0$ are also polynomial in $n$ and $m$, completing the proof.
\end{proof}

To specify $G$, we must record not only its diagonal subgroup $\Delta(G)$ and its space of $X$-parts $X(G)$, but also which diagonal factors accompany those $X$-parts.
For example, $\langle X\rangle$ and $\langle -X\rangle$ have the same diagonal subgroup $\{I\}$ and the same image $\F_2$ under $\pi$, but are distinct groups.

Choose a basis $\vec{x}^{(1)},\ldots,\vec{x}^{(r)}$ of $X(G)$.
For each basis vector, choose an operator
\begin{align}
    h_j=D_jX(\vec{x}^{(j)})\in G,
\end{align}
where $D_j$ is diagonal.
Such an operator is called a \emph{lift} of $\vec{x}^{(j)}$, because $\pi(h_j)=\vec{x}^{(j)}$.
By \Cref{prop:dcp-generators}, these operators, together with generators of $\Delta(G)$, determine $G$.

We now make these choices canonical.
First, choose the basis of $X(G)$ in row-reduced echelon form and represent $\iota(\Delta(G))$ in Howell normal form.
By \Cref{lem:diagonal-generators}, the required generators and lifts can be obtained from arbitrary generators of $G$, so both normal forms can be computed in polynomial time.
We recall the definition and relevant properties of Howell normal form in the \hyperref[app:howell]{Howell normal form subsection} of \Cref{appendix:dcp-algebra}.

It remains to choose each lift canonically.
If $h_j'=D_j'X(\vec{x}^{(j)})$ is another lift of the same basis vector, then
\begin{align}
    h_j'h_j^{-1}=D_j'D_j^{-1}\in\Delta(G).
\end{align}
Conversely, multiplying $h_j$ on the left by any $\delta\in\Delta(G)$ gives another lift of $\vec{x}^{(j)}$.
Thus the possible diagonal factors are exactly
\begin{align}
    \{\,\delta D_j:\delta\in\Delta(G)\,\}.
\end{align}
Since $\iota$ turns multiplication of diagonal operators into addition, their coordinate vectors form the additive coset
\begin{align}
    \iota(D_j)+\iota(\Delta(G)).
\end{align}
By \Cref{lem:howell-remainders}, reduction using the Howell normal form selects a unique representative of this coset.
It therefore selects one lift independently of the initial choice of $h_j$.

\begin{definition}[Canonical representation]
    \label{def:canon-form}
    Let $G\leq\mathcal{D}^n$.
    With the coordinate order fixed as above, the canonical representation of $G$ consists of:
    \begin{enumerate}
        \item the ordered basis $\vec{x}^{(1)},\ldots,\vec{x}^{(r)}$ of $X(G)$ in row-reduced echelon form;
        \item the Howell normal form $H$ of $\iota(\Delta(G))$;
        \item for each $j$, the lift $h_j=D_jX(\vec{x}^{(j)})\in G$ whose diagonal coordinate vector satisfies
        \begin{align}
            \iota(D_j)=\operatorname{rem}_H\big(\iota(D_j)\big).
        \end{align}
    \end{enumerate}
\end{definition}

\begin{proposition}[Uniqueness of the canonical representation]
    Two DCP subgroups $G,G'\leq\mathcal{D}^n$ have the same canonical representation if and only if $G=G'$.
\end{proposition}

\begin{proof}
    We first show that the representation depends only on the group.
    The row-reduced echelon basis of $X(G)$ and the Howell normal form of $\iota(\Delta(G))$ are uniquely determined.
    For each basis vector, the diagonal factors of its lifts form one additive coset modulo $\iota(\Delta(G))$, whose remainder is unique by \Cref{lem:howell-remainders}.
    Thus equal groups have equal canonical representations.

    Conversely, the representation determines generators of $\Delta(G)$ and lifts of a basis of $X(G)$.
    These generate $G$ by \Cref{prop:dcp-generators}, so equal representations determine equal groups.
\end{proof}

\subsection{Structure of DCP spaces}

For the rest of the section, we'll analyze the structure of the space fixed by DCP groups, and work towards describing an explicit basis of it in \Cref{thm:atomic-basis} and a characterization of its dimension in \cref{cor:dimension-dcp-space}. Throughout, let $G \le \mathcal{D}^n$ be a DCP subgroup, let $r = \dim X(G)$, and let
\begin{align}
    \Sigma_{\Delta(G)} := \{ \vec{v} \in \F_2^n \mid \delta\ket{\vec{v}} = \ket{\vec{v}} \ \ \forall \delta \in \Delta(G) \}
\end{align}
denote the set of computational basis states fixed by every element of $\Delta(G)$. Equivalently, it is an algebraic set defined as the common zero set of the phase polynomials of the diagonal operators in $\Delta(G)$.
Indeed, writing $\delta\ket{\vec{v}}=\omega_8^{p_\delta(\vec{v})}\ket{\vec{v}}$, we have
\begin{align}
    \Sigma_{\Delta(G)}
    =\{\vec{v}\in\F_2^n:p_\delta(\vec{v})=0 \ \ \forall \delta\in\Delta(G)\}
\end{align}
where the equations are taken modulo $8$.
It suffices to impose these equations for a generating set of $\Delta(G)$, since multiplication of diagonal operators adds their phase polynomials modulo $8$.
Thus $\Sigma_{\Delta(G)}$ is an algebraic set of binary solutions to quadratic congruences.

\begin{example}\label{ex:dcp-atomic-basis}
    As a running example, consider $G=\langle A,B,C\rangle$ with
    \begin{align}
        A=\omega_8^{-1}\S_1X_1,
        \quad B=\CZ_{12},
        \quad C=\CZ_{34}.
    \end{align}
    The generators square to the identity, $C$ commutes with $A$ and $B$, and $\gcomm{A}{B}=Z_2$, which commutes with all three generators.
    Hence every element has the form $A^aZ_2^bB^cC^d$ with $a,b,c,d\in\F_2$, and
    \begin{align}
        X(G)=\langle e_1\rangle,
        \qquad
        \Delta(G)=\langle Z_2,\CZ_{12},\CZ_{34}\rangle.
    \end{align}
    The phase polynomials of these diagonal generators are $4v_2$, $4v_1v_2$, and $4v_3v_4$, respectively.
    Their common zeros therefore satisfy $v_2=0$ and $v_3v_4=0$, with the constraint $v_1v_2=0$ redundant.
    Consequently,
    \begin{align}
        \Sigma_{\Delta(G)}
        &=\{(t,0,a,b):t,a,b\in\F_2,\ ab=0\}\\
        &=\{0000,1000,0001,1001,0010,1010\}.
    \end{align}
    Note that this set can be written as a disjoint union of three cosets of $X(G)$:
    \begin{align}
        \Sigma_{\Delta(G)}
        =X(G)\sqcup\big(e_4+X(G)\big)\sqcup\big(e_3+X(G)\big).
    \end{align}
    The goal of this section is to show that this is generally the case, and that the space $\Fix(G)$ has a basis such that each basis vector is supported on one of the cosets.
\end{example}

\begin{lemma}\label{lem:supp-partition}
    Let $\ket{\psi} \in \Fix(G)$. Then $\supp(\psi)$ is a disjoint union of cosets of $X(G)$, and the amplitudes have equal magnitude within each coset. Equivalently, $\ket{\psi}$ can be written as
    \begin{align}
        \ket{\psi} = \sum_{j=1}^k a_j \sum_{v \in X(G)} b_{j,v} \ket{x_j + v}
    \end{align}
    for some $k \geq 1$, $a_j \in \C \setminus \{0\}$, representatives $x_j \in \F_2^n$ of distinct cosets, and $b_{j,v} \in \C$ with $|b_{j,v}|=1$.
\end{lemma}
\begin{proof}
    Write $\ket{\psi} = \sum \alpha_x \ket{x}$. We need to show that if $\alpha_x \neq 0$, then $\alpha_{x+v} \neq 0$ for all $v \in X(G)$. Let $x \in \supp(\psi)$, $v \in X(G)$ and a corresponding stabilizer $g_v = \delta_v X(v) \in G$ where $\delta_v$ is a diagonal Clifford operator. Writing $g_v \ket{x} = \lambda_v(x) \ket{x+v}$ for some $\lambda_v(x) \in \C$, we get
    \begin{align}
        g_v \ket{\psi} = \sum_{x} \alpha_x g_v \ket{x} = \sum_{x} \alpha_x \lambda_v(x) \ket{x+v}.
    \end{align}
    Since $g_v \ket{\psi} = \ket{\psi}$, we can compare the coefficients, getting
    \begin{align}
        \alpha_x \lambda_v(x) = \alpha_{x+v}.
    \end{align}
    Since $g_v$ is unitary, $|\lambda_v(x)|=1$, so taking absolute values gives
    \[
        |\alpha_{x+v}|=|\alpha_x|.
    \]
    The coefficient relation also gives
    \begin{align}
        g_v \ket{x} = (\alpha_{x+v}/\alpha_x) \ket{x+v}.
    \end{align}
    Thus $\alpha_{x+v} \neq 0$ and $x+v \in \supp(\psi)$. As $v$ was arbitrary, the whole coset $x + X(G)$ is contained in the support of $\ket{\psi}$, with equal amplitude magnitudes throughout. Choosing one representative $x_j$ per support coset and setting $a_j=\alpha_{x_j}$ and $b_{j,v}=\alpha_{x_j+v}/\alpha_{x_j}$ gives the claimed form with $|b_{j,v}|=1$.
\end{proof}

\begin{lemma}\label{lem:supp-in-sigma}
    Let $\ket{\psi} \in \Fix(G)$. Then $\supp(\psi) \subseteq \Sigma_{\Delta(G)}$.
\end{lemma}
\begin{proof}
    Let $\ket{\psi} = \sum \alpha_x \ket{x}$, and $\delta \in \Delta(G)$. We write $\delta \ket{x} = \lambda_x \ket{x}$ for some $\lambda_x \in \C \setminus \{0\}$. Then $\delta \ket{\psi} = \ket{\psi}$ implies
    \begin{align}
        \sum_x \alpha_x \lambda_x \ket{x} = \sum_x \alpha_x \ket{x},
    \end{align}
    which implies that $\lambda_x = 1$ for all $x$ with $\alpha_x \neq 0$. Therefore, $\supp(\psi) \subseteq \Sigma_{\Delta(G)}$.
\end{proof}

The smallest possible nonempty support allowed by these constraints is a single translation coset.

\begin{definition}[Atomic state]\label{def:atomic}
    A state $\ket{\psi} \in \Fix(G)$ is said to be \emph{atomic} for $G$ if $\supp(\psi)$ is a single coset of $X(G)$, that is, if there exists $x \in \F_2^n$ such that
    \begin{align}
        \ket{\psi} = \sum_{v \in X(G)} \alpha_v \ket{x+v}
    \end{align}
    for some $\alpha_v \in \C \setminus \{0\}$.
\end{definition}

Atomic states will provide a basis of $\Fix(G)$. Let us first establish uniqueness on each coset.

\begin{lemma}\label{lem:atomic-overlap}
    If $\ket{\psi_1}, \ket{\psi_2} \in \Fix(G)$ are atomic and $\supp(\psi_1) \cap \supp(\psi_2) \neq \emptyset$, then $\ket{\psi_2} = \lambda \ket{\psi_1}$ for some scalar $\lambda$ with $|\lambda|=1$.
\end{lemma}
\begin{proof}
    Two cosets of $X(G)$ that intersect must coincide, so there is some $x_0 \in \F_2^n$ with $\supp(\psi_1) = \supp(\psi_2) = x_0 + X(G)$. Write $\ket{\psi_1} = \sum_{x} \alpha^{(1)}_{x} \ket{x}$ and $\ket{\psi_2} = \sum_{x} \alpha^{(2)}_{x} \ket{x}$, and let $\lambda = \alpha^{(2)}_{x_0} / \alpha^{(1)}_{x_0}$. For any $y \in x_0 + X(G)$, pick a stabilizer $g \in G$ with $\pi(g) = y + x_0$. As in the proof of \Cref{lem:supp-partition}, we have $g\ket{x_0} = (\alpha^{(i)}_{y}/\alpha^{(i)}_{x_0})\ket{y}$ for both $i = 1, 2$. Since the left-hand side does not depend on $i$, it follows that $\alpha^{(2)}_{y} = \lambda\, \alpha^{(1)}_{y}$ for all $y \in x_0 + X(G)$, that is, $\ket{\psi_2} = \lambda \ket{\psi_1}$.
    Since both states have unit norm, $|\lambda|=1$.
\end{proof}

\begin{lemma}\label{lem:atomic-existence}
    For every $x\in\Sigma_{\Delta(G)}$, there exists an atomic state $\ket{\psi_x}\in\Fix(G)$ with support $x+X(G)$.
    This state is unique up to global phase among atomic states in $\Fix(G)$ whose support contains $x$.
\end{lemma}
\begin{proof}
    Let $x \in \Sigma_{\Delta(G)}$. Let us first construct a state with support $x+X(G)$, then show that it is fixed by $G$ and establish uniqueness.

    For each $v\in X(G)$, choose $g_v\in G$ with $\pi(g_v)=v$, taking $g_0=I$, and define
    \begin{align} \label{eq:atomic-state}
        \ket{\psi_x}
        :=\frac{1}{\sqrt{|X(G)|}}
        \sum_{v\in X(G)}g_v\ket{x}.
    \end{align}
    Each $g_v\ket{x}$ is a multiple of $\ket{x+v}$ with modulus $1$.
    These computational-basis vectors are distinct as $v$ varies, so $\ket{\psi_x}$ is normalized and has support exactly $x+X(G)$.

    To show that $\ket{\psi_x}\in\Fix(G)$, let $h\in G$ and write $w=\pi(h)$.
    For each $v\in X(G)$, the operators $hg_v$ and $g_{w+v}$ have the same $X$-part, since $\pi(hg_v) = \pi(h) + \pi(g_v) = w+v$. So
    \begin{align}
        \delta_{h,v}:=g_{w+v}^{-1}hg_v\in\Delta(G).
    \end{align}
    Since $x\in\Sigma_{\Delta(G)}$, we have $\delta_{h,v}\ket{x}=\ket{x}$, and hence
    \begin{align}
        hg_v\ket{x}
        =g_{w+v}\delta_{h,v}\ket{x}
        =g_{w+v}\ket{x}.
    \end{align}
    Summing over $v$ and using that $v\mapsto w+v$ permutes $X(G)$ gives
    \begin{align}
        h\ket{\psi_x}
        &=\frac{1}{\sqrt{|X(G)|}}
          \sum_{v\in X(G)}g_{w+v}\ket{x}
        =\ket{\psi_x}.
    \end{align}
    Thus $\ket{\psi_x}\in\Fix(G)$, and its support is a single coset of $X(G)$, so it is atomic.

    Finally, any other atomic state in $\Fix(G)$ whose support contains $x$ has support intersecting that of $\ket{\psi_x}$.
    By \Cref{lem:atomic-overlap}, the two states differ only by a global phase, proving uniqueness.
\end{proof}

\begin{corollary}\label{cor:sigma-cosets}
    The set $\Sigma_{\Delta(G)}$ is a disjoint union of cosets of $X(G)$.
\end{corollary}
\begin{proof}
    We show that every coset of $X(G)$ intersecting $\Sigma_{\Delta(G)}$ is contained in it.
    Let $x\in\Sigma_{\Delta(G)}$.
    By \Cref{lem:atomic-existence}, there exists a state $\ket{\psi_x}\in\Fix(G)$ with support $x+X(G)$.
    Since every element of $\Delta(G)$ is diagonal and fixes $\ket{\psi_x}$, it fixes each computational-basis vector in its support.
    Therefore
    \begin{align}
        x+X(G)=\supp(\psi_x)\subseteq\Sigma_{\Delta(G)}.
    \end{align}
    Thus $\Sigma_{\Delta(G)}$ is a union of cosets of $X(G)$, and distinct cosets are disjoint.
\end{proof}

\begin{theorem}[Atomic basis]\label{thm:atomic-basis}
    Let $C_1, \dots, C_m$ be the distinct cosets of $X(G)$ contained in $\Sigma_{\Delta(G)}$, and for each $j$ pick a point $\vec{x}^{(j)} \in C_j$. Then the atomic states $\ket{\psi_{\vec{x}^{(1)}}}, \dots, \ket{\psi_{\vec{x}^{(m)}}}$ defined in \cref{eq:atomic-state} form a basis of $\Fix(G)$.
\end{theorem}
\begin{proof}
    The states $\ket{\psi_{\vec{x}^{(j)}}}$ lie in $\Fix(G)$ by \Cref{lem:atomic-existence} and are orthonormal because their supports $C_1,\ldots,C_m$ are pairwise disjoint.
    It remains to show that they span $\Fix(G)$.

    Let $\ket{\psi}\in\Fix(G)$.
    By \Cref{lem:supp-in-sigma}, $\supp(\psi)\subseteq\Sigma_{\Delta(G)}$.
    Grouping the computational-basis amplitudes according to the partition $\Sigma_{\Delta(G)}=C_1\sqcup\cdots\sqcup C_m$ of \cref{cor:sigma-cosets}, we may therefore write
    \begin{align}
        \ket{\psi}=\sum_{j=1}^m a_j\ket{\psi_j},
    \end{align}
    where $a_j\in\C$ and each $\ket{\psi_j}$ is a state supported within $C_j$. We show that $\ket{\psi_j}$ is atomic whenever $a_j\neq0$.

    Every $g\in G$ preserves each subspace $\Span\{\ket{\vec{y}}:\vec{y}\in C_j\}$, since $g\ket{\vec{y}}$ is proportional to $\ket{\vec{y}\oplus\pi(g)}$ and $\pi(g)\in X(G)$.
    Comparing the components in these mutually orthogonal subspaces on both sides of $g\ket{\psi}=\ket{\psi}$ gives
    \begin{align}
        a_jg\ket{\psi_j}=a_j\ket{\psi_j}.
    \end{align}
    Thus, whenever $a_j\neq0$, we have $\ket{\psi_j}\in\Fix(G)$.
    By \Cref{lem:supp-partition}, its support is a nonempty union of $X(G)$-cosets contained in $C_j$, and therefore equals $C_j$.
    Hence $\ket{\psi_j}$ is atomic and, by \Cref{lem:atomic-overlap}, proportional to $\ket{\psi_{\vec{x}^{(j)}}}$.
    Consequently,
    \begin{align}
        \ket{\psi}\in\Span\{\ket{\psi_{\vec{x}^{(1)}}},\ldots,\ket{\psi_{\vec{x}^{(m)}}}\},
    \end{align}
    proving the spanning claim.
\end{proof}

\begin{corollary}\label{cor:dimension-dcp-space}
    The dimension of the space fixed by $G$ is given by
    \begin{align}
        \dim \Fix(G) \;=\; \frac{|\Sigma_{\Delta(G)}|}{2^{r}}
    \end{align}
    In particular $|\Sigma_{\Delta(G)}|$ is a multiple of $2^r$.
\end{corollary}
\begin{proof}
    By \Cref{cor:sigma-cosets}, $\Sigma_{\Delta(G)}$ is a disjoint union of exactly $m = |\Sigma_{\Delta(G)}| / 2^r$ cosets of $X(G)$, and by \Cref{thm:atomic-basis} the corresponding atomic states form a basis of $\Fix(G)$, so $\dim \Fix(G) = m$.
\end{proof}

\begin{remark}
    Since $\Sigma_{\Delta(G)}$ is an algebraic subset of $\F_2^n$ and not necessarily an affine subset, its cardinality need not be a power of two. This translates to the fact that the dimension of a space stabilized by a DCP group also need not be a power of two.
\end{remark}

\begin{example}
    For the group $G$ of \Cref{ex:dcp-atomic-basis}, the three cosets of $X(G)=\langle e_1\rangle$ contained in $\Sigma_{\Delta(G)}$ have representatives
    \begin{align}
        \begin{split}
            \vec{x}^{(00)}=(0,0,0,0),\\
            \vec{x}^{(01)}=(0,0,0,1),\\
            \vec{x}^{(10)}=(0,0,1,0).
        \end{split}
    \end{align}
    Choose the lifts $g_0=I$ and $g_{e_1}=A=\omega_8^{-1}\S_1X_1$.
    For each $(a,b)\in\{(0,0),(0,1),(1,0)\}$, the construction in \Cref{lem:atomic-existence} gives
    \begin{align}
        \ket{\psi_{ab}}
        &=\frac{1}{\sqrt{|X(G)|}}
          \sum_{v\in X(G)}g_v\ket{\vec{x}^{(ab)}}\\
        &=\frac{1}{\sqrt{2}}(I+A)\ket{00ab}\\
        &=\frac{1}{\sqrt{2}}
          \big(\ket{00ab}+\omega_8\ket{10ab}\big)\\
        &=\ket{T}\otimes\ket{0}\otimes\ket{ab},
    \end{align}
    where $\ket{T}=(\ket{0}+\omega_8\ket{1})/\sqrt{2}$.
    Here we used $A\ket{00ab}=\omega_8^{-1}i\ket{10ab}=\omega_8\ket{10ab}$.
    By \Cref{thm:atomic-basis}, the three states
    \begin{align}
        \ket{T}\otimes\ket{000},
        \qquad
        \ket{T}\otimes\ket{001},
        \qquad
        \ket{T}\otimes\ket{010}
    \end{align}
    form an orthonormal basis of $\Fix(G)$.
    The dimension of $\Fix(G)$ is therefore 3, which we could also have derived from \cref{cor:dimension-dcp-space}:
    \begin{align}
        \dim\Fix(G)
        =\frac{|\Sigma_{\Delta(G)}|}{2^{r}}
        =\frac{6}{2^1}=3.
    \end{align}
\end{example}

\subsection{Equivalence with 3PP states}

We first record the closure property that connects CNOT-dihedral preparation to DCP stabilizers.

\begin{proposition}[Stability of $\mathcal{D}^n$ under conjugation]\label{prop:dcp-closed}
    The group $\mathcal{D}^n$ of DCP operators is stable under conjugation by CNOT-dihedral gates. Moreover, given a DCP operator in normal form, its conjugate by an elementary Pauli, $\CNOT$, or $\T$ gate can be computed in $O(n^2)$ time.
\end{proposition}
\begin{proof}
Recall that
\begin{align}
    \mathcal{D}^n=\left\langle
        \omega_8 I,\ X_i,\ \S_i,\ \CZ_{ij}
    \right\rangle.
\end{align}
Since conjugation respects products, it is enough to check the conjugation of these generators.

Every Pauli operator belongs to $\mathcal{D}^n$, so conjugation by
a Pauli preserves $\mathcal{D}^n$ by the group property.

Next, let $U=\CNOT_{c,t}$. Conjugation by $U$ maps products of $X$ operators to products of $X$ operators. Moreover, it maps diagonal Clifford operators to diagonal Clifford operators.
Indeed, conjugation by a Clifford preserves the Clifford group, while conjugation by a computational-basis permutation preserves diagonality.

Finally, let $U=\T_i$. Since $\T_i$ is diagonal, it commutes with every $\S_j$ and $\CZ_{jk}$, and it commutes with $X_j$ for $j\neq i$. The remaining generator satisfies
\begin{align}
    \T_iX_i\T_i^\dagger=\omega_8^{-1}\S_iX_i\in \mathcal{D}^n.
\end{align}
This proves that $UgU^\dagger\in \mathcal{D}^n$ for every $g\in \mathcal{D}^n$
and every CNOT-dihedral gate $U$. The same argument applies to their inverses, using
\begin{align}
    \T_i^\dagger X_i\T_i=\omega_8 \S_i^\dagger X_i,
\end{align}
so in fact $U\mathcal{D}^nU^\dagger=\mathcal{D}^n$.

To establish the complexity bound, represent $g$ by its normal-form coefficients $(k,s,M,x)$, requiring $O(n^2)$ entries.
Conjugation by $\T_i$ changes only the global phase and the coefficient of $\S_i$, depending on $x_i$, and therefore takes $O(1)$ time.
Conjugation by a single-qubit Pauli changes at most $O(n)$ coefficients: conjugation by $Z_i$ changes only the global phase, while conjugation by $X_i$ updates the diagonal phase by substituting $y_i\mapsto y_i\oplus1$.
Conjugation by $\CNOT_{c,t}$ updates the $X$ component by a binary linear transformation and the diagonal phase by substituting $y_t\mapsto y_t\oplus y_c$.
Only the $O(n)$ terms containing $y_t$ are affected, each producing a bounded number of terms, so this update also takes $O(n)$ time.
An arbitrary Pauli operator can be handled using at most $n$ single-qubit Pauli conjugations, and each update takes $O(n)$ time, giving an $O(n^2)$ bound.
\end{proof}

\begin{theorem}\label{thm:3pp-dcp}
    A state $\ket{\psi}$ is a DCP state if and only if it is a 3PP state.
\end{theorem}

\begin{proof}
We first show that every 3PP state is a DCP state. By \Cref{prop:3pp-preparation}, it can be prepared from $\ket{0}$ and $\ket{+}$ states using only Pauli, $\CNOT$, and $\T$ gates. The initial product state is uniquely stabilized by a Pauli subgroup of $\mathcal D^n$. If a DCP group $G$ uniquely stabilizes $\ket{\psi}$, then $UGU^\dagger$ uniquely stabilizes $U\ket{\psi}$. By \Cref{prop:dcp-closed}, this conjugated group remains a DCP group for every gate in the preparation. Hence, the prepared 3PP state is a DCP state.

For the converse direction, let $\ket{\psi}$ be a DCP state and let $G = \mathcal{D}(\psi)$. Since $\dim \Fix(G) = 1$, \Cref{thm:atomic-basis} tells us that $\ket{\psi}$ is, up to global phase, an atomic state of the form
\begin{align}
    \ket{\psi} = 2^{-r/2}\sum_{v \in X(G)} g_{v} \ket{x_0}.
\end{align}
We now show that the phases of this atomic state are given by a third-order phase polynomial.
Choose a basis $v^{(1)},\ldots,v^{(r)}$ of $X(G)$ and lifts $g_1,\ldots,g_r\in G$ with $\pi(g_i)=v^{(i)}$.
Write
\begin{align}
    g_i\ket{x}=\omega_8^{d_i(x)}\ket{x\oplus v^{(i)}},
\end{align}
where each $d_i$ is a Clifford phase polynomial.
For $t\in\mathbb F_2^r$, choose
\begin{align}
    g_{v(t)}=g_r^{t_r}\cdots g_1^{t_1},
    \qquad
    v(t)=\bigoplus_{i=1}^r t_i v^{(i)}.
\end{align}
Applying these operators successively gives
\begin{align}
    g_{v(t)}\ket{x_0}
    &=\omega_8^{P(t)}\ket{x_0\oplus v(t)},\\
    P(t)
    &=\sum_{i=1}^r t_i\,d_i\left(x_0\oplus\bigoplus_{j<i}t_jv^{(j)}\right)\pmod 8.
\end{align}
Affine substitution preserves Clifford phase polynomials by \cref{cor:phase-polynomial-affine-pullback}, so each substituted $d_i$ has the form
\begin{align}
    c_i+2\sum_{j<i}a_{ij}t_j+4\sum_{j<k<i}b_{ijk}t_jt_k.
\end{align}
Consequently,
\begin{align}
    P(t)=\sum_i c_it_i+2\sum_{j<i}a_{ij}t_it_j+4\sum_{j<k<i}b_{ijk}t_it_jt_k,
\end{align}
which is a third-order phase polynomial.

Finally, let $B$ be the matrix with columns $v^{(1)},\ldots,v^{(r)}$ and choose a left inverse $L$ of $B$ over $\mathbb F_2$.
Then
\begin{align}
    q(x)=P\bigl(L(x\oplus x_0)\bigr)
\end{align}
is a third-order phase polynomial by affine pullback, and $q(x_0\oplus Bt)=P(t)$.
Hence, up to global phase,
\begin{align}
    \ket{\psi}=2^{-r/2}\sum_{x\in x_0\oplus X(G)}\omega_8^{q(x)}\ket{x},
\end{align}
so $\ket{\psi}$ is a 3PP state.
\end{proof}

\begin{remark}
Although the full DCP stabilizer group of a state can be non-abelian, every DCP state is uniquely stabilized by an abelian DCP subgroup generated by $n$ independent commuting involutions.
Indeed, by \Cref{prop:3pp-preparation,thm:3pp-dcp}, we may write $\ket{\psi}=U\ket{\phi}$, where $\ket{\phi}$ is a product of $\ket{0}$ and $\ket{+}$ states and $U$ is a CNOT-dihedral unitary.
If $P_1,\ldots,P_n$ are the standard Pauli stabilizer generators of $\ket{\phi}$, then
\begin{align}
    H=\langle UP_1U^\dagger,\ldots,UP_nU^\dagger\rangle
\end{align}
is a DCP subgroup made of commuting order-2 elements, with $\Fix(H)=\operatorname{span}\{\ket{\psi}\}$.
\end{remark}

\section{Simulation in the DCP representation}\label{sec:simulator-dcp-representation}
In this section, we describe how to update the DCP-stabilizer group of a state under the operations of preparing $|0\rangle$ and $|+\rangle$ states, applying Pauli, $\CNOT$ and $\T$ gates, and measuring $Z$ and certain $X$ Pauli operators projectively. We show that each of these operations can be simulated in polynomial time using the canonical representation of the DCP-stabilizer group.

First, note that any product state $\ket{0}^{\otimes a} \otimes \ket{+}^{\otimes b}$ is a DCP state, since it is stabilized by the DCP group generated by $\langle Z_1,\dots,Z_a, X_{a+1},\dots,X_{a+b} \rangle$. Thus, we can initialize the DCP-stabilizer group of any product state in polynomial time. The other operations are more complicated and described in the following subsections.

\subsection{Unitary operations}
Let's consider the following very general fact about describing states in the Heisenberg picture, that is, by considering a set of matrices that stabilize the state:
\begin{proposition}\label{prop:heisenberg}
    Suppose that we have a state $\ket{\psi}$ and a set of matrices $S$ that stabilize it. Consider a unitary operation $U$ that is applied to the state. Then, the set of matrices $U S U^\dagger$ stabilizes the state $U\ket{\psi}$, where
    \[
        U S U^\dagger = \{ U g U^\dagger \mid g \in S \}.
    \]
    Moreover, $\dim \Fix(S) = \dim \Fix(U S U^\dagger)$.
\end{proposition}
\begin{proof}
    We have that for any $L \in S$,
    \[
        U L U^\dagger (U \ket{\psi}) = U L \ket{\psi} = U \ket{\psi},
    \]
    so $U S U^\dagger$ stabilizes $U\ket{\psi}$.

    For the second part, note that if $\{v^{(1)}, \dots, v^{(r)}\}$ is a basis for $\Fix(S)$, then $\{U v^{(1)}, \dots, U v^{(r)}\}$ remains linearly independent and lies in $\Fix(U S U^\dagger)$, so $\dim \Fix(S) \le \dim \Fix(U S U^\dagger)$. The same argument applies to $U^\dagger$, so we also have $\dim \Fix(U S U^\dagger) \le \dim \Fix(S)$, and thus $\dim \Fix(S) = \dim \Fix(U S U^\dagger)$.
\end{proof}

By the conjugation-closure result \Cref{prop:dcp-closed}, we can apply \Cref{prop:heisenberg} to update the DCP-stabilizer group directly by conjugating its generators:

\begin{proposition}[Updating under Pauli, $\CNOT$, and $\T$ gates]\label{prop:dcp-update}
    Let $\ket{\psi}$ be a state with DCP-stabilizer group $\mathcal{D}(\psi) = \langle g_1, \dots, g_m \rangle$, and let $U$ be a Pauli operator, a $\CNOT$ gate, or a $\T$ (or $\T^\dagger$) gate. Then $Ug_1U^\dagger, \dots, Ug_mU^\dagger$ are DCP operators, by \Cref{prop:dcp-closed}, and
    \[
        \mathcal{D}(U\psi) \;=\; \big\langle\, Ug_1U^\dagger,\, \dots,\, Ug_mU^\dagger \,\big\rangle.
    \]
\end{proposition}
\begin{proof}
    Write $H = \langle Ug_1U^\dagger, \dots, Ug_mU^\dagger\rangle = U\,\mathcal{D}(\psi)\,U^\dagger$. For $h \in \mathcal{D}^n$, $h(U\psi) = U\psi$ iff $(U^\dagger h U)\psi = \psi$; since $U^\dagger$ is again a Pauli operator, $\CNOT$ gate, or $\T$-type gate, \Cref{prop:dcp-closed} guarantees $U^\dagger h U \in \mathcal{D}^n$ whenever $h \in \mathcal{D}^n$, so this is equivalent to $U^\dagger h U \in \mathcal{D}(\psi)$, i.e.\ to $h \in U\,\mathcal{D}(\psi)\,U^\dagger = H$. Hence $\mathcal{D}(U\psi) = H$.
\end{proof}

\subsection{Z-measurements}

When considering the issue of updating the DCP-stabilizer group under $Z$-measurements, the following is fundamental:
\begin{lemma}\label{lem:z-commute}
    Given a DCP operator $L$ and a product of Pauli $Z$ operators $z$, we have that $L$ and $z$ either commute or anticommute (i.e. $L z = \pm z L$).
\end{lemma}
\begin{proof}
    Write $L=\omega_8^kD(\vec{s},M)X(\vec{x})$ in the normal form of
    \Cref{prop:dcp-normal-form}, and write $z=Z(\vec{v})$ for some
    $\vec{v}\in\F_2^n$.  Since $D(\vec{s},M)$ and $z$ are both diagonal,
    they commute.  On the other hand, the Pauli commutation relations give
    \[
        X(\vec{x})Z(\vec{v})
        =(-1)^{\langle\vec{x},\vec{v}\rangle}
          Z(\vec{v})X(\vec{x}).
    \]
    It follows that
    $Lz=(-1)^{\langle\vec{x},\vec{v}\rangle}zL$, proving the claim.
\end{proof}

Just as in the Pauli stabilizer formalism, \Cref{lem:z-commute} lets us split the update rule for measuring a $Z$-string $z$ into a random and a deterministic case, according to how $z$ interacts with the current generators:

\newcommand{\dcpZMeasurementUpdateStatement}{%
    Let $\ket{\psi}$ be a DCP state, let
    $G=\langle g_1,\dots,g_m\rangle$ be a DCP-stabilizer group for
    $\ket{\psi}$, and let $z$ be a product of Pauli $Z$ operators.
    \begin{enumerate}[nosep,label=(\alph*)]
        \item \emph{(Random outcome.)} Suppose that $z$ anticommutes with
        some element of $G$, and let
        \[
            K=C_G(z)=\{g\in G:gz=zg\}.
        \]
        Measuring $z$ yields each outcome
        $\varepsilon\in\{+1,-1\}$ with probability $\tfrac12$, and
        \[
            H_\varepsilon=\langle K,\varepsilon z\rangle
        \]
        is a DCP-stabilizer group for the post-measurement state.  In the
        canonical representation of \Cref{def:canon-form}, a generating set
        for $K$ is obtained by retaining all diagonal generators, choosing
        an anticommuting lift $g_j$, and replacing every other anticommuting
        lift $g_i$ by $g_i g_j$.  Thus the post-measurement update amounts
        to replacing $g_j$ by $\varepsilon z$ after these row operations.
        \item \emph{(Deterministic outcome.)} If $z$ commutes with every
        element of $G$, the outcome $\varepsilon\in\{+1,-1\}$ is
        determined: it is the unique sign for which
        $\varepsilon z\in\mathcal{D}(\psi)$.  The state is unchanged, so
        $G$ remains a DCP-stabilizer group for it.
    \end{enumerate}
}

\begin{proposition}[Updating under $Z$-measurements]\label{prop:z-update}
\dcpZMeasurementUpdateStatement
\end{proposition}
The proof can be found in the \hyperref[appendix:tableau-update-measurements]{measurement update subsection} of \cref{appendix:dcp-tableaux}.

\subsection{X-measurements}
In the Heisenberg picture, the compatibility criteria and update rules for $X$-measurements closely mirror those of the phase-polynomial picture, described in \cref{thm:x-measurements}.
For a general $X(m)$ measurement, \cref{eq:x-measurement-reduction} lets us conjugate the tableau by $C_m$, perform the single-qubit $X_a$ update below, and conjugate back by $C_m^\dagger$. The identity measurement leaves the tableau unchanged.
Consider a 3PP state $\ket{\psi}$ with DCP-stabilizer group $G$, and suppose that we measure the single-qubit operator $X_a$, obtaining the outcome $\varepsilon = (-1)^\mu$ with $\mu \in \F_2$. Write $V = X(G)$, regarded as a vector space over $\F_2$.
Since $\ket{\psi}$ spans the one-dimensional space $\Fix(G)$, it is atomic for $G$ by \Cref{thm:atomic-basis}, so its support is a coset $S = x_0 \oplus V$ of $V$.
The two top-level cases of \cref{thm:x-measurements} are therefore exactly:
\begin{align}
    \text{(1)}\;\; e_a \not\in V
    \qquad\text{or}\qquad
    \text{(2)}\;\; e_a \in V.
\end{align}
Let us consider case (1) first. Measurements falling into this case are always compatible with the state and yield each outcome with probability $\tfrac12$, so all that remains is to describe how the tableau updates.

Suppose first that $v_a = 0$ for every $v \in V$.
Then qubit $a$ is unentangled from the rest of the system and sits in a computational basis state $\ket{\sigma}$, where $(-1)^\sigma$ is the deterministic outcome of measuring $Z_a$. The measurement simply replaces $\ket{\sigma}$ by $\ket{\pm}$, according to the sign of $\varepsilon$.
On the tableau, we first detach qubit $a$ from the generators, replacing each factor $\S_a$ by the scalar $i^\sigma$ and each factor $\CZ_{aj}$ by $Z_j^\sigma$, which does not change their action on the state, and then add the generator $\varepsilon X_a$.

The general case reduces to this one. Since $e_a \not\in V$, Gaussian elimination yields a vector $w \in \F_2^n$ with $w_a = 0$ and $\langle w, v \rangle = v_a$ for all $v \in V$. The circuit $C = \prod_{i \,:\, w_i = 1} \CNOT_{i,a}$ then clears the $a$-th coordinate of every element of $V$, and it commutes with $X_a$. Measuring $X_a$ is therefore equivalent to applying $C$, measuring $X_a$ as above, and then applying $C^\dagger$. Thus we conjugate the generators by $C$, update them as in the previous paragraph, and conjugate them back by $C^\dagger$, all of which takes polynomial time by \Cref{prop:dcp-closed}.

We now turn to case 2, where $e_a \in V$. By the surjectivity of $\pi$, there is some $g \in G$ of the form $g = \delta^* X(e_a)$, where $\delta^*$ is a diagonal DCP operator. Let $E : \F_2^n \to \Z_8$ denote the phase function of $\delta^*$, defined by $\delta^* \ket{x} = \omega_8^{E(x)} \ket{x}$. This function determines the discrete derivative on which \Cref{thm:x-measurements} is based:

\begin{proposition}\label{prop:derivative-from-E}
    Let $q$ be a phase polynomial for $\ket{\psi}$, so that, up to global phase,
    \[
        \ket{\psi} = \frac{1}{\sqrt{|S|}} \sum_{x \in S} \omega_8^{q(x)} \ket{x},
        \qquad S = x_0 \oplus V.
    \]
    Then, for every $x \in S$,
    \[
        \Delta_{e_a} q(x) = E(x \oplus e_a) = -E(x) \pmod 8.
    \]
\end{proposition}
\begin{proof}
    Since $e_a \in V$, we have $S \oplus e_a = S$. For $x \in S$,
    \[
        g\ket{x} = \delta^*\ket{x \oplus e_a} = \omega_8^{E(x \oplus e_a)}\ket{x \oplus e_a},
    \]
    so comparing the amplitudes of $\ket{x \oplus e_a}$ on both sides of $g\ket{\psi} = \ket{\psi}$ gives
    $q(x) + E(x \oplus e_a) = q(x \oplus e_a) \pmod 8$. Likewise, $g^{-1} = X(e_a)(\delta^*)^{-1}$ fixes $\ket{\psi}$ and
    \[
        g^{-1}\ket{x} = \omega_8^{-E(x)}\ket{x \oplus e_a},
    \]
    so that $q(x) - E(x) = q(x \oplus e_a) \pmod 8$.
\end{proof}

The derivative $\Delta_{e_a} q$ can thus be read off the DCP-stabilizer group, without tracking the phase polynomial itself. Substituting $\Delta_{e_a} q = -E$ into \Cref{thm:x-measurements} gives necessary and sufficient conditions for the compatibility of the measurement in terms of quantities computable from $G$, and \Cref{prop:x-measurement-case-decision} shows how to decide in polynomial time which of the cases 2a, 2b and 2c applies, if any. What remains is to describe how the DCP-stabilizer group is updated in each of them.

The starting point is the same in all three cases. Since $X_a = (\delta^*)^{-1} g$ and $g$ fix $\ket{\psi}$, the projection onto the $\varepsilon$-eigenspace of $X_a$ acts on $\ket{\psi}$ as a diagonal operator:
\begin{align}
    \frac{1}{2}(1 + \varepsilon X_a) \ket{\psi}
    &= \frac{1}{2}\big(1 + \varepsilon (\delta^*)^{-1} g\big) \ket{\psi} \notag \\
    &= \frac{1}{2}\big(1 + \varepsilon (\delta^*)^{-1}\big) \ket{\psi}.
    \label{eq:x-projection-diagonal}
\end{align}

In case 2a, $E(x) = -\Delta_{e_a} q(x) = 4f(x) \pmod 8$ on $S$ for an affine function $f$, which we write as $f(x) = c \oplus \langle u, x \rangle$ with $c \in \F_2$ and $u \in \F_2^n$.
Equivalently, $\delta^*$ acts on the support of $\ket{\psi}$ as the signed $Z$-string $(-1)^c Z(u)$.
This is the case, for instance, whenever $\delta^*$ is itself a signed product of $Z$ operators.
In general, $\delta^*$ may differ from $(-1)^c Z(u)$ by a diagonal operator that fixes every $\ket{x}$ with $x \in S$. As $(\delta^*)^{-1}\ket{\psi} = \delta^*\ket{\psi}$ in this case, \cref{eq:x-projection-diagonal} becomes
\begin{align}
    \frac{1}{2}(1 + \varepsilon X_a) \ket{\psi}
    &= \frac{1}{2}(1 + \varepsilon\, \delta^*) \ket{\psi} \\
    &= \frac{1}{2}\big(1 + (-1)^{\mu + c} Z(u)\big) \ket{\psi}.
\end{align}
Thus, measuring $X_a$ with outcome $\mu$ transforms the state exactly as measuring the $Z$-string $Z(u)$ with outcome $(-1)^{\mu + c}$. The outcome probabilities and the updated DCP-stabilizer group are therefore given by \Cref{prop:z-update}, and can be computed in polynomial time.

In cases 2b and 2c, the support is unchanged and, by \Cref{thm:x-measurements}, the difference $r = q' - q$ between the new and the old phase polynomial is itself a third-order phase polynomial. Writing it as a sum of terms,
\[
    r(x) = r_0 + \sum_i a_i x_i + 2\sum_{i<j} b_{ij} x_i x_j + 4\sum_{i<j<k} c_{ijk} x_i x_j x_k,
\]
adding $r$ to the phase polynomial has the same effect on the state, up to the global phase $\omega_8^{r_0}$, as applying the circuit
\[
    U =
    \prod_i \T_i^{a_i}
    \prod_{i<j} \CS_{ij}^{b_{ij}}
    \prod_{i<j<k} \CCZ_{ijk}^{c_{ijk}}.
\]
By \Cref{prop:heisenberg}, $U G U^\dagger$ is then a DCP-stabilizer group of the post-measurement state. Since $\CS$ and $\CCZ$ can be synthesized from $\CNOT$ and $\T$ gates, $U$ is a circuit of polynomially many such gates, so the new tableau can be computed in polynomial time by conjugating the generators gate by gate, as in \Cref{prop:dcp-update}.

\section{Applications to quantum error correction}\label{sec:noisy-simulations}

We now establish sufficient conditions for non-Clifford quantum error-correction circuits to admit efficient simulation under stochastic Pauli noise, and more generally under stochastic DCP faults.
The main challenge is to ensure that every $X$ measurement remains compatible with the 3PP representation, including after earlier faults and measurements have changed the state.

We first identify two useful measurement mechanisms: checks that are deterministic in the absence of noise, and $X$ measurements that anticommute with a current $Z$-type stabilizer.
We prove that in the ideal case, the former are compatible case (2a), and change to other cases with noise, while the latter are compatible case (1).
We then prove that circuits whose ideal $X$ measurements are either deterministic or case (1) remain compatible throughout arbitrary sequences of DCP faults and measurements.
This reduces the analysis of noisy circuits to conditions on their ideal evolution.
We apply these results to magic state distillation and transversal computation, DCP observable checks used in magic state cultivation, gauge fixing and code switching, teleportation, and cat-state syndrome extraction.

\subsection{Deterministic checks and gauge measurements}

We begin with measurements whose outcomes are deterministic in the ideal circuit, as occurs for stabilizer checks measured on a logical DCP state. Although faults can make these outcomes random, the following theorem shows that any DCP error preserves measurement compatibility.

Throughout, we use the following notation for the projector onto an affine set $A$:
\begin{align}
    \Pi_A = \sum_{x \in A} \ket{x}\bra{x}.
\end{align}

\begin{theorem}[Deterministic measurements under noise]
    \label{thm:deterministic-measurements-noise}
    Let $\ket{\psi}$ be a 3PP state satisfying $X(m) \ket{\psi}=(-1)^{c_0} \ket{\psi}$. Then, measuring $X(m)$ on $\ket{\psi}$ is a compatible case (2a), and for every DCP error $E$, measuring $X(m)$ on $E\ket{\psi}$ is a compatible case (1), (2a) or (2b). Moreover, for every measurement outcome $\epsilon = (-1)^s$ with nonzero probability, the resulting state is proportional to
    \begin{align}
        E_s\Pi_{B_s} \ket{\psi}
    \end{align}
    where $E_s$ is DCP and $B_s\subseteq \supp(\psi)$ is affine.
\end{theorem}
\begin{proof}
    Let us first identify the ideal case. Plugging
    \begin{align}
        \ket{\psi}=\frac{1}{\sqrt{|S|}} \sum_{x \in S} \omega_8^{q(x)} \ket{x},\quad S = x_0 + V
    \end{align}
    into
    \begin{align}
        X(m) \ket{\psi} = (-1)^{c_0} \ket{\psi},
    \end{align}
    we get
    \begin{align}
        \sum_{x \in S} \omega_8^{q(x)} \ket{x\oplus m} = \sum_{x \in S} \omega_8^{q(x)+4c_0} \ket{x}.
    \end{align}
    After a change of variable on the left-hand side, we can identify phases on both sides:
    \begin{align}
        q(x\oplus m)= q(x)+4c_0,
    \end{align}
    which gives us
    \begin{align}
        \Delta_m q = 4c_0.
    \end{align}
    This is case (2a).

    Let us now look at the noisy case. Let $E$ be a DCP error. We write $E=DX(e)$ where $D$ is a diagonal Clifford gate, satisfying $D\ket{x}=\omega_8^{d(x)} \ket{x}$. The function $d$ is a Clifford phase polynomial of the form characterized in \Cref{thm:clifford-phase-polynomial}.
    Let $q'$ be the phase polynomial associated to $E\ket{\psi}$.  Relabeling
    $y=x\oplus e$ gives
    \begin{align}
        \begin{split}
            E\ket{\psi}
            &= \frac{1}{\sqrt{|S|}} \sum_{x \in S} \omega_8^{q(x)} DX(e) \ket{x} \\
            &= \frac{1}{\sqrt{|S|}} \sum_{x \in S} \omega_8^{q(x)} \omega_8^{d(x\oplus e)} \ket{x \oplus e} \\
            &= \frac{1}{\sqrt{|S|}} \sum_{y \in S\oplus e}
            \omega_8^{q(y\oplus e)+d(y)} \ket{y}.
        \end{split}
    \end{align}
    Therefore, $q'(x)=q(x \oplus e)+d(x)$ on $S\oplus e$.  This is again
    a third-order phase polynomial by
    \Cref{cor:phase-polynomial-affine-pullback,thm:clifford-phase-polynomial},
    and
    \begin{align}
        \begin{split}
            \Delta_m q' &= d(x \oplus m) - d(x) + q(x \oplus e \oplus m) - q(x \oplus e) \\
            &= 4c_0 + \Delta_m d(x)
        \end{split}
    \end{align}
    By \Cref{thm:clifford-phase-derivative}, we can write $\Delta_m d$ as
    \begin{align}
        \Delta_m d(x) = 2r + 4\ell(x)
    \end{align}
    where $r\in\Z_4$ is a constant and $\ell:\F_2^n\to\F_2$ is linear. We therefore have
    \begin{align}
        \Delta_m q'(x) = 4c_0 + 2r + 4\ell(x)
    \end{align}
    If $r$ is even, then
    $f(x)=c_0+r/2+\ell(x)\pmod 2$ is affine and
    $\Delta_mq'(x)=4f(x)\pmod8$, so we are in case (2a).  If $r$ is odd,
    write $r=2t+1$ and define the affine Boolean function
    $g(x)=c_0+t+\ell(x)\pmod2$.  Then
    \begin{align}
        \Delta_m q'(x) = 2 + 4g(x) \pmod 8.
    \end{align}
    Moreover, $2g$ is a third-order phase polynomial by the parity-sum
    characterization in \Cref{thm:phase-polynomial-equivalence}.  Hence this
    is case (2b).

    Let us finally identify the resulting state. For a nonzero outcome
    $\epsilon=(-1)^s$, the post-measurement state is proportional to
    \begin{align}
        \sum_{x\in S\oplus e}
        \omega_8^{q'(x)}
        \left(1+(-1)^s\omega_8^{\Delta_mq'(x)}\right)\ket{x}.
    \end{align}
    If $r$ is even, the factor in parentheses is
    \begin{align}
        1+(-1)^{s+f(x)}
        =2\,\mathbf{1}_{\{f(x)=s\}}.
    \end{align}
    Defining
    \begin{align}
        B_s=\{x\in S:f(x\oplus e)=s\},
    \end{align}
    the resulting state is therefore proportional to
    $E\Pi_{B_s}\ket{\psi}$. Since $f$ is affine, $B_s$ is affine,
    and we can take $E_s=E$.

    If $r$ is odd, we instead have
    \begin{align}
        1+(-1)^s\omega_8^{2+4g(x)}
        =1+i(-1)^{s\oplus g(x)}
        =\sqrt{2}\,\omega_8^{1-2(s\oplus g(x))}.
    \end{align}
    Define the diagonal operator $D_s$ by
    \begin{align}
        D_s\ket{x}
        =\omega_8^{-2(s\oplus g(x))}\ket{x}.
    \end{align}
    Since $s\oplus g$ is affine, $D_s$ is diagonal Clifford by
    \Cref{thm:clifford-phase-polynomial}. The resulting state is
    therefore proportional to $D_sE\ket{\psi}$, so we can take
    $E_s=D_sE$ and $B_s=S$.
\end{proof}

Gauge measurements, by contrast, often have random outcomes even in the ideal circuit. A useful compatibility criterion is anticommutation with a current $Z$-type stabilizer: this ensures that the measured $X$ operator maps the state's support to a disjoint affine space, giving case (1).

\begin{theorem}[Anticommuting $Z$ constraint]
    \label{thm:anticommuting-z-constraint}
    Let $\ket{\psi}$ a 3PP states such that $Z(z)\ket{\psi}=(-1)^c \ket{\psi}$. Then any measurement $X(m)$ anticommuting with $Z(z)$ is a compatible case (1).
\end{theorem}
\begin{proof}
    We show $m \notin V$, by proving that for any $v \in V$, $v \cdot z = 0$, while $m \cdot z = 1$ by assumption. Developing the stabilizer equation $Z(z) \ket{\psi} = (-1)^c \ket{\psi}$ gives us:
    \begin{align}
        \begin{split}
            Z(z) \ket{\psi} &= \frac{1}{\sqrt{|S|}} \sum_{x \in S} \omega_8^{q(x)} Z(z) \ket{x} \\
            &= \frac{1}{\sqrt{|S|}} \sum_{x \in S} (-1)^c \omega_8^{q(x)} \ket{x}
        \end{split}
    \end{align}
    On the other hand, we have that $Z(z) \ket{x} = (-1)^{x \cdot z} \ket{x}$. Therefore, we must have for every $x \in S$, $x \cdot z = c$.

    Now, for $v \in V$, we can write $v = x_0 \oplus (x_0 \oplus v)$, and develop the dot product over the XOR:
    \begin{align}
        \begin{split}
            v \cdot z
            &= (x_0 \oplus (x_0 \oplus v)) \cdot z \\
            &= (x_0 \cdot z) \oplus \left((x_0 \oplus v) \cdot z\right) \\
            &= c+c \\
            &= 0
        \end{split}
    \end{align}
    since we have both $x_0 \in S$ and $x_0 \oplus v \in S$. Therefore, $m \notin V$, and we are in case (1).
\end{proof}

These results identify two common mechanisms for compatible measurements. To cover a full noisy protocol, however, we must also account for the changes induced by earlier measurements. We next show that compatibility persists throughout the circuit whenever every ideal $X$ measurement is deterministic or case (1).

\subsection{Circuits with iterative measurements}

To extend the preceding results to circuits with repeated measurements, the key is to track each noisy branch as a DCP error acting on an affine restriction of an ideal branch. This form is preserved throughout the circuit and allows us to establish compatibility using only the ideal measurement conditions.

\begin{theorem}[Iterative measurements under noise]
    \label{thm:iterative-measurements}
    Consider a 3PP circuit such that every outcome-dependent correction is DCP. Suppose that every ideal $X$ measurement is either a compatible case (1) or deterministic. Then arbitrary inserted DCP faults turn the measurements into compatible cases (1), (2a) or (2b). On every nonzero noisy branch, the resulting state after any given measurement has the form
    \begin{align} \label{eq:noisy-form-post-iterative-measurement}
        \ket{\phi} \propto E\Pi_{A}\ket{\psi},
    \end{align}
    where $E$ is DCP, $A \subseteq \supp(\psi)$ is an affine space, and $\ket{\psi}$ is a nonzero ideal branch.
\end{theorem}
\begin{proof}
    We prove the statement by induction.
    At the beginning of the circuit, taking $E=I$ and $A=\supp(\psi)$ gives the correct form for the state.

    Let's now assume that, at a given stage of the circuit, the state has the form of \cref{eq:noisy-form-post-iterative-measurement}.
    We first consider how faults, unitary gates and $Z$ measurements transform the state, before analyzing $X$ measurements.

    A DCP fault $F$ replaces $E$ by $FE$. For a CNOT-dihedral gate $U$ such that
    \begin{align}
        U\ket{x} \propto \ket{f(x)},
    \end{align}
    with $f$ an affine function, we have
    \begin{align}
        UE\Pi_A\ket{\psi}
        =(UEU^\dagger)\Pi_{f(A)}U\ket{\psi}.
    \end{align}
    Since $U$ normalizes the DCP group, since $f(A)$ is affine, and since $\ket{\psi'}=U\ket{\psi}$ is the ideal state after the unitary, this preserves the required form.

    For a measurement of $Z(v)$, \cref{lem:z-commute} gives $Z(v)E=(-1)^c EZ(v)$ for some $c\in\F_2$.
    Thus, commuting the measurement projector through $E$ changes the outcome label from $s$ to $s\oplus c$.
    By \cref{eq:phase-polynomial-z-measurement-update}, this measurement restricts the support to the affine hyperplane
    \begin{align}
        H=\{x\in\F_2^n:\langle v,x\rangle=s\oplus c\}.
    \end{align}
    The resulting state is therefore proportional to
    \begin{align}
        \begin{split}
            E\Pi_H\Pi_A\ket{\psi}
            &=E\Pi_A\Pi_H\ket{\psi} \\
            &=E\Pi_A\Pi_H \Pi_H\ket{\psi} \\
            &=E\Pi_{A \cap H} \Pi_H\ket{\psi}.
        \end{split}
    \end{align}
    We take the new ideal reference branch to be $\ket{\psi'}\propto\Pi_H\ket{\psi}$. It is a valid ideal branch, since it corresponds to measuring the previous ideal state $\ket{\psi}$ and taking the shifted outcome $s \oplus c$, which has nonzero probability of occurring (otherwise, the noisy output would also be zero).
    Moreover, $A \cap H$ is affine and contained in $\supp(\psi')$.
    Thus, the required form is preserved.

    Let us now consider a measurement of $X(m)$ with outcome
    $\epsilon=(-1)^s$, and write
    \begin{align}
        P_s=\frac{I+(-1)^sX(m)}{2},
    \end{align}
    as well as $\supp(\psi)=x_0+V$ and $A=a+W$.
    By the induction hypothesis, $\ket{\psi}$ is an ideal branch.
    The measurement on $\ket{\psi}$ is therefore case (1) or deterministic by assumption.
    Suppose first that the ideal measurement is deterministic, and that
    $m\in W$. Then, $A\oplus m=A$ and we have
    \begin{align}
        [X(m),\Pi_A]=0.
    \end{align}
    Therefore, if $X(m)\ket{\psi}=(-1)^{c_0}\ket{\psi}$, then
    \begin{align}
        X(m)\Pi_A\ket{\psi}
        =(-1)^{c_0}\Pi_A\ket{\psi}.
    \end{align}
    The normalized restricted state
    $\ket{\chi}\propto\Pi_A\ket{\psi}$ is a 3PP state, so
    \Cref{thm:deterministic-measurements-noise} gives
    \begin{align}
        P_sE\ket{\chi}
        \propto E_s\Pi_B\ket{\chi}
        \propto E_s\Pi_B\ket{\psi},
    \end{align}
    where $E_s$ is DCP and $B\subseteq A$ is affine. We used the fact that $\Pi_B \Pi_A=\Pi_B$ (since $B\subseteq A$). This preserves the required form, and the noisy measurement is case (2a) or (2b).

    It remains to consider an ideal case (1) measurement, or an ideal deterministic measurement with $m\notin W$.
    In the first situation, $m\notin V$ implies $m\notin W$, since $W\subseteq V$; in the second, this holds by assumption. Since $\supp(\Pi_A \ket{\psi})=A=a+W$, and the DCP operator $E$ can only translate the support by a constant vector, it means that $\supp(E \Pi_A \ket{\psi})=a'+W$ and $m \notin W$ implies that the measurement is case (1).

    To identify the resulting state, write $E=\omega_8^k X(e)D$, where $D$
    is diagonal Clifford. Since $X(e)$ commutes with $P_s$, we have
    \begin{align}
        P_sE\Pi_A\ket{\psi}
        \propto X(e)P_sD\Pi_A\ket{\psi}.
    \end{align}
    We therefore first calculate $P_sD\Pi_A\ket{\psi}$.
    Writing the phases of $\ket{\psi}$ and $D$ as $q$ and $d$,
    respectively, gives
    \begin{align}
        P_sD\Pi_A\ket{\psi}
        \propto
        \sum_{x\in A}\omega_8^{q(x)+d(x)}
        \left(\ket{x}+(-1)^s\ket{x\oplus m}\right).
    \end{align}
    Its support is the affine set
    \begin{align}
        A'=A+\langle m\rangle.
    \end{align}
    Every $y\in A'$ has a unique expression $y=x\oplus tm$, where $x\in A$ and $t\in\F_2$.
    Moreover, $t$ is an affine function of $y$: choose a linear functional $\lambda$ such that $\lambda(W)=0$ and $\lambda(m)=1$. We can then set $t(y)=\lambda(y\oplus a)$. Indeed, writing $y=a \oplus w \oplus tm$, we get
    \begin{align}
        \lambda(y\oplus a)
        =\lambda(w\oplus tm)
        =\lambda(w)\oplus t\lambda(m)
        =t.
    \end{align}
    The phase of the noisy output at $y$ is therefore
    \begin{align}
        q(x)+d(x)+4st.
    \end{align}

    Let us compare this with an ideal reference output.
    If the ideal measurement is case (1), choose either outcome
    $(-1)^{s_0}$ and take
    $\ket{\psi'}\propto P_{s_0}\ket{\psi}$.
    If the ideal measurement is deterministic, take $s_0=c_0$
    and $\ket{\psi'}=\ket{\psi}$.

    Let us now show that the output state has the form
    \begin{align}
        P_sE\Pi_A\ket{\psi}
        \propto X(e)D'\Pi_{A'}\ket{\psi'}
    \end{align}
    for $D'$ a diagonal Clifford operator that we will determine, and $A'\subseteq\supp(\psi')$. This last inclusion is true, since in the ideal case (1), the reference output has support $S+\langle m\rangle$, which contains
    $A'=A+\langle m\rangle$ because $A\subseteq S$.
    In the ideal deterministic case, $\ket{\psi'}=\ket{\psi}$
    and $m\in V$, so
    $A'\subseteq S+\langle m\rangle=S=\supp(\psi')$.

    To construct $D'$, let us understand the phases of $\ket{\psi'}$, and compare them to the phases of the output state. If the ideal measurement case is (1), the phase of $\ket{\psi'}$ at
    $y=x\oplus tm$ is $q(x)+4s_0t$. If the ideal measurement is deterministic, with outcome $c_0$, it can only be a case (2a) with $\Delta_m q = 4c_0$. The relation
    \begin{align}
        q(x\oplus m)=q(x)+4c_0
    \end{align}
    then shows that its phase at $y$ is again $q(x)+4s_0t$.

    The phase difference between the noisy output and this
    reference state is consequently
    \begin{align}
        d'(y)=d(x)+4(s\oplus s_0)t.
    \end{align}
    To show that this defines a diagonal Clifford operator, we
    express it in terms of $y$. Since $x=y\oplus tm$, we get
    \begin{align}
        \begin{split}
            d'(y)
            &=d(y\oplus tm)+4(s\oplus s_0)t\\
            &=d(y)+t\Delta_m d(y)+4(s\oplus s_0)t.
        \end{split}
    \end{align}
    By \Cref{thm:clifford-phase-derivative}, we can write
    $\Delta_m d(y)=2r+4\ell(y)$, where $\ell$ is linear. Thus,
    \begin{align}
        d'(y)
        =d(y)+2rt+4t\ell(y)+4(s\oplus s_0)t \pmod 8.
    \end{align}
    Since $t$ is affine, this is a Clifford phase polynomial.
    Letting $D'$ be the corresponding diagonal Clifford operator,
    we obtain
    \begin{align}
        P_sE\Pi_A\ket{\psi}
        \propto X(e)D'\Pi_{A'}\ket{\psi'}.
    \end{align}
    This preserves the required form.

    Finally, we consider an outcome-dependent DCP correction.
    Let $C_{\mathrm{n}}$ and $C_{\mathrm{i}}$ be the corrections selected by the noisy record and the ideal reference record, respectively.
    Then
    \begin{align}
        C_{\mathrm{n}}E\Pi_A\ket{\psi}
        =
        \bigl(C_{\mathrm{n}}EC_{\mathrm{i}}^\dagger\bigr)
        \bigl(C_{\mathrm{i}}\Pi_A C_{\mathrm{i}}^\dagger\bigr)
        C_{\mathrm{i}}\ket{\psi}.
    \end{align}
    The first factor is DCP.
    Since $C_{\mathrm{i}}$ translates computational-basis labels by a fixed vector, the second factor is the projector onto the corresponding affine translate of $A$, which lies within $\supp(C_{\mathrm{i}}\psi)$.
    Moreover, $C_{\mathrm{i}}\ket{\psi}$ is the ideal branch after its prescribed correction.
    Thus, outcome-dependent DCP corrections preserve the induction invariant.

    Every noisy $X$ measurement is therefore case (1), (2a) or (2b).
\end{proof}

To apply this theorem, it therefore suffices to verify that every ideal $X$ measurement is deterministic or case (1), for example by identifying an anticommuting $Z$-type stabilizer. No separate compatibility analysis is needed for each noisy branch. The restriction on feedforward is essential: the theorem allows outcome-dependent DCP corrections, but does not assert noise robustness for arbitrary adaptive choices of CNOT-dihedral gates or measured observables.

We now illustrate this criterion in several quantum error-correction protocols.

\subsection{Examples}

We now apply the preceding results to several classes of quantum error correction circuits.

\subsubsection{Magic state distillation and non-Clifford computation}

Logical CNOT-dihedral gates implemented by physical CNOT-dihedral circuits, interleaved with CSS stabilizer checks, directly satisfy the hypotheses of \cref{thm:iterative-measurements}: each ideal gate preserves the code space, so the subsequent checks are deterministic.
This covers magic state distillation based on (generalized) triorthogonal codes, including transversal and quasitransversal constructions with diagonal Clifford corrections \cite{bravyi2005universal,bravyi2012magic, campbell2017unified, haah2018codes, scruby2025quantum, vuillot2022quantum, zhu2025nonclifford}.
It also covers Litinski's reformulation of these protocols as sequences of $Z$-type $\pi/8$ rotations followed by $X$ readouts \cite{litinski2019magic}, since these rotations are CNOT-dihedral and the syndrome readouts are ideally deterministic.
The same argument applies to recent code constructions supporting logical non-Clifford gates through constant-depth physical circuits of $Z$, $CZ$, and $CCZ$ gates, including algebraic geometry constructions \cite{golowich2024asymptotically, golowich2024quantum, he2025asymptotically, nguyen2024good, wills2024constantoverhead, wills2026improved, zhu2025nonclifford} and quantum LDPC constructions based on cup products \cite{bauer2026planar, breuckmann2026cups, davydova2026universal, fuente2026highthreshold, tiew2026copycup}.
These gates can be used either for magic state preparation or for direct non-Clifford computation: in both settings, subsequent CSS stabilizer checks are ideally deterministic, so \cref{thm:iterative-measurements} applies.

\subsubsection{Clifford checks, cultivation, and non-Pauli codes}

The preceding examples rely on Pauli stabilizer checks that are deterministic in the ideal circuit.
More generally, checking non-Pauli Clifford observables provides a way to detect errors in magic states, as exploited in distillation protocols based on measuring Clifford checks \cite{jones2013multilevel, knill2004faulttolerant, meier2012magicstate,haah2017magic,haah2018codes} and magic state cultivation \cite{chamberland2020very, gidney2024magic, vaknin2026high}.
We now show how that the same deterministic-check argument extends to DCP observables.

Let $M$ be a Hermitian DCP operator satisfying $M^2=I$, and suppose $M\ket{\psi}=\epsilon_0\ket{\psi}$.
Preparing an ancilla in $\ket{+}$ and applying controlled-$M$ gives
\begin{align}
    \frac{\ket{0}\ket{\psi}+\ket{1}M\ket{\psi}}{\sqrt{2}}
    =
    \frac{\ket{0}+\epsilon_0\ket{1}}{\sqrt{2}}\otimes\ket{\psi},
\end{align}
so the ancilla's $X$ readout is deterministic.
Moreover, controlled-$M$ admits a CNOT-dihedral realization: controlling $X$, $S$, and $CZ$ gives CNOT, $CS$, and $CCZ$, respectively, while controlling an eighth-root global phase gives a power of $T$.
Thus, sequences of such ideally deterministic checks satisfy \cref{thm:iterative-measurements}.

For example, the magic state $\ket{T}=T\ket{+}$ satisfies $W_T\ket{T}=\ket{T}$ with
\begin{align}
    W_T=TXT^\dagger=\frac{X+Y}{\sqrt{2}}=\omega_8^{-1}SX.
\end{align}
A controlled-$W_T$ circuit therefore checks a deterministic Clifford observable.
Similarly, the $\ket{CCZ}$ state is stabilized by the DCP operators $X_1CZ_{23}$ and their permutations.
The same reasoning applies to encoded checks whenever their physical implementations have the controlled-DCP form above.
This provides a compatibility certificate for such checks in distillation and cultivation, while preparation, growth, and any additional resource-consumption measurements must also satisfy the compatibility conditions.

The argument also applies to syndrome extraction in non-Pauli codes whose stabilizers are Hermitian DCP operators, such as products of $X$ operators decorated with $CZ$ gates.
These checks are deterministic on ideal code states, even when the stabilizer generators do not commute on the full Hilbert space.
Such codes underlie recent protocols for direct non-Clifford computation in two dimensions \cite{davydova2026universal}, including syndrome-extraction circuits with just-in-time decoding \cite{bauer2026planar, davydova2026universal, fuente2026highthreshold}.

Not every Clifford-check implementation has this form.
In particular, distillation and cultivation protocols using controlled-Hadamard or controlled-SWAP operations \cite{chamberland2020very, jones2013multilevel, knill2004faulttolerant, meier2012magicstate} are not directly covered by the argument above.
A change of basis or an alternative implementation may bring a check into DCP form, but compatibility must then be established for the entire transformed circuit.

\subsubsection{Gauge fixing, code switching, and teleportation}

Consider a CSS subsystem code whose stabilizer eigenvalues are fixed.
Measuring a complete set of independent $Z$-type gauge operators fixes the gauge qubits in a $Z$ basis.
Each subsequent $X$-type gauge measurement either anticommutes with a currently fixed $Z$-type gauge operator, making it a compatible case (1), or has an outcome determined by the stabilizers and earlier $X$-type gauge measurements.
Thus, successive rounds of $Z$- and $X$-type gauge measurements satisfy \cref{thm:iterative-measurements}, establishing compatibility under DCP faults.

The same reasoning applies to CSS gauge fixing used in code switching \cite{bombin2014gauge} and lattice surgery \cite{horsman2012surface}: newly measured $X$-type checks are compatible whenever they anticommute with a current $Z$ constraint or are already fixed.
For example, a lattice-surgery measurement of $\bar X_A\bar X_B$ is case (1) whenever block $B$ is initialized in $\ket{0}_L$, since the measured operator anticommutes with the fixed constraint $\bar Z_B=+1$.

State injection provides another simple application.
To inject a $T$ gate into a 3PP state $\ket{\psi}$, prepare an ancilla in $\ket{T}=T\ket{+}$, apply a CNOT from the data to the ancilla, and measure the ancilla in the $Z$ basis.
For outcome $s\in\{0,1\}$, the data state is, up to global phase,
\begin{align}
    \ket{\psi_s}=
    \begin{cases}
        T\ket{\psi}, & s=0,\\
        T^\dagger\ket{\psi}, & s=1.
    \end{cases}
\end{align}
Applying the correction $S^s$ therefore yields $T\ket{\psi}$ in both branches.
The entire gadget uses only a 3PP ancilla, a CNOT, a $Z$ measurement, and a diagonal Clifford correction, so it preserves the 3PP representation under DCP faults.
Encoded implementations using lattice surgery additionally require checking the compatibility of the physical measurements realizing the logical operations.

\subsubsection{Cat-state syndrome extraction}

The physical readouts of a deterministic check can have non-deterministic intermediate measurements, as is the case for instance in Shor's style syndrome extraction \cite{gottesman2009introduction, shor1997faulttolerant}. Consider measuring a weight-$w$
$X$ stabilizer using
\begin{align}
    \ket{\mathrm{cat}_w}
    =\frac{\ket{0}^{\otimes w}+\ket{1}^{\otimes w}}{\sqrt{2}},
\end{align}
with each cat qubit controlling the corresponding $X$ stabilizer measurement, followed by individual $X$ readouts of the cat.

The coupling preserves the cat's $Z_iZ_j$ constraints.
Before each of the first $w-1$ readouts, a constraint
$Z_iZ_j$ involving the measured qubit and an unread qubit
therefore certifies case~(1). The product of all readout
signs equals the measured syndrome. On an ideal code state,
this syndrome is fixed, so the final readout is
deterministic conditional on the preceding outcomes.
\Cref{thm:iterative-measurements} consequently covers noisy cat-state extraction,
even though most individual ancilla outcomes are random.

\section{Benchmarking results}
\label{sec:benchmarking-results}

\newcommand*{\micros}[1]{\num[round-precision=2,round-mode=places]{\fpeval{1000000*(#1)}}}
\newcommand*{\millis}[1]{\num[round-precision=2,round-mode=places]{\fpeval{1000*(#1)}}}

\begin{figure*}
    \includegraphics[width=0.7\linewidth]{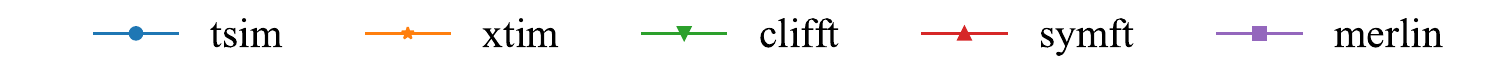} \\[-5pt]
    \begin{subfigure}[b]{0.32\textwidth}
        \includegraphics[width=\linewidth]{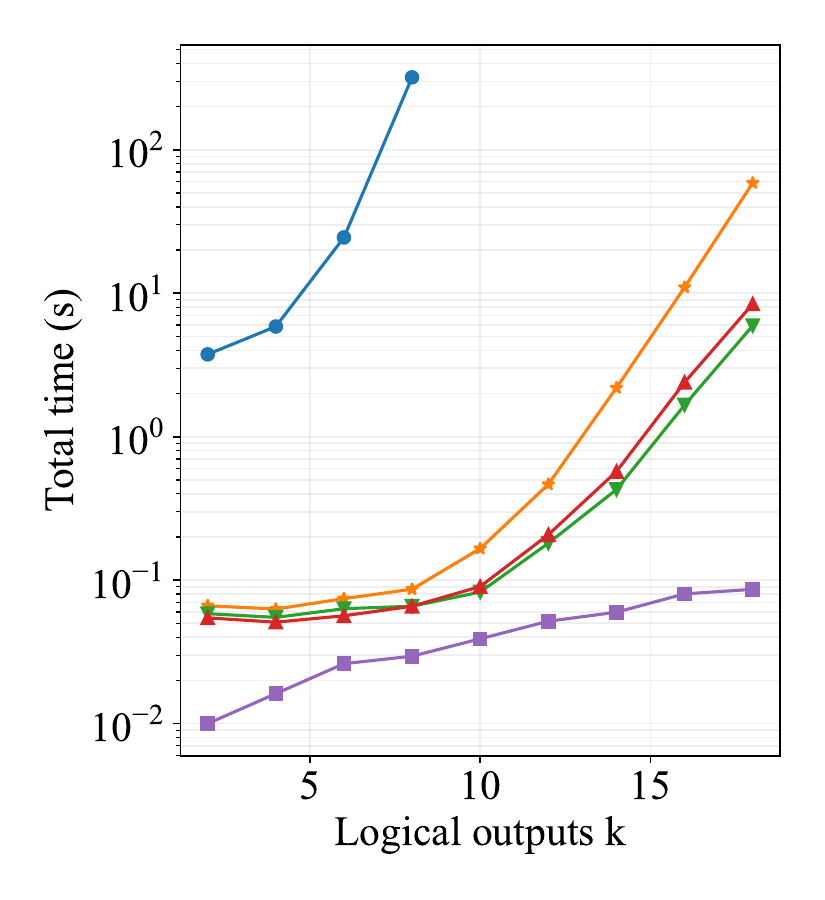}
        \vspace{-20pt}
        \subcaption{} \label{fig:bh-total-time}
    \end{subfigure}
    \begin{subfigure}[b]{0.32\textwidth}
        \includegraphics[width=\linewidth]{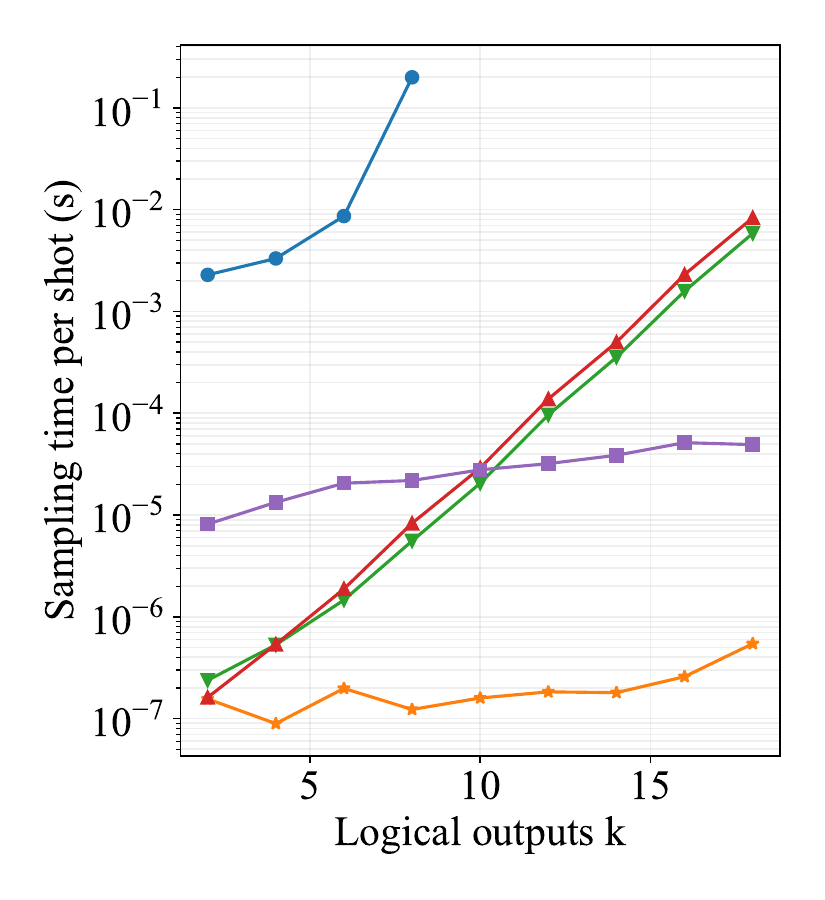}
        \vspace{-20pt}
        \subcaption{} \label{fig:bh-sample-time}
    \end{subfigure}
    \begin{subfigure}[b]{0.32\textwidth}
        \includegraphics[width=\linewidth]{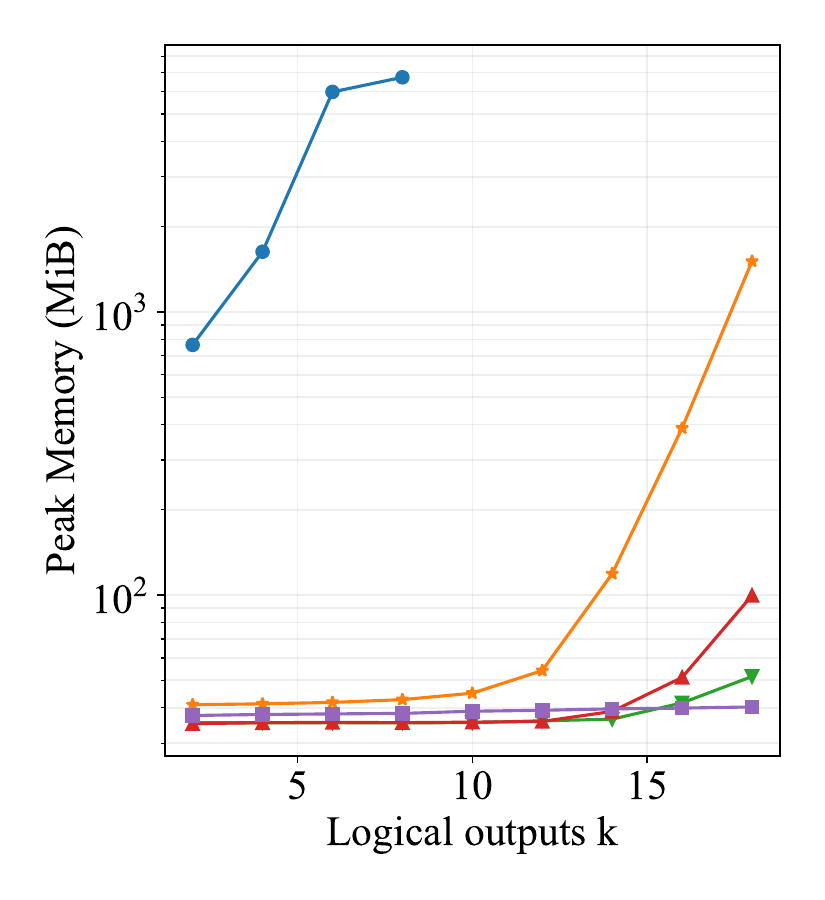}
        \vspace{-20pt}
        \subcaption{} \label{fig:bh-memory}
    \end{subfigure}
    \caption{
        Comparison of different simulators for 1000 attempted shots of the Bravyi-Haah $3k+8$ to $k$ magic state distillation protocol \cite{bravyi2012magic} for increasing numbers of logical qubits $k$.
        (\subref{fig:bh-total-time})~Total runtime, including pre-processing.
        (\subref{fig:bh-sample-time})~Average time per shot, excluding pre-processing.
        (\subref{fig:bh-memory})~Peak memory consumption.
    }
    \label{fig:bh}
\end{figure*}

\setlength{\tabcolsep}{5pt}
\begin{table}[t]
    \centering
    \textbf{Total runtime (including pre-processing) in \qty{}{\milli\second}}
    \begin{tabular}{lrrrrr}
        \toprule
        \textbf{Circuit} &
        \texttt{tsim} &
        \texttt{xtim} &
        \texttt{clifft} &
        \texttt{symft} &
        \texttt{merlin} \\
        \midrule
        15-to-1  & \millis{5.029947541002912} & \millis{0.08585524999944028} & \millis{0.09955558299770928} & \millis{0.057665624997753184} & \textbf{\millis{0.012521416996605694}} \\
        Cultiv-3 & OOM & \millis{0.06529049999517156} & \millis{0.05429354199804948} & \millis{0.06362920799801941} & \textbf{\millis{0.02898095799901057}} \\
        Cultiv-5 & OOM & \millis{0.2535097500003758} & \textbf{\millis{0.067934000002424}} & \millis{0.5772002910016454} & \millis{0.20201220999661018} \\
        Switch-3 & OOM & \millis{2.817529874999309} & OOM & ERR & \textbf{\millis{0.5736940430033428}} \\
        Switch-5 & OOM & OOM & OOM & ERR & \textbf{\millis{6.959127001005982}} \\
        \bottomrule
    \end{tabular}

    \bigskip
    \textbf{Time per shot (excluding pre-processing) in \qty{}{\micro\second}}
    \begin{tabular}{lrrrrr}
        \toprule
        \textbf{Circuit} &
        \textbf{tsim} &
        \textbf{xtim} &
        \textbf{clifft} &
        \textbf{symft} &
        \textbf{merlin} \\
        \midrule
        15-to-1  & \micros{0.0034468530410013045} & \micros{4.025839989481028e-07} & \micros{3.707920004671905e-07} & \textbf{\micros{1.8404200091026724e-07}} & \micros{9.620625001844019e-06} \\
        Cultiv-3 & OOM & \micros{3.898329996445682e-07} & \micros{5.06666001456324e-07} & \textbf{\micros{3.4187499841209503e-07}} & \micros{2.8341749999526654e-05} \\
        Cultiv-5 & OOM & \micros{7.704665997152915e-06} & \micros{7.97404200056917e-06} & \textbf{\micros{6.517500001791632e-06}} & \micros{0.00020014883399926476} \\
        Switch-3 & OOM & \textbf{\micros{2.652375002071494e-06}} & OOM & ERR & \micros{0.0005499323750009353} \\
        Switch-5 & OOM & OOM & OOM & ERR & \textbf{\micros{0.006704935917001421}} \\
        \bottomrule
    \end{tabular}

    \medskip
    \raggedright
    \footnotesize
    \quad OOM: Simulator tried to allocate more than 20GB of memory. \\
    \quad ERR: Simulator failed with a runtime error.

    \caption{
        Simulation time for 1000 attempted shots of the 15-to-1 magic state distillation \cite{bravyi2005universal}, magic state cultivation without the escape stage \cite{gidney2024magic} and the code switching protocol from Ref. \cite{li2026transversal}.
    }
    \label{tab:other-circuits}
\end{table}

We provide an open-source implementation of our simulator, which we call \texttt{merlin} \cite{koch2026merlin}.
In this section, we benchmark \texttt{merlin} against existing non-Clifford simulation libraries on a range of magic-state preparation circuits.
Concretely, we compare against Tsim~\cite{haenel2026tsim}, Clifft~\cite{chase2026clifft}, SymFT~\cite{fang2026symft}, and Xtim~\cite{surti2026efficient}.
All benchmarks use a noise parameter $p=10^{-3}$ and 1000 attempted shots per circuit.
The cultivation benchmarks exclude the escape stage.

While these simulators each use different representations and algorithms, they are all particularly effective when the number of non-stabilizer degrees of freedom remains small.
This is a favorable regime for many existing magic-state preparation protocols.
At the logical level, the target states often only represent a small number of magic states $k$, with $k = 1$ being especially simple.
These simulators exploit the fact that this simplicity can persist after the protocol is expanded into a noisy physical circuit.
However, their cost can increase rapidly as the amount of logical magic grows.
In contrast, our approach does not require $k$ to remain fixed or small.
Our state representation and update rules have polynomial complexity in the physical system size and circuit size, and therefore introduce no exponential dependence on the number of logical magic states $k$.

Figure~\ref{fig:bh} illustrates this difference using the Bravyi-Haah $3k+8$ to $k$ magic state distillation protocol \cite{bravyi2012magic} for increasing numbers of logical output magic states $k$.
The other simulators exhibit exponentially increasing runtime and/or memory consumption as $k$ grows, consistent with the growth of the non-stabilizer representations that they maintain.
In contrast, the runtime and memory consumption of \texttt{merlin} grow smoothly with $k$, reflecting the polynomial scaling of our representation.

Table~\ref{tab:other-circuits} reports results for the 15-to-1 magic state distillation~\cite{bravyi2005universal}, magic-state cultivation~\cite{gidney2024magic} and the code switching protocol for $CCZ$ magic states introduced in Ref. \cite{li2026transversal}.
The distillation and cultivation circuits have previously been used as benchmark cases for non-Clifford simulators~\cite{chase2026clifft,fang2026symft,surti2026efficient}.
They each prepare a single logical magic state and therefore lie in the $k=1$ regime, where Xtim, Clifft, and SymFT achieve substantially lower per-shot runtimes than \texttt{merlin}.
For the code switching benchmarks, we use the $[\![27,3,3]\!]$ and $[\![81,3,5]\!]$ BT code constructions of \cite{li2026transversal} which both encode $k = 3$ logical qubits.
On these circuits, the other simulators' techniques to keep their non-stabilizer representations compact appear to be less effective.
Only Xtim completes the distance-3 simulation, while all other simulator-circuit combinations run out of memory or fail with a runtime error.

These examples highlight the complementary regime targeted by our approach:
while specialized methods can have smaller constant factors when the effective non-stabilizer representation remains small, \texttt{merlin} can handle circuits for which such representations cannot be found or become prohibitively large.

\section{Universal simulation}
\label{sec:universal-simulation}

To simulate circuits that leave the DCP state class, we represent the evolving state as a linear combination of DCP states. The analogous construction using Pauli stabilizer states is quantified by stabilizer rank.

\begin{definition}[Stabiliser rank \cite{bravyi2016improved}]\label{def:stabilizer-rank}
    The \emph{stabilizer rank} $\chi(\psi)$ of a state $\ket{\psi}$ is the smallest integer $\chi$ such that $\ket{\psi}$ can be written as a linear combination $\ket{\psi} = \sum_{i=1}^\chi c_i\ket{\varphi_i}$ of stabiliser states $\ket{\varphi_i}$ with complex coefficients $c_i$.
\end{definition}

The notion of stabiliser rank easily generalises to \emph{DCP rank}:

\begin{definition}[DCP rank]\label{def:dcp-rank}
    The \emph{DCP rank} $\zeta(\psi)$ of a state $\ket{\psi}$ is the smallest integer $\zeta$ such that $\ket{\psi}$ can be written as a linear combination $\ket{\psi} = \sum_{i=1}^\zeta c_i\ket{\varphi_i}$ of DCP states $\ket{\varphi_i}$ with complex coefficients $c_i$.
\end{definition}

Since stabiliser states are DCP states, we trivially have $\zeta(\psi) \leq \chi(\psi)$.
Furthermore, if $\ket{\psi}$ is prepared by a $\CNOT+\T+\Hgate$ circuit with $h$ Hadamards, then we trivially have $\zeta(\psi) \leq 2^h$.
This is because each Hadamard gate can be decomposed into $\Hgate = (X + Z) / \sqrt 2$, leaving only gates that preserve DCP states.
In general, we expect $\zeta$ to scale exponentially, however there might be decomposition strategies that achieve better scaling than this naive approach.
We leave it to future work to investigate this question and improve upon the bounds given here.

Given an explicit decomposition
\begin{align}
    \ket{\psi}=\sum_{j=1}^{R}c_j\ket{\phi_j}
\end{align}
into DCP states, we can compute any computational-basis amplitude by evaluating the amplitude of each term and summing the contributions.
By \cref{thm:efficient-simulation}, this requires $O(n^3R)$ arithmetic operations.
Here $R$ denotes the number of terms in the supplied decomposition, which need not achieve the minimum DCP rank $\zeta(\psi)$.
The cost of constructing the decomposition must also be accounted for.
Throughout, the global-phase convention of each DCP term must be tracked consistently with its coefficient $c_j$.
Weak simulation, that is, sampling computational-basis outcomes, requires an additional argument.
Computing marginal probabilities directly introduces $R^2$ cross terms involving overlaps of projected DCP states, which our polynomial-time algorithms do not generally evaluate.
Several approaches have been proposed in the stabilizer-decomposition literature to avoid or mitigate analogous difficulties.
For example, Refs.~\cite{bravyi2016improved, bravyi2019simulation} approximate marginal probabilities via a norm estimation technique based on overlaps with carefully chosen random stabiliser states.
Ref.~\cite{bravyi2019simulation} also proposed a Metropolis-Hastings Markov Chain Monte Carlo method for approximate sampling.
Finally, Bravyi and Liu \cite{bravyi2022how} introduced the \emph{gate-by-gate} algorithm for exact sampling without computing marginals.
The gate-by-gate algorithm applies directly because it requires only computational-basis amplitudes of circuit prefixes.
We describe this construction below.

Suppose $\ket{\psi} = C\ket{0}$ is prepared by some $\CNOT+\T+\Hgate$ circuit $C = U_m \cdots U_1$ made up of gates $U_i$.
Our goal is to sample from the distribution $P(x) := |\bra{x}C\ket{0}|^2$.
Let $C_k = U_k \cdots U_1$ be the subcircuit containing only the first $k$ gates and define the output distribution at time step $k$ as $P_k(x) := |\bra{x}C_k\ket{0}|^2$.
The idea behind gate-by-gate sampling is to start with a sample from the trivial distribution $P_0$ (i.e. the all-zero bitstring) and update it one gate at a time until we arrive at a sample that is distributed according to $P_m = P$.
This is because adjacent distributions $P_{k-1}$ and $P_k$ are closely related since they only differ by a single gate $U_k$.
Suppose $U_k$ is a $\T$ gate on qubit $i$, then
\[ P_k(x) = |\bra{x}\T_iC_{k-1}\ket{0}|^2 = |\bra{x}C_{k-1}\ket{0}|^2 = P_{k-1}(x), \]
i.e. the distribution stays the same.
For a $\CNOT$ between qubits $i$ and $j$ we have
\begin{align}
    \begin{split}
        P_k(x) &= |\bra{x}\CNOT_{ij}C_{k-1}\ket{0}|^2 \\
        &= |\bra{x \oplus x_ie_j}C_{k-1}\ket{0}|^2 \\
        &= P_{k-1}(x \oplus x_ie_j),
    \end{split}
\end{align}
so a sample $x \sim P_{k-1}$ can be updated to $x' := x \oplus x_ie_j \sim P_k$.
Applying a Hadamard on the other hand will require some more work.
Without loss of generality, suppose $U_k$ is a Hadamard on the first qubit.
Writing $y=(x_2,\ldots,x_n)$, we have
\begin{align}
    P_k(x_1,y)
    &=P_k(x_1\mid y)P_k(y)\\
    &=P_k(x_1\mid y)P_{k-1}(y)\\
    &=\frac{P_k(x_1,y)}{P_k(0,y)+P_k(1,y)}P_{k-1}(y),
\end{align}
whenever $P_{k-1}(y)>0$.
Here $P_k(y)=P_{k-1}(y)$ because the Hadamard acts only on the first qubit.

Thus, given a sample $(x_1,y)\sim P_{k-1}$, we retain $y$ and draw a new first bit according to
\begin{align}
    \Pr(x'_1=b\mid y)
    =\frac{P_k(b,y)}{P_k(0,y)+P_k(1,y)},
    \qquad b\in\{0,1\}.
\end{align}
The denominator is positive for every $y$ sampled with nonzero probability.
Each Hadamard therefore requires two computational-basis amplitude evaluations, whose squared magnitudes give the required probabilities.
A circuit containing $h$ Hadamards consequently requires $2h$ such evaluations.
Let $R_j$ be the number of terms in a DCP decomposition of the state immediately after the $j$th Hadamard.
Excluding the cost of constructing these decompositions and sampling from the resulting probabilities, the amplitude evaluations required for one sample take
\begin{align}
    O\left(n^3\sum_{j=1}^{h}R_j\right)
\end{align}
arithmetic operations.
When using the previously discussed decomposition $H=(X+Z)/\sqrt{2}$, we get $R_j := 2^j$ such that
\begin{align}
    \sum_{j=1}^{h}R_j
    =\sum_{j=1}^{h}2^j
    =2^{h+1}-2
    = O(R_h).
\end{align}
In other words, the cost of all intermediate amplitude evaluations is asymptotically dominated by the cost of evaluating the final state with $R_h$ terms.
This means that gate-by-gate sampling matches the complexity of strong simulation in that setting.
The same conclusion holds for any decomposition strategy whose prefix term counts satisfy $\sum_{j=1}^{h}R_j=O(R_h)$.

\section{Discussion}
\label{sec:discussion}

We have developed an exact polynomial-time simulation framework for a broad class of non-Clifford circuits, with no restriction on the number of logical qubits or non-Clifford gates.
It rests on two equivalent descriptions: third-order phase-polynomial states, specified by affine supports and polynomial phases, and DCP stabilizer states, which extend the stabilizer picture beyond Pauli operators and yield a tableau-based simulator.
We characterize precisely which Pauli measurements preserve this state class and provide efficient tests and update rules supporting intermediate measurements and classical feedforward.

For quantum error correction, a simple condition on the ideal circuit guarantees compatibility under arbitrary inserted DCP faults: each ideal $X$ measurement must either be deterministic or anticommute with a fixed $Z$-type constraint.
This remains sufficient even when noise makes deterministic outcomes random, establishing efficient simulation for protocols involving magic state preparation, non-Pauli checks, and CSS gauge fixing.

Our open-source implementation, \texttt{merlin}~\cite{koch2026merlin}, makes this framework available for circuit-level simulation.
Benchmarks show favorable runtime and memory scaling with the number of logical outputs in Bravyi-Haah distillation, and successful simulation of the distance-5 code switching circuit on which all other tested simulators exhaust memory or fail.
On the other hand, existing methods currently achieve lower per-shot runtimes on the single-output distillation and cultivation examples.

A natural next step is to use this framework to evaluate recent proposals for non-Clifford computation with quantum LDPC codes, including transversal gates, code switching, and high-rate magic state preparation, whenever their circuit implementations satisfy our compatibility conditions.
Combined with suitable decoders, our simulator could enable direct comparisons of logical error rates, acceptance probabilities, and resource overheads at system sizes beyond the reach of existing methods.

An important direction is to extend the range of supported measurement gadgets.
Some magic state distillation protocols based on gate injection, such as the parity-checking constructions of Campbell and Howard~\cite{campbell2018magic}, fall outside the sufficient conditions of  \cref{thm:iterative-measurements} and therefore require a separate compatibility analysis.
Extending our guarantees to cover these protocols would broaden the scope of the simulator.
Another possibility is to describe states in a rotated local Clifford basis, potentially bringing Hadamard-based cultivation within the framework \cite{chamberland2020very,sahay2026foldtransversal}.
Although such a basis change can map the Hadamard observable to a DCP operator, compatibility must also be established for the transformed preparations, gates, and measurements throughout the protocol.

Finally, the DCP decompositions introduced here suggest a route to circuits that leave the 3PP state class: understanding which fault-tolerant protocols admit small DCP decompositions could extend the practical scope of the simulator beyond the polynomial-time regime established in this work.

\begin{acknowledgments}
    Ș.C. is supported by a Department of Computer Science Scholarship from the University of Oxford.
    M.K. is supported by the Clarendon Fund Scholarship and the Sloane Robinson Scholarship.
    AP is supported by the Engineering and Physical Sciences Research Council (EP/S021582/1) and PASQUOPS (BMFTR).

    Ș.C. would like to thank Selwyn Simsek for inspiring conversations, especially for pointing out the fact that DCP operators are invariant under conjugation by $\CNOT+\T$ circuits.

    \vspace{\baselineskip}
    \paragraph*{AI usage.}
    AI tools were used to edit the text, assist with proofs, and write significant parts of the software package.
    The simulation framework, the equivalence between DCP stabilizer and phase-polynomial representations, and the measurement cases motivating the QEC applications were developed by the authors.
    AI tools helped complete the classification of compatible measurements and refine the proofs.
    The authors take full responsibility for the results, text, and software.
\end{acknowledgments}

\onecolumngrid
\appendix
\crefalias{section}{appendix}
\renewcommand{\theequation}{\Alph{section}.\arabic{equation}}
\numberwithin{theorem}{section}
\numberwithin{proposition}{section}
\numberwithin{lemma}{section}
\numberwithin{corollary}{section}
\numberwithin{conjecture}{section}
\numberwithin{assumption}{section}
\numberwithin{definition}{section}
\numberwithin{example}{section}
\numberwithin{notation}{section}
\numberwithin{remark}{section}
\section{Phase polynomials}\label{app:phase-polynomials}

This appendix collects the phase-polynomial facts used throughout the paper.
For $A\subseteq [n]$, write $x_A=\prod_{i\in A}x_i$, and regard the
binary inner product $\langle u,x\rangle$ as an element of
\(\{0,1\}\subseteq\Z_{2^k}\).

\begin{theorem}[Equivalent forms of phase polynomials]
    \label{thm:phase-polynomial-equivalence}
    Let $k\geq 1$.  A function $q:\F_2^n\to\Z_{2^k}$ can be written in
    the weighted form
    \begin{equation}
        \label{eq:weighted-phase-polynomial}
        q(x)=c_\varnothing+
        \sum_{\substack{\varnothing\neq A\subseteq[n]\\ |A|\leq k}}
        2^{|A|-1}c_Ax_A
        \pmod {2^k}
    \end{equation}
    if and only if it can be written as a sum of parity functions,
    \begin{equation}
        \label{eq:phase-polynomial-sum-parities}
        q(x)=a_0+
        \sum_{u\in\F_2^n\setminus\{0\}}a_u\langle u,x\rangle
        \pmod {2^k}.
    \end{equation}
    All coefficients may be taken in \(\Z_{2^k}\), although a coefficient
    $c_A$ with $A\neq\varnothing$ only matters modulo
    $2^{k-|A|+1}$.  In particular, for $k=3$, the weighted form is
    \[
        c_\varnothing+\sum_i c_i x_i
        +2\sum_{i<j}c_{ij}x_ix_j
        +4\sum_{i<j<\ell}c_{ij\ell}x_ix_jx_\ell
        \pmod 8.
    \]
\end{theorem}

\begin{proof}
    Repeatedly applying $a\oplus b=a+b-2ab$ gives, for every nonempty
    $B\subseteq[n]$,
    \begin{equation}
        \label{eq:parity-to-weighted}
        \bigoplus_{i\in B}x_i
        =\sum_{\varnothing\neq A\subseteq B}
        (-2)^{|A|-1}x_A.
    \end{equation}
    Modulo $2^k$, all terms with $|A|>k$ vanish.  Applying
    \cref{eq:parity-to-weighted} to each term in
    \cref{eq:phase-polynomial-sum-parities} therefore produces a weighted
    polynomial of the form \eqref{eq:weighted-phase-polynomial}.

    Conversely, for every nonempty $A\subseteq[n]$, Boolean
    inclusion--exclusion gives
    \begin{equation}
        \label{eq:weighted-to-parity}
        2^{|A|-1}x_A
        =\sum_{\varnothing\neq B\subseteq A}
        (-1)^{|B|-1}
        \bigoplus_{i\in B}x_i.
    \end{equation}
    Each parity on the right is
    $\langle \sum_{i\in B}e_i,x\rangle$, where $e_i$ is the $i$th
    standard basis vector.  Substituting \cref{eq:weighted-to-parity} into
    \cref{eq:weighted-phase-polynomial} yields
    \cref{eq:phase-polynomial-sum-parities}.
\end{proof}

We call the equivalent functions in
\Cref{thm:phase-polynomial-equivalence} \emph{$k$th-order phase
polynomials}.

The parity-sum form makes closure under changes of binary coordinates
immediate.

\begin{corollary}[Affine pullback]
    \label{cor:phase-polynomial-affine-pullback}
    If $q:\F_2^n\to\Z_{2^k}$ is a $k$th-order phase polynomial and
    $\phi(y)=Ay\oplus b$ is an affine map from \(\F_2^m\) to \(\F_2^n\),
    then $q\circ\phi$ is a $k$th-order phase polynomial on \(\F_2^m\).
\end{corollary}

\begin{proof}
    For $\beta\in\F_2$ and $t\in\{0,1\}$,
    \(\beta\oplus t=\beta+(1-2\beta)t\).  Hence
    \[
        \langle u,Ay\oplus b\rangle
        =\langle u,b\rangle
        +\bigl(1-2\langle u,b\rangle\bigr)
        \langle A^{\mathsf T}u,y\rangle.
    \]
    Substitution into the parity-sum form of
    \Cref{thm:phase-polynomial-equivalence} gives another constant plus a
    sum of parity functions.
\end{proof}

We next record how phase polynomials behave under discrete differentiation.
For $v\in\F_2^n$, recall that
\[
    \Delta_vq(x)=q(x\oplus v)-q(x).
\]

\begin{theorem}[Discrete derivatives of phase polynomials]
    \label{thm:phase-polynomial-derivative}
    Let $q:\F_2^n\to\Z_{2^k}$ be a $k$th-order phase polynomial, with
    $k\geq2$.  For every $v\in\F_2^n$, there are a constant
    $r_v\in\Z_{2^k}$ and a $(k-1)$st-order phase polynomial
    $q_v:\F_2^n\to\Z_{2^{k-1}}$ such that
    \begin{equation}
        \label{eq:general-phase-derivative}
        \Delta_vq(x)=r_v+2q_v(x)\pmod {2^k}.
    \end{equation}
    More generally, an $s$-fold discrete derivative, for $1\leq s<k$,
    is a constant plus $2^s$ times a $(k-s)$th-order phase polynomial;
    every $k$-fold derivative is constant.
\end{theorem}

\begin{proof}
    Write $q$ in the parity-sum form
    \eqref{eq:phase-polynomial-sum-parities}.  Since
    \[
        \langle u,x\oplus v\rangle-\langle u,x\rangle
        =\begin{cases}
            0, & \langle u,v\rangle=0,\\
            1-2\langle u,x\rangle, & \langle u,v\rangle=1,
        \end{cases}
    \]
    we obtain
    \begin{equation}
        \label{eq:derivative-parity-form}
        \Delta_vq(x)
        =\sum_{\langle u,v\rangle=1}a_u
        -2\sum_{\langle u,v\rangle=1}a_u\langle u,x\rangle
        \pmod {2^k}.
    \end{equation}
    The first sum is constant.  After reducing its coefficients modulo
    $2^{k-1}$, the second sum is a phase polynomial over
    \(\Z_{2^{k-1}}\), proving \cref{eq:general-phase-derivative}.
    Iterating this argument proves the remaining statements.
\end{proof}

For $k=3$, a discrete derivative is therefore a Clifford phase polynomial.
We make the corresponding normal forms explicit.

\begin{theorem}[Clifford phase polynomials]
    \label{thm:clifford-phase-polynomial}
    A function $d:\F_2^n\to\Z_8$ is the phase function of a diagonal
    Clifford operator, up to global phase, if and only if it has either of
    the equivalent forms
    \begin{align}
        d(x)
        &=d_0+2\sum_i \alpha_i x_i
          +4\sum_{i<j}\beta_{ij}x_ix_j \pmod 8,
          \label{eq:clifford-phase-weighted}\\
        &=d_0+2\sum_{u\in\F_2^n\setminus\{0\}}
          b_u\langle u,x\rangle \pmod 8,
          \label{eq:clifford-phase-parity}
    \end{align}
    where \(d_0\in\Z_8\), \(\alpha_i,b_u\in\Z_4\), and
    \(\beta_{ij}\in\F_2\).
\end{theorem}

\begin{proof}
    The diagonal Clifford group is generated, up to global phase, by the
    $\S_i$ and \(\CZ_{ij}\) gates.  Their phase functions are $2x_i$ and
    $4x_ix_j$, respectively, which proves
    \eqref{eq:clifford-phase-weighted}.  Conversely, every polynomial of
    that form is implemented by
    \(\prod_i \S_i^{\alpha_i}\prod_{i<j}\CZ_{ij}^{\beta_{ij}}\).

    Finally, dividing the nonconstant part of
    \eqref{eq:clifford-phase-weighted} by $2$ gives a second-order phase
    polynomial over \(\Z_4\).  Applying
    \Cref{thm:phase-polynomial-equivalence} with $k=2$ proves its
    equivalence with \eqref{eq:clifford-phase-parity}.
\end{proof}

\begin{theorem}[Discrete derivatives of Clifford phase polynomials]
    \label{thm:clifford-phase-derivative}
    Let $d:\F_2^n\to\Z_8$ be a Clifford phase polynomial.  For every
    $v\in\F_2^n$, there are $r_v\in\Z_4$ and a linear function
    \(\ell_v:\F_2^n\to\F_2\) such that
    \begin{equation}
        \label{eq:clifford-phase-derivative}
        \Delta_vd(x)=2r_v+4\ell_v(x)\pmod 8.
    \end{equation}
\end{theorem}

\begin{proof}
    Use \cref{eq:clifford-phase-parity} and write
    $d(x)=d_0+2\sum_u b_u\langle u,x\rangle$.
    \cref{eq:derivative-parity-form} gives
    \[
        \Delta_vd(x)
        =2\sum_{\langle u,v\rangle=1}b_u
        -4\sum_{\langle u,v\rangle=1}b_u\langle u,x\rangle
        \pmod 8.
    \]
    Take $r_v=\sum_{\langle u,v\rangle=1}b_u\pmod4$ and
    \[
        \ell_v(x)=
        \bigoplus_{\substack{u:\,\langle u,v\rangle=1\\ b_u\text{ odd}}}
        \langle u,x\rangle.
    \]
    The minus sign is immaterial after multiplication by $4$ modulo $8$,
    and \(\ell_v\) is linear.
\end{proof}

\section{Classification of compatible \texorpdfstring{$X$}{X} measurements}
\label{appendix:x-measurements}

This appendix proves the classification of compatible $X$ measurements and derives the corresponding outcome probabilities and post-measurement states.
The proof separates measurements that enlarge the affine support from those acting within it, whose compatibility is determined by the discrete derivative of the phase polynomial.

\begingroup
\renewcommand{\label}[1]{}
\begin{xmeasurementrestatement}
    \xMeasurementClassificationStatement
\end{xmeasurementrestatement}
\endgroup
\begin{proof}
Let $\varepsilon := (-1)^\mu$ and consider the projector $\Pi_\varepsilon = \frac{1}{2}(1 + \varepsilon X_a)$.

\begin{enumerate}

\item
First, suppose that $e_a \not\in V$.
Then, $S$ and $S \oplus e_a$ are disjoint, such that
\begin{align}
    \Pi_\varepsilon \ket{\psi}
    &= \frac{1}{2\sqrt{|S|}} \left( \sum_{x\in S} \omega_8^{q(x)}\ket{x} + \varepsilon\sum_{x\in S} \omega_8^{q(x)}\ket{x \oplus e_a} \right) \\
    &= \frac{1}{2\sqrt{|S|}} \left( \sum_{x\in S} \omega_8^{q(x)}\ket{x} + \sum_{x\in S} \omega_8^{q(x) + 4\mu}\ket{x \oplus e_a} \right) \\
    &= \frac{1}{2\sqrt{|S|}} \sum_{\substack{x\in S\\ k \in\mathbb F_2}} \omega_8^{q(x) + 4k\mu} \ket{x + ke_a},
\end{align}
which remains 3PP.
The post-measurement state is given by $S'_\mu = S \cup (S \oplus e_a)$ and the unique phase polynomial $q'_\mu$ with $q'_\mu(x) = q(x)$ and $q'_\mu(x \oplus e_a) = q(x) + 4\mu$ for all $x \in S$.
Furthermore, $Pr(\mu) = |\Pi_\varepsilon\ket{\psi}|^2 = \frac{|S'_\mu|}{4|S|}|\ket{\psi}|^2 = \frac{1}{2}$.

\item
Now, suppose $e_a \in V$.
In that case,
\begin{align}
    \Pi_\varepsilon \ket{\psi}
    &= \frac{1}{2\sqrt{|S|}} \left( \sum_{x\in S} \omega_8^{q(x)}\ket{x} + \varepsilon\sum_{x\in S} \omega_8^{q(x)}\ket{x \oplus e_a} \right) \\
    &= \frac{1}{2\sqrt{|S|}} \left( \sum_{x\in S} \omega_8^{q(x)}\ket{x} + \varepsilon\sum_{x\in S} \omega_8^{q(x \oplus e_a)}\ket{x} \right) \\
    &= \frac{1}{2\sqrt{|S|}} \sum_{x\in S} \omega_8^{q(x)} (1 + \varepsilon\omega_8^{\Delta_{e_a}q(x)}) \ket{x}  \label{eqn:projected-3pp}
\end{align}
An important property of 3PP states is that all basis states in their support have equal probability.
In other words, a necessary condition for $\Pi_\varepsilon \ket{\psi}$ to remain 3PP is that  $|\bra{x}\Pi_\varepsilon \ket{\psi}| = |\bra{y}\Pi_\varepsilon \ket{\psi}|$ for all $x,y$ with $\bra{x}\Pi_\varepsilon \ket{\psi} \neq 0$ and $\bra{y}\Pi_\varepsilon \ket{\psi} \neq 0$.
Following \cref{eqn:projected-3pp}, this is equivalent to
\[
    \left|1 + \varepsilon\omega_8^{\Delta_{e_a}q(x)}\right| = \left|1 + \varepsilon\omega_8^{\Delta_{e_a}q(y)}\right|
\]
for all $x,y$ with $1 + \varepsilon\omega_8^{\Delta_{e_a}q(x)} \neq 0$ and $1 + \varepsilon\omega_8^{\Delta_{e_a}q(y)} \neq 0$.
Computing those magnitudes, we get:
\begin{align*}
    \left| 1 + \varepsilon\omega_8^{\Delta_{e_a}q(x)} \right| &=
    \left| \omega_8^{\Delta_{e_a} q(x) / 2} \left( \omega_8^{-\Delta_{e_a} q(x) / 2} + \varepsilon\omega_8^{\Delta_{e_a} q(x) / 2} \right) \right| \\
    &= \begin{cases}
        2 \left| \cos \left(\dfrac{\Delta_{e_a} q(x) \pi}{8} \right) \right| & \text{if } \varepsilon = +1, \\
        2 \left| \sin \left(\dfrac{\Delta_{e_a} q(x) \pi}{8} \right) \right| & \text{if } \varepsilon = -1.
    \end{cases}
\end{align*}
For $\varepsilon = +1$, this means that the image of $\Delta_{e_a}q$ must be a subset of one of $\{0, 4\}, \{2,4,6\}, \{1,4,7\}$, or $\{3,4,5\}$.
Similarly, for $\varepsilon = -1$ its image must be a subset of one of $\{0, 4\}, \{0,2,6\}, \{0,1,7\}$, or $\{0,3,5\}$.
Since both conditions must hold simultaneously, compatibility requires
\begin{align} \label{eqn:gradient-two-valued}
    \operatorname{Im}(\Delta_{e_a}q)\subseteq D
    \quad\text{for some}\quad
    D\in\bigl\{\{0,4\},\{2,6\},\{1,7\},\{3,5\}\bigr\}.
\end{align}
We will discuss these cases separately:

\begin{enumerate}

\item
Suppose $\Delta_{e_a} q(x) \in \{ 0, 4 \}$ for all $x \in S$, i.e. $\Delta_{e_a} q(x) = 4f(x)$ for some function $f : S \to \mathbb F_2$.
Then,
\[
    1 + \varepsilon\omega_8^{E(x)}
    = \begin{cases}
        2 & \text{if } f(x) = \mu \\
        0 & \text{otherwise}.
    \end{cases}
\]
Inserting this into \cref{eqn:projected-3pp}, we see that basis states with $f(x) \neq \mu$ are removed from the support while the others are scaled by 2, giving us
\[ \Pi_\varepsilon\ket{\psi} = \frac{1}{\sqrt{|S|}} \sum_{\substack{x \in S \\ f(x) = \mu}} \omega_8^{q(x)} \ket{x}. \]
This state remains 3PP for both outcomes $\mu = 0,1$ simultaneously iff $f$ is affine.
In this case, we can read off that $S'_\mu = S \cap \{ x ~|~ f(x) = \mu \}$ and $q'_\mu = q$.
Furthermore, $Pr(\mu) = |\Pi_\varepsilon\ket{\psi}|^2 = \frac{|S'_\mu|}{|S|}|\ket{\psi}|^2 = \frac{|S'_\mu|}{|S|}$.
Since $S$ and $S'_\mu$ are affine spaces differing by a single constraint, this ratio must be either 0, 1, or $\frac{1}{2}$ depending on whether $f(x) = \mu$ holds for none, all, or some $x \in S$.

\item
Suppose $\Delta_{e_a}q(x) \in \{2, 6\}$ for all $x \in S$, i.e. $\Delta_{e_a}q(x) = 2 + 4g(x) \pmod 8$ for some function $g: S \to \mathbb F_2$.
Then,
\[
    1 + \varepsilon\omega_8^{\Delta_{e_a}q(x)}
    = 1 + i\varepsilon(-1)^{g(x)}
    = \sqrt 2 \omega_8^{\varepsilon(-1)^{g(x)}}
    = \sqrt 2 \omega_8^{\varepsilon(1 - 2g(x))},
\]
where we used the identities $1+s i = \sqrt{2} e^{is\pi/4}$ for $s \in \{-1,1\}$ and  $(-1)^m = 1 - 2m$ for $m \in \{0,1\}$.
Inserting this into \cref{eqn:projected-3pp}, we get the projected state
\[ \Pi_\varepsilon \ket{\psi} = \frac{\sqrt 2 \omega_8^\varepsilon}{2\sqrt{|S|}} \sum_{x \in S} \omega_8^{q(x) - 2\varepsilon g(x)} \ket{x} \]
which remains 3PP iff $2g$ is a valid third-order phase polynomial.
In this case, we can read off that $S' = S$ and $q'(x) = q(x) - 2\varepsilon g(x)$.
Furthermore, $Pr(\mu) = |\Pi_\varepsilon\ket{\psi}|^2 = \frac{1}{2}|\ket{\psi}|^2 = \frac{1}{2}$.

\item
Suppose $\Delta_{e_a}q(x) \in \{ d, -d \}$ for all $x \in S$ where $d \in \{ 1, 3, 5, 7 \}$, i.e. $\Delta_{e_a}q(x) = (-1)^{h(x)}d \pmod 8$ for some function $h: S \to \mathbb F_2$.
Then,
\begin{align*}
    1 + \varepsilon\omega_8^{\Delta_{e_a}q(x)}
    &= (1 + \varepsilon\omega_8^d) \left(\frac{1 + \varepsilon\omega_8^{(-1)^{h(x)} d}}{1 + \varepsilon\omega_8^{(-1)^{h(x) + 1} d}}\right)^{h(x)} \\
    &= (1 + \varepsilon\omega_8^d) (\varepsilon\omega_8^{-d})^{h(x)} \\
    &= (1 + \varepsilon\omega_8^d) \omega_8^{(4\mu - d)h(x)}.
\end{align*}
Inserting this into \cref{eqn:projected-3pp}, we get the projected state
\[ \Pi_\varepsilon \ket{\psi} = \frac{1 + \varepsilon\omega_8^d}{2\sqrt{|S|}} \sum_{x \in S} \omega_8^{q(x) + (4\mu - d)h(x)} \ket{x}. \]
We claim that this state remains 3PP iff $h$ is affine over $\mathbb F_2$.
First, suppose $h$ is affine over $\mathbb F_2$, i.e. $h(x) = a_0 \oplus \left(\bigoplus_i a_i x_i\right)$ for some $a_0,a_i \in \mathbb F_2$. Writing $a=(a_i)_i$, we have
\[
    h(x) = a_0 \oplus \langle a,x\rangle
    = a_0 + (1-2a_0)\langle a,x\rangle \pmod 8.
\]
Thus, $h$ has the form of \cref{eq:phase-polynomial-sum-parities}, and hence $q'_\mu=q+(4\mu-d)h$ is a third-order phase polynomial.
Conversely, note that
\[ h(x) = (4\mu - d)^{-1}(q'_\mu(x) - q(x)) \pmod 2. \]
Thus, if $q'_\mu$ is a 3PP phase polynomial, then $(4\mu - d)^{-1}(q'_\mu - q)$ must also be one.
Reducing mod 2 removes all quadratic and cubic terms, showing that $h$ is affine over $\mathbb F_2$.
If the measurement is compatible, we can furthermore read off that $S'_\mu = S$ and
\[ Pr(\mu) = |\Pi_\varepsilon\ket{\psi}|^2 = \left|\frac{1 + \varepsilon\omega_8^d}{2}\right|^2 |\ket{\psi}|^2 = \frac{1}{2}(1 + \varepsilon\cos(d\pi/4)). \]

\end{enumerate}

In any other case, the measurement cannot be compatible as established in \cref{eqn:gradient-two-valued}. \qedhere

\end{enumerate}
\end{proof}

The measurement classification in \cref{thm:x-measurements} depends on the derivative $\Delta_{e_a}q$ of the phase polynomial in the direction of the measured qubit.
Remember that, by \cref{prop:derivative-from-E},  we can directly construct this derivative from the DCP stabilizer group of $\ket{\psi}$ without having to explicitly track the phase polynomial itself.

More precisely, let $\ket{\psi}$ be a 3PP state with phase polynomial $q$ and DCP stabilizer group $G$. Suppose $e_a \in X(G)$ and choose an element $g \in G$ with $\pi(g) = e_a$.

Write $g=X_a\delta$, where $\delta=\omega_8^kD(s,M)$ is diagonal.
The stabilizer equation $g\ket{\psi}=\ket{\psi}$ gives
\begin{align}
    \Delta_{e_a}q(x)
    =k+2\sum_i s_i x_i+4\sum_{i<j}M_{ij}x_ix_j
    \pmod 8
\end{align}
for every $x\in S$.
Thus, we can extract $\Delta_{e_a} q$ directly from the normal form of the DCP stabilizer element $g$.
Furthermore, given $\Delta_{e_a} q$, we can distinguish the different cases of \cref{thm:x-measurements} efficiently.
To this end, Let $S=x_0\oplus V$, and let $v^{(1)},\ldots,v^{(r)}$ be a basis of $V$.
Define the restricted derivative
\begin{align}
    \widetilde D(c)
    :=\Delta_{e_a}q\left(x_0\oplus\bigoplus_{i=1}^r c_i v^{(i)}\right),
    \qquad c\in\mathbb F_2^r.
\end{align}
By affine pullback, this is a Clifford phase polynomial, which we write as
\begin{align}
    \widetilde D(c)
    =\tilde k+2\sum_i\tilde s_i c_i
    +4\sum_{i<j}\tilde M_{ij}c_ic_j
    \pmod 8,
\end{align}
where $\tilde k\in\mathbb Z_8$, $\tilde s_i\in\mathbb Z_4$, and $\tilde M$ is symmetric over $\mathbb F_2$ with zero diagonal.

\begin{proposition}\label{prop:x-measurement-case-decision}
    The different cases of \cref{thm:x-measurements} can be decided based on $\tilde k$, $\tilde s$, and $\tilde M$:
    \begin{enumerate}
        \item
        Case \ref{thm:x-measurements:even-0-4} applies iff $\tilde k \in \{0, 4\}$, $\tilde s_i \in \{0, 2\}$, and $\tilde M_{ij} = 0$.

        \item
        Case \ref{thm:x-measurements:even-2-6} applies iff $\tilde k \in \{2, 6\}$, $\tilde s_i \in \{0, 2\}$, and $\rank \tilde M \leq 2$.

        \item
        Case \ref{thm:x-measurements:odd} applies iff $\tilde k$ is odd,
        $\tilde s_i\in\{0,-\tilde k\}\pmod 4$ for every $i$, and
        $\tilde M_{ij}=\tilde s_i\tilde s_j\pmod 2$ for every $i<j$.
    \end{enumerate}
\end{proposition}
\begin{proof}
We begin by giving explicit representations for $\tilde k, \tilde s_i$, and  $\tilde M_{ij}$ in terms of $\Delta_{e_a} q$.
Concretely, noting that $\widetilde D(0) = \tilde k$, $\widetilde D(e_i) = \tilde k + 2\tilde s_i$, and $\widetilde D(e_i \oplus e_j) = \tilde k + 2\tilde s_i + 2\tilde s_j + 4\tilde M_{ij}$, we have
\begin{align}
    \tilde k &= \Delta_{e_a}q(x_0) \pmod 8 \label{eqn:k-tilde} \\
    2\tilde s_i &= \Delta_{e_a}q(x_0 \oplus v^{(i)}) - \Delta_{e_a}q(x_0) \pmod 8 \label{eqn:s-tilde} \\
    4\tilde M_{ij} &= \Delta_{e_a}q(x_0 \oplus v^{(i)} \oplus v^{(j)}) - \Delta_{e_a}q(x_0 \oplus v^{(i)}) - \Delta_{e_a}q(x_0 \oplus v^{(j)}) + \Delta_{e_a}q(x_0) \pmod 8 \label{eqn:M-tilde}.
\end{align}
In other words $2\tilde s_i = \Delta_{v^{(i)}}\Delta_{e_a}q(x_0)$ and $4\tilde M_{ij} = \Delta_{v^{(i)}}\Delta_{v^{(j)}}\Delta_{e_a}q(x_0)$ are the first and second discrete derivatives of $\Delta_{e_a}q$ at $x_0$ in the directions of $V$'s basis vectors.
Furthermore, we note the following facts about arithmetic over Boolean variables $a_i \in \mathbb F_2$ modulo 8:
\begin{align}
    2\bigoplus a_i &= 2\sum_i a_i - 4\sum_{i<j} a_ia_j \pmod 8 \label{eqn:2-xor} \\
    4\bigoplus a_i &= 4\sum_{i} a_i \pmod 8 \label{eqn:4-xor}
\end{align}
We consider each case separately:
\begin{enumerate}
\item
First, suppose that $\Delta_{e_a}\tilde q(c) = 4f(c) \pmod 8$ for some affine function $f: \mathbb F_2^r \to \mathbb F_2$. Write $f(c) = \bigoplus_i f_ic_i \oplus f_0$.
Using \cref{eqn:4-xor}, we see that $\Delta_{e_a}\tilde q(c) = 4f_0 + 4\sum_i f_ic_i \pmod 8$.
We can read off that $\tilde k := 4f_0 \in \{0, 4\}$, $2\tilde s_i := 4f_i$ is even and $\tilde M_{ij} = 0$ since there are no quadratic terms.

Conversely, suppose $\tilde k \in \{0,4\}$, $\tilde s_i \in \{0,2\}$, and $\tilde M_{ij} = 0$.
Then $f(c) := \left(\bigoplus_i \frac{s_i}{2}c_i\right) \oplus \frac{k}{4}$ is affine and $\widetilde D(c) = 4f(c)$ such that case \ref{thm:x-measurements:even-0-4} applies.

\item
Suppose $\Delta_{e_a}\tilde q(c) = 2 + 4g(c) \pmod 8$ for some function $g: \mathbb F_2^r \to \mathbb F_2$ such that $2g$ is a 3PP phase polynomial.
Evaluating equations \eqref{eqn:k-tilde} and \eqref{eqn:s-tilde}, we get
\begin{align*}
    \tilde k &= \Delta_{e_a}\tilde q(0) = 2 + 4g(0) \pmod 8 \\
    2\tilde s_i &= \Delta_{e_a}\tilde q(e_i) - \Delta_{e_a}\tilde q(0) = 4(g(e_i) - g(0)) \pmod 8
\end{align*}
showing that $\tilde k \in \{2,6\}$ and that all $\tilde s_i$ are even.
Then
\begin{align*}
    4g(c)
    = \Delta_{e_a}\tilde q(c) - 2
    = \tilde k - 2 + 2\sum_i \tilde s_ic_i + 4\sum_{i<j} \tilde M_{ij}c_ic_j \pmod 8
\end{align*}
Using \cref{eqn:4-xor}, we conclude that
\begin{equation} \label{eqn:g-quadratic}
    g(c) = \frac{\tilde k - 2}{4} \oplus \left(\bigoplus_i \frac{\tilde s_i}{2}c_i\right) \oplus \left(\bigoplus_{i<j} \tilde M_{ij}c_ic_j\right).
\end{equation}
In other words, $g$ must be a quadratic Boolean function.
Therefore, $\rank \tilde M \leq 2$ by Lemma~\ref{lem:3pp-rank}.

Conversely, suppose $\tilde k \in \{2, 6\}$, $\tilde s_i \in \{0, 2\}$, and $\rank\tilde M \leq 2$.
Define $g$ by \cref{eqn:g-quadratic} and use \cref{eqn:4-xor} to show
\begin{align}
    2+4g(c)
    =\tilde k+2\sum_i\tilde s_i c_i
    +4\sum_{i<j}\tilde M_{ij}c_ic_j
    =\widetilde D(c)
    \pmod 8.
\end{align}
By Lemma~\ref{lem:3pp-rank}, $2g$ is a 3PP phase polynomial such that case \ref{thm:x-measurements:even-2-6} applies.

\item
Suppose first that case \ref{thm:x-measurements:odd} applies.
By absorbing the constant term of the affine function into the sign of $d$, we may write
\begin{align}
    \widetilde D(c)=\tilde k(-1)^{h(c)},
    \qquad
    h(c)=\bigoplus_i h_i c_i,
\end{align}
where $\tilde k$ is odd and $h_i\in\mathbb F_2$.
Expanding the parity modulo $8$ gives
\begin{align}
    \widetilde D(c)
    &=\tilde k-2\tilde k\sum_i h_i c_i
    +4\tilde k\sum_{i<j}h_ih_jc_ic_j
    \pmod 8.
\end{align}
Comparing coefficients yields
\begin{align}
    \tilde s_i&=-\tilde k h_i\pmod 4,\\
    \tilde M_{ij}&=h_ih_j\pmod 2.
\end{align}
Since $\tilde k$ is odd, we also have $h_i=\tilde s_i\pmod 2$, proving the stated conditions.

Conversely, suppose those conditions hold.
Define $h_i:=\tilde s_i\pmod 2\in\mathbb F_2$ and $h(c):=\bigoplus_i h_ic_i$.
The condition $\tilde s_i\in\{0,-\tilde k\}\pmod 4$ implies $\tilde s_i=-\tilde k h_i\pmod 4$, while the quadratic condition gives $\tilde M_{ij}=h_ih_j\pmod 2$.
Therefore,
\begin{align}
    \tilde k(-1)^{h(c)}
    &=\tilde k-2\tilde k\sum_i h_i c_i
    +4\tilde k\sum_{i<j}h_ih_jc_ic_j\\
    &=\tilde k+2\sum_i\tilde s_i c_i
    +4\sum_{i<j}\tilde M_{ij}c_ic_j\\
    &=\widetilde D(c)\pmod 8.
\end{align}
Thus case \ref{thm:x-measurements:odd} applies.
\qedhere
\end{enumerate}
\end{proof}

\begin{lemma} \label{lem:3pp-rank}
    Let $g: \mathbb F_2^n \to \mathbb F_2$ be a quadratic Boolean function given by a bilinear form
    \[ g(x) = g_0 \oplus \left( \bigoplus_i a_i x_i \right) \oplus \left( \bigoplus_{i<j} M_{ij}x_ix_j \right) \]
    with $g_0 \in \mathbb F_2, a\in\mathbb F_2^n, M \in \mathbb F_2^{n\times n}$ where $M$ is symmetric with zero diagonal.
    Then $2g$ is a 3PP phase polynomial iff $\rank M \leq 2$.
\end{lemma}
\begin{proof}
We first put the quadratic part into canonical coordinates, then determine when the weighted polynomial for $2g$ contains quartic terms.
Since $M$ is symmetric with zero diagonal, its rank is even, say $\operatorname{rank}M=2t$.
An invertible linear change of coordinates puts $g$ into the form
\begin{align}
    \hat g(z)
    =\ell(z)\oplus\bigoplus_{i=1}^t z_{2i-1}z_{2i},
    \qquad z\in\mathbb F_2^n,
\end{align}
where $\ell$ is affine and may depend on all $n$ variables.
By affine pullback in both directions, $2g$ is a third-order phase polynomial iff $2\hat g$ is one.

Expanding the XOR modulo $8$ gives
\begin{align}
    2\hat g(z)
    &=2\ell(z)+2\sum_{i=1}^t z_{2i-1}z_{2i}
    -4\sum_{i=1}^t\ell(z)z_{2i-1}z_{2i}\\
    &\quad-4\sum_{i<j}z_{2i-1}z_{2i}z_{2j-1}z_{2j}
    \pmod 8.
\end{align}
The first term is a third-order phase polynomial because $\ell$ is affine.
The second term has quadratic coefficients divisible by $2$.
In the third term, the prefactor $4$ allows us to expand $\ell$ over $\mathbb F_2$, so this term has degree at most three and all coefficients are divisible by $4$.
Thus, when $t\leq 1$, every term has the required weighted form.

If $t\geq 2$, the last sum contains $z_1z_2z_3z_4$ with coefficient $4$ modulo $8$.
None of the other terms has degree four, so this coefficient cannot cancel.
By uniqueness of the multilinear polynomial representation modulo $8$, $2\hat g$ is then not a third-order phase polynomial.
Consequently, $2g$ is a third-order phase polynomial iff $t\leq 1$, equivalently $\operatorname{rank}M\leq 2$.
\end{proof}

\section{Efficient simulation in the phase polynomial representation}
\label{appendix:efficient-simulation}

This appendix establishes the runtime and memory bounds for simulation in the phase-polynomial representation.
We analyze the cost of state updates and compatibility tests, then show how to compute amplitudes and Pauli expectation values, sample computational-basis outcomes, and test equality of represented states.

\begin{efficientsimulationrestatement}[Efficient simulation of 3PP circuits]
\efficientSimulationStatement{}
\end{efficientsimulationrestatement}

\begin{proof}
We first establish the cost of updating the state representation and checking measurement compatibility, then prove the bounds for queries on represented states.
Throughout, we use dense coefficient arrays and ordinary Gaussian elimination, with arithmetic on phase coefficients performed modulo $8$.

The affine support is specified by $x_0$ and $B$, requiring $O(n^2)$ bits.
The phase polynomial has the weighted form
\begin{align}
    p(y)
    = d+\sum_i a_i y_i
    +2\sum_{i<j}b_{ij}y_i y_j
    +4\sum_{i<j<k}c_{ijk}y_i y_j y_k,
\end{align}
where $d,a_i\in\mathbb{Z}_8$, $b_{ij}\in\mathbb{Z}_4$, and $c_{ijk}\in\mathbb{F}_2$.
Since $r\leq n$, storing these coefficients requires $O(n^3)$ bits.

We will repeatedly use two elementary facts about polynomial updates.
First, each substitution $y_i\mapsto y_i\oplus 1$ or $y_i\mapsto y_i\oplus y_j$ can be performed in $O(r^2)$ time, since only the $O(r^2)$ monomials containing $y_i$ are affected, and each produces a bounded number of terms.
Second, an affine Boolean function can be expanded as a weighted polynomial modulo $8$ in $O(r^3)$ time.
Indeed, for a parity of variables,
\begin{align}
    \bigoplus_{i\in I}y_i
    =\sum_{i\in I}y_i
    -2\sum_{\substack{i<j\\i,j\in I}}y_i y_j
    +4\sum_{\substack{i<j<k\\i,j,k\in I}}y_i y_j y_k
    \pmod 8,
\end{align}
and an affine constant is incorporated using $1\oplus f=1-f$.

We now consider circuit evolution.
Preparations append a physical coordinate and, for a $\ket{+}$ preparation, a free support coordinate.
An $X$ gate changes $x_0$, while a $\CNOT$ gate applies the corresponding row operation to $x_0$ and $B$.
Neither requires changing $p$.
A $\T$ or $Z$ gate adds an integer multiple of the affine Boolean function describing the affected physical bit to $p$.
These updates, and hence arbitrary Pauli updates, take at most $O(n^3)$ time.

For a $Z$-type measurement, write the measured operator as $Z(v)$.
Its value on a support point is determined by the affine function
\begin{align}
    f(y)=\langle v,x_0\oplus By\rangle.
\end{align}
If $f$ is constant, the outcome is deterministic and the state is unchanged.
Otherwise, each outcome has probability $1/2$ and imposes one affine equation on $y$.
Choose a variable with nonzero coefficient in that equation and use $O(r)$ elementary affine substitutions to make the equation a constraint on that variable alone.
Applying these substitutions to $p$ costs $O(r^3)$, and updating the support parametrization costs at most $O(nr^2)$.
Fixing the constrained variable and removing it therefore gives the conditional-state representation in $O(n^3)$ time.

A measurement of a nonidentity $X$-type Pauli string of weight $w$ reduces to a single-qubit measurement by \cref{eq:x-measurement-reduction}, using at most $2(w-1)$ CNOTs in total.
Each CNOT updates rows of $x_0$ and $B$ in $O(n)$ time and leaves $p$ unchanged, so the reduction adds only $O(n^2)$ time.
The $O(n^3)$ expected cost established below therefore applies to one original Pauli-string measurement, without changing the operation count $m$ or the total runtime.
The identity measurement is deterministic and leaves the state unchanged.

Next, consider a proposed measurement of $X_a$.
Gaussian elimination decides whether $e_a\in\operatorname{im}B$ in $O(n^3)$ time.
If not, the measurement is compatible case~(1).
For outcome $\mu$, append $e_a$ as a new column of $B$ and replace $p$ by
\begin{align}
    p'(y,z)=p(y)+4\mu z.
\end{align}
Both outcomes have probability $1/2$, and the update takes at most $O(n^3)$ time.

If $e_a\in\operatorname{im}B$, solve $Bt=e_a$ and construct
\begin{align}
    D(y)=p(y\oplus t)-p(y).
\end{align}
The translation $y\mapsto y\oplus t$ uses at most $r$ elementary substitutions, so constructing $D$ costs $O(r^3)$.
By the discrete-derivative property, $D$ is a Clifford phase polynomial and therefore has only $O(r^2)$ coefficients.

We now check the three internal-measurement cases and bound their updates.
For case~(2a), coefficient inspection determines whether $D=4f$ for an affine Boolean function $f$.
Conditioning on $f(y)=\mu$ is then the same affine restriction used for $Z$ measurements, costing $O(n^3)$ time.

For case~(2b), coefficient inspection first determines whether $D=2+4g$ for a quadratic Boolean function $g$.
The compatibility criterion is that the alternating matrix of its quadratic part has rank at most two, which can be checked in $O(r^3)$ time.
When this condition holds, binary elimination produces affine $\ell$ and linear $u,v$ such that
\begin{align}
    g(y)=\ell(y)\oplus u(y)v(y).
\end{align}
Here one may take $u=v=0$ when the quadratic part vanishes.
Using
\begin{align}
    2g=2\ell+2uv-4\ell uv\pmod 8,
\end{align}
we can construct the weighted coefficients of $2g$ in $O(r^3)$ time.
For the term $2uv$, only the linear and quadratic parts of the parity expansions of $u$ and $v$ contribute modulo $8$.
For the term $4\ell uv$, it suffices to multiply the affine forms over $\mathbb{F}_2$.
Thus the conditional phase
\begin{align}
    p'=p-2(-1)^\mu g
\end{align}
can be constructed within the same bound.

For case~(2c), let $d=D(0)$.
If $d$ is odd, the linear coefficients of $D$ determine the only possible affine Boolean function $h$ with $h(0)=0$ satisfying
\begin{align}
    D(y)=d(-1)^{h(y)}.
\end{align}
Comparing coefficients verifies this identity in polynomial time bounded by $O(r^3)$.
If it holds, expanding $h$ and forming
\begin{align}
    p'=p+(4\mu-d)h
\end{align}
also costs $O(r^3)$.
If none of these cases applies, the measurement is incompatible.
Consequently, compatibility checking and every compatible conditional-state update take $O(n^3)$ time.

It remains to account for sampling measurement outcomes.
The possible probabilities belong to the fixed set
\begin{align}
    \left\{0,\frac12,1,\frac{2-\sqrt2}{4},
    \frac{2+\sqrt2}{4}\right\}.
\end{align}
The rational probabilities are sampled directly using fair random bits.
For either irrational probability, compare a uniformly random binary expansion with successively refined rational bounds on that probability.
After $k$ bits and bounds of precision $O(2^{-k})$, the probability that the comparison remains unresolved is $O(2^{-k})$.
Since the bounds can be computed in time polynomial in $k$, this procedure has constant expected runtime.
Combining outcome sampling with the updates above gives $O(n^3)$ expected time per operation and $O(mn^3)+T_{\mathrm{classical}}$ expected time for the circuit.

We next establish the costs of queries on represented states.
For a computational-basis amplitude, Gaussian elimination decides whether
\begin{align}
    By=x\oplus x_0
\end{align}
has a solution in $O(n^3)$ time.
Since $B$ has full column rank, any solution is unique, and
\begin{align}
    \langle x|\psi\rangle
    =
    \begin{cases}
        2^{-r/2}\omega_8^{p(y)}, & By=x\oplus x_0,\\
        0, & x\notin x_0\oplus\operatorname{im}B.
    \end{cases}
\end{align}
Evaluating $p(y)$ costs $O(r^3)$, proving the amplitude bound.
For computational-basis sampling, choose $y$ uniformly in $\mathbb{F}_2^r$ and return $x_0\oplus By$.
This costs $O(nr)=O(n^2)$ time.

To compute a Pauli expectation, write
\begin{align}
    P=i^\kappa X(a)Z(b).
\end{align}
If $a\notin\operatorname{im}B$, the translated support is disjoint from the original support and the expectation is zero.
Otherwise, solve $Bt=a$ and expand
\begin{align}
    \bra{\psi}P\ket{\psi}
    =i^\kappa 2^{-r}
    \sum_{y\in\mathbb{F}_2^r}
    \omega_8^{p(y)-p(y\oplus t)
    +4\langle b,x_0\oplus By\rangle}.
    \label{eq:pauli-expectation-gauss-sum}
\end{align}
Constructing the exponent costs $O(n^3)$ by the translation procedure above.
It is a Clifford phase polynomial, so the remaining sum is a quadratic Gauss sum over $\mathbb{Z}_4$.
Reduction of the quadratic form by binary elimination evaluates this sum exactly in $O(r^3)$ time, giving the claimed $O(n^3)$ bound.
The result can be retained exactly using a root of unity and a power of $\sqrt{2}$, including the possibility of zero.

Finally, consider equality of two represented states.
Gaussian elimination checks whether their affine supports coincide in $O(n^3)$ time.
If they do not, the states are unequal.
If they do, their support coordinates are related by an invertible affine transformation, which can also be found in $O(n^3)$ time.
Decompose this transformation into $O(r^2)$ elementary affine substitutions and apply them to one phase polynomial.
Since each substitution costs $O(r^2)$, expressing both phases in common coordinates costs $O(n^4)$ time.

In these common coordinates, the states are equal up to global phase exactly when their phase difference is constant modulo $8$.
Uniqueness of the multilinear polynomial representation over $\mathbb{Z}_8$ makes this equivalent to the vanishing of every nonconstant coefficient of the difference.
Checking these coefficients costs $O(r^3)$, completing the $O(n^4)$ equality bound.
\end{proof}

\section{Algebra of DCP operators} \label{appendix:dcp-algebra}

This appendix develops the algebraic tools used to represent and manipulate DCP operators and subgroups.
We first prove the normal form and derive explicit formulas for group operations, then describe the Howell normal form and canonical remainders used to construct canonical subgroup representations.

\subsection{Normal form and group operations}

\begin{dcpnormalformrestatement}[Normal form for DCP operators]
\dcpNormalFormStatement
\end{dcpnormalformrestatement}
\begin{proof}
    We first prove existence by moving all Pauli $X$ factors to
    the right, then prove uniqueness by comparing the action on
    computational-basis states.

    Every DCP operator can be written as a product of scalar
    phases, $\S_i$, $\CZ_{ij}$, and $X_i$, since
    $Z_i=\S_i^2$ and $Y_i=iX_iZ_i$. The relations
    \[
    X_i\S_i=i\S_i^{-1}X_i,
    \qquad
    X_i\CZ_{ij}=Z_j\CZ_{ij}X_i
    \]
    allow us to move each $X_i$ past every diagonal generator;
    generators acting on disjoint qubits commute. The additional
    factors introduced are diagonal, with $Z_j=\S_j^2$, and all
    scalar phases are powers of $\omega_8$. Thus
    \[
    U=\omega_8^k
    \left(\prod_i \S_i^{s_i}\right)
    \left(\prod_{i<j}\CZ_{ij}^{M_{ij}}\right)
    X(\vec{x}).
    \]
    Since $\S_i^4=I$, $\CZ_{ij}^2=I$, and $X_i^2=I$, we may
    take $k\in\Z_8$, $\vec{s}\in\Z_4^n$, and
    $\vec{x}\in\F_2^n$. The exponents $M_{ij}$ uniquely specify
    a symmetric binary matrix with zero diagonal. This proves
    existence.

    To prove uniqueness, suppose
    \[
    \omega_8^k D(\vec{s},M)X(\vec{x})
    =
    \omega_8^{k'}D(\vec{s}',M')X(\vec{x}').
    \]
    Applied to any computational-basis state $\ket{z}$, the two
    sides produce nonzero multiples of $\ket{z\oplus\vec{x}}$
    and $\ket{z\oplus\vec{x}'}$, respectively. Hence
    $\vec{x}=\vec{x}'$, and cancelling the common $X$ factor gives
    \[
    \omega_8^k D(\vec{s},M)
    =
    \omega_8^{k'}D(\vec{s}',M').
    \]
    Both diagonal operators fix $\ket{0^n}$, so evaluating this
    identity on $\ket{0^n}$ gives $k=k'$ in $\Z_8$.
    For the standard basis vector $e_i$,
    \[
    D(\vec{s},M)\ket{e_i}=i^{s_i}\ket{e_i},
    \]
    so $s_i=s_i'$ in $\Z_4$ for every $i$. Finally, for $i<j$,
    \[
    D(\vec{s},M)\ket{e_i\oplus e_j}
    =
    i^{s_i+s_j}(-1)^{M_{ij}}\ket{e_i\oplus e_j}.
    \]
    Since $\vec{s}=\vec{s}'$, this implies $M_{ij}=M_{ij}'$
    for every $i<j$, and therefore $M=M'$. All four parameters
    are unique.
\end{proof}

\begin{proposition}[Group operations on DCP operators]\label{prop:dcp-group-ops}
For $\vec{s}\in\Z_4^{\,n}$ and $M\in\F_2^{\,n\times n}$ symmetric with zero diagonal, write
\begin{align} \label{eq:phi-s-M}
    \varphi_{\vec{s},M}(\vec{v}) \;:=\; \vec{s}\cdot \vec{v} \,+\, 2 \sum_{a<b} M_{ab}\, v_a v_b \pmod 4, \qquad \vec{v} \in \F_2^{\,n},
\end{align}
so that $D(\vec{s},M)\ket{\vec{v}} = i^{\varphi_{\vec{s},M}(\vec{v})}\ket{\vec{v}}$, and let $\odot$ denote the entrywise product of vectors. Given two DCP operators $g_1 = \omega_8^{k_1} D(\vec{s}_1,M_1) X(\vec{x}_1)$ and $g_2 = \omega_8^{k_2} D(\vec{s}_2,M_2) X(\vec{x}_2)$, we have the following formulas for their product, inverse, commutator, and conjugation:
\begin{itemize}
    \item \textbf{Product:}
    \[
    g_1 g_2 \;=\; \omega_8^{\,k_1+k_2+2\varphi_{\vec{s}_2,M_2}(\vec{x}_1)}\, D\big(\vec{s}_1+\vec{s}_2\odot(1-2\vec{x}_1)+2M_2\vec{x}_1,\ M_1+M_2\big)\, X(\vec{x}_1\oplus \vec{x}_2).
    \]
    \item \textbf{Inverse:}
    \[
    g_1^{-1} \;=\; \omega_8^{\,-k_1+2\varphi_{-\vec{s}_1,M_1}(\vec{x}_1)}\, D\big((-\vec{s}_1)\odot(1-2\vec{x}_1)+2M_1\vec{x}_1,\ M_1\big)\, X(\vec{x}_1).
    \]
    \item \textbf{Commutator:}
    \[
    \gcomm{g_1}{g_2} \;:=\; g_1 g_2 g_1^{-1} g_2^{-1} \;=\; i^{\,c}\, Z(\vec{z}),
    \qquad
    c \;=\; \varphi_{\vec{s}_2,M_2}(\vec{x}_1) - \varphi_{\vec{s}_1,M_1}(\vec{x}_2) \pmod 4,
    \]
    \[
    \vec{z} \;=\; M_2\vec{x}_1 + \vec{s}_2\odot \vec{x}_1 + M_1\vec{x}_2 + \vec{s}_1\odot \vec{x}_2 \pmod 2,
    \]
    where $Z(\vec{z}) := Z_1^{z_1} \otimes \dots \otimes Z_n^{z_n}$
    \item \textbf{Conjugation:}
    \[
    g_1 g_2 g_1^{-1} \;=\; \omega_8^{\,k_2+2c}\, D(\vec{s}_2+2\vec{z},\ M_2)\, X(\vec{x}_2),
    \]
    with $(c,\vec{z})$ as in the commutator above.
\end{itemize}
\end{proposition}

\begin{proof}
We first compute how conjugation by $X(\vec{x})$ acts on
a diagonal DCP operator. We then derive the product and
inverse formulas, and obtain the commutator by comparing
the products in the two possible orders.

For binary vectors $\vec{v},\vec{x}$, expanding
$v_a\oplus x_a=v_a+x_a-2v_ax_a$ gives
\begin{align}
    \varphi_{\vec{s},M}(\vec{v}\oplus\vec{x})
    &=
    \varphi_{\vec{s},M}(\vec{x})
    +\varphi_{\vec{t},M}(\vec{v})
    \pmod 4,
    \label{eq:dcp-phase-translation}
    \\
    \vec{t}
    &=
    \vec{s}\odot(1-2\vec{x})+2M\vec{x}
    \pmod 4.
\end{align}
Indeed, the linear terms contribute
$s_a x_a+s_a(1-2x_a)v_a$, while each quadratic term
satisfies
\[
    2M_{ab}(v_a\oplus x_a)(v_b\oplus x_b)
    =
    2M_{ab}
    (v_av_b+v_ax_b+x_av_b+x_ax_b)
    \pmod 4.
\]
Collecting the constant, linear, and quadratic terms gives
\cref{eq:dcp-phase-translation}. Consequently,
\begin{equation}
    X(\vec{x})D(\vec{s},M)X(\vec{x})
    =
    i^{\varphi_{\vec{s},M}(\vec{x})}
    D\big(
        \vec{s}\odot(1-2\vec{x})+2M\vec{x},M
    \big).
    \label{eq:dcp-diagonal-conjugation}
\end{equation}

For the product formula, write
\begin{align*}
    g_1g_2
    &=
    \omega_8^{k_1+k_2}
    D(\vec{s}_1,M_1)
    \big[
        X(\vec{x}_1)D(\vec{s}_2,M_2)X(\vec{x}_1)
    \big]
    X(\vec{x}_1\oplus\vec{x}_2).
\end{align*}
Applying \cref{eq:dcp-diagonal-conjugation} and using
\[
    D(\vec{s},M)D(\vec{s}',M')
    =
    D(\vec{s}+\vec{s}',M+M')
\]
gives the stated expression.

For the inverse formula, note that
$D(\vec{s}_1,M_1)^{-1}=D(-\vec{s}_1,M_1)$. Hence
\[
    g_1^{-1}
    =
    \omega_8^{-k_1}
    \big[
        X(\vec{x}_1)D(-\vec{s}_1,M_1)X(\vec{x}_1)
    \big]
    X(\vec{x}_1),
\]
and another application of
\cref{eq:dcp-diagonal-conjugation} gives the result.

To obtain the commutator, let
\begin{align*}
    \vec{a}_{12}
    &=
    \vec{s}_1+\vec{s}_2\odot(1-2\vec{x}_1)
    +2M_2\vec{x}_1,
    \\
    \vec{a}_{21}
    &=
    \vec{s}_2+\vec{s}_1\odot(1-2\vec{x}_2)
    +2M_1\vec{x}_2
\end{align*}
be the linear parameters in the product formulas for
$g_1g_2$ and $g_2g_1$. Their difference satisfies
\[
    \vec{a}_{12}-\vec{a}_{21}
    =
    2\big(
        M_2\vec{x}_1+\vec{s}_2\odot\vec{x}_1
        +M_1\vec{x}_2+\vec{s}_1\odot\vec{x}_2
    \big)
    =
    2\vec{z}
    \pmod 4.
\]
The two products have the same quadratic parameter
$M_1+M_2$ and the same $X$ parameter
$\vec{x}_1\oplus\vec{x}_2$. Their scalar factors differ
by $i^c$, where
\[
    c=
    \varphi_{\vec{s}_2,M_2}(\vec{x}_1)
    -
    \varphi_{\vec{s}_1,M_1}(\vec{x}_2)
    \pmod 4.
\]
Since $D(2\vec{z},0)=Z(\vec{z})$, it follows that
\[
    g_1g_2=i^c Z(\vec{z})g_2g_1.
\]
Multiplying on the right by $g_1^{-1}g_2^{-1}$ proves
the commutator formula.

Finally,
\begin{align*}
    g_1g_2g_1^{-1}
    &=\gcomm{g_1}{g_2}g_2\\
    &=i^c Z(\vec{z})\,
      \omega_8^{k_2}D(\vec{s}_2,M_2)X(\vec{x}_2)\\
    &=\omega_8^{k_2+2c}
      D(\vec{s}_2+2\vec{z},M_2)X(\vec{x}_2),
\end{align*}
which proves the conjugation formula.
\end{proof}

\subsection{Howell normal form and canonical remainders}
\label{app:howell}

Howell normal form generalizes reduced row-echelon form
from fields to rings $\Z_m$. It provides a canonical
generating matrix for a submodule and a reduction
procedure for choosing representatives of its cosets.
Throughout, we fix the coordinate order and represent
elements of $\Z_m$ by integers in $\{0,\ldots,m-1\}$.

\begin{definition}[Howell normal form]
    Let $N\leq\Z_m^d$. A matrix $H$ with nonzero rows
    $\vec{h}_1,\ldots,\vec{h}_r$ is a \emph{Howell
    normal form} of $N$ if its rows generate $N$ and:
    \begin{enumerate}
        \item The first nonzero entry of $\vec{h}_i$
        occurs in column $j_i$, with
        $j_1<\cdots<j_r$.

        \item Each pivot $p_i:=H_{i,j_i}$ is a positive
        divisor of $m$.

        \item Entries above a pivot satisfy
        $0\leq H_{\ell,j_i}<p_i$ for $\ell<i$.

        \item For every $j\in\{1,\ldots,d\}$,
        \begin{align}
            \{\vec{a}\in N:
              a_1=\cdots=a_{j-1}=0\}
            =
            \operatorname{span}_{\Z_m}
            \{\vec{h}_i:j_i\geq j\}.
            \label{eq:howell-tail}
        \end{align}
    \end{enumerate}
    Zero rows are omitted.
\end{definition}

The last condition ensures that cancellations in early
coordinates are fully captured by the later rows.
For example, over $\Z_8$, the row $(2,1)$ alone is
in echelon form but not Howell form: multiplying it
by $4$ gives $(0,4)$. The Howell normal form of the
module it generates is
\begin{align}
    H=
    \begin{pmatrix}
        2 & 1\\
        0 & 4
    \end{pmatrix}.
\end{align}

Every submodule has a unique Howell normal form.
Given an $s\times d$ generating matrix, this form
can be computed in polynomial time
\cite{storjohann1998fast}.
Thus it depends only on the embedded submodule,
not on its chosen generators. Since pivot positions
are distinct, it has at most $d$ nonzero rows.

The following lemma gives the reduction procedure
used to choose canonical DCP lifts.

\begin{lemma}[Canonical remainders]
    \label{lem:howell-remainders}
    Let $N\leq\Z_m^d$ have Howell normal form $H$,
    with pivot positions $j_1,\ldots,j_r$ and
    pivots $p_1,\ldots,p_r$.
    Every coset of $N$ contains exactly one vector
    $\vec{a}$ satisfying
    \begin{align}
        0\leq a_{j_i}<p_i
        \qquad (i=1,\ldots,r).
        \label{eq:howell-reduced}
    \end{align}
    This vector can be computed by successively
    reducing the pivot coordinates using the rows
    of $H$.
\end{lemma}

\begin{proof}
    We first construct a reduced representative,
    then prove its uniqueness.

    Start with any representative $\vec{a}$ of
    the coset. At step $i$, write its current
    pivot coordinate as
    \begin{align}
        a_{j_i}=q_ip_i+t_i,
        \qquad 0\leq t_i<p_i,
    \end{align}
    and replace $\vec{a}$ by
    $\vec{a}-q_i\vec{h}_i$, with arithmetic performed
    modulo $m$. This preserves the coset and makes
    the $j_i$-th coordinate equal to $t_i$.
    Since $\vec{h}_i$ vanishes before column $j_i$,
    previously reduced pivot coordinates remain
    unchanged. After all rows have been processed,
    \cref{eq:howell-reduced} holds.

    For uniqueness, suppose $\vec{a}$ and $\vec{b}$
    are distinct reduced representatives of the
    same coset. Then
    $\vec{c}:=\vec{a}-\vec{b}\in N$.
    Let $j$ be the first nonzero coordinate of
    $\vec{c}$. By \cref{eq:howell-tail},
    $\vec{c}$ is a linear combination of rows
    whose pivots occur at or after $j$.
    If $j$ were not a pivot position, all such
    rows would vanish in column $j$, a contradiction.
    Hence $j=j_i$ for some $i$, and $c_{j_i}$
    is a multiple of $p_i$ in $\Z_m$.

    Since $p_i$ divides $m$, the integer
    $a_{j_i}-b_{j_i}$ is therefore divisible by
    $p_i$. But both coordinates lie in
    $\{0,\ldots,p_i-1\}$, so
    \begin{align}
        -(p_i-1)
        \leq a_{j_i}-b_{j_i}
        \leq p_i-1.
    \end{align}
    The only multiple of $p_i$ in this interval
    is zero, contradicting the choice of $j$.
    Thus the reduced representative is unique.
\end{proof}

Write $\operatorname{rem}_H(\vec{a})$ for this
unique representative of $\vec{a}+N$.
The lemma immediately gives
\begin{align}
    \vec{a}\in N
    &\iff \operatorname{rem}_H(\vec{a})=\vec{0},
    \\
    \vec{a}+N=\vec{b}+N
    &\iff
    \operatorname{rem}_H(\vec{a})
    =\operatorname{rem}_H(\vec{b}).
\end{align}
Once $H$ is known, reduction takes $O(rd)$
coordinate operations and integer divisions on
representatives smaller than $m$.
In particular, for the fixed modulus $m=8$, it
takes $O(d^2)$ time. The remainder map need not
be linear; it selects representatives of quotient
classes without asserting that the quotient splits.

To apply this construction to a DCP subgroup $G$,
take
\begin{align}
    N=\iota(\Delta(G))\leq\Z_8^d,
    \qquad d=1+n+\binom{n}{2},
\end{align}
and let $H$ be its Howell normal form.
For a lift $h_j=D_jX(\vec{x}^{(j)})\in G$, define
$\widetilde{D}_j$ by
\begin{align}
    \iota(\widetilde{D}_j)
    =
    \operatorname{rem}_H\big(\iota(D_j)\big).
\end{align}
This is a valid diagonal DCP operator: reduction
changes $\iota(D_j)$ by an element of
$\iota(\Delta(G))$, so the result remains in the
image of $\iota$. Moreover,
$\widetilde{D}_jD_j^{-1}\in\Delta(G)$, and hence
\begin{align}
    \widetilde{h}_j
    :=
    \widetilde{D}_jX(\vec{x}^{(j)})
    =
    (\widetilde{D}_jD_j^{-1})h_j
    \in G.
\end{align}
All lifts of $\vec{x}^{(j)}$ have diagonal coordinate
vectors in the same coset of $N$. Therefore
$\widetilde{h}_j$ depends only on $G$, $\vec{x}^{(j)}$,
and the fixed coordinate order, not on the
initial lift.

\section{DCP tableaux: equivalence and update rules}
\label{appendix:dcp-tableaux}

This appendix provides the remaining details for comparing and updating DCP tableaux.
We give a polynomial-time test for whether two tableaux represent the same pure state, then prove the stabilizer update rules for $Z$-type measurements.

\subsection{Equivalence of DCP tableaux}\label{app:tableau-equivalence}

The stabilizing closure introduced in
\Cref{sec:dcp-stabilizer-formalism} identifies two DCP-stabilizer groups
precisely when they have the same fixed space. Here we give a direct
polynomial-time test for this equivalence in the pure-state case.

Let \(G_a\leq\mathcal D^n\) be a DCP-stabilizer group for
\(\ket{\psi_a}\), \(a\in\{1,2\}\), and retain a support representative
\(x_a\in\supp(\psi_a)\). Thus the data to be compared are
\begin{equation}
    T_a=(G_a,x_a),
    \qquad
    \Fix(G_a)=\C\ket{\psi_a}.
\end{equation}
Such a representative is already part of the phase-polynomial description
in \cref{eq:simulation-representation}; it may equally be carried
alongside the canonical group representation. There is no need for separate
update rules: its affine support is updated by the rules of
\Cref{sec:simulator-phase-polynomial-representation}.

Set \(V_a=X(G_a)\). Since the fixed space of \(G_a\) is one-dimensional,
\Cref{thm:atomic-basis} implies
\begin{equation}
    \supp(\psi_a)=x_a\oplus V_a.
    \label{eq:tableau-affine-support}
\end{equation}
Consequently the supports agree if and only if
\begin{equation}
    V_1=V_2
    \quad\text{and}\quad
    x_1\oplus x_2\in V_1.
    \label{eq:tableau-support-test}
\end{equation}

Suppose this test passes, write the common space as \(V\), and choose a basis
\(v^{(1)},\ldots,v^{(r)}\) of \(V\). For every \(i\), choose lifts
\begin{equation}
    g_{a,i}\in G_a,
    \qquad
    \pi(g_{a,i})=v^{(i)},
\end{equation}
and define the diagonal quotient
\begin{equation}
    \delta_i:=g_{1,i}g_{2,i}^{-1}\in\ker\pi.
\end{equation}
The lifts can be obtained by row reduction of the \(X\)-parts while carrying
the corresponding DCP operators.

\begin{proposition}[Tableau-equivalence criterion]
    \label{prop:tableau-equivalence}
    The tableaux \(T_1\) and \(T_2\) represent the same state up to global
    phase if and only if \cref{eq:tableau-support-test} holds and every
    $\delta_i$ acts as the identity on the common support
    $S=x_1\oplus V$.
\end{proposition}

\begin{proof}
    If the states agree up to global phase, their supports agree, giving
    \cref{eq:tableau-support-test}. Moreover, both \(g_{1,i}\) and
    \(g_{2,i}\) fix the same state. Their quotient is diagonal, so it must
    act as the identity on every computational-basis vector with nonzero
    amplitude, hence on all of \(S\).

    Conversely, suppose the stated conditions hold and write
    \begin{equation}
        \ket{\psi_a}=\sum_{x\in S}\alpha_a(x)\ket{x}.
    \end{equation}
    Decompose each lift as \(g_{a,i}=D_{a,i}X(v^{(i)})\), with \(D_{a,i}\)
    diagonal, and write
    \(D_{a,i}\ket{x}=d_{a,i}(x)\ket{x}\). The stabilizer equations give
    \begin{equation}
        \alpha_a(x)=d_{a,i}(x)\alpha_a(x\oplus v^{(i)}).
    \end{equation}
    Since
    \(\delta_i=D_{1,i}D_{2,i}^{-1}\) is the identity on \(S\), we have
    \(d_{1,i}(x)=d_{2,i}(x)\) there. Hence
    \(\alpha_1(x)/\alpha_2(x)\) is invariant under translation by every
    \(v^{(i)}\). These translations connect all points of \(S\), so this ratio
    is constant and \(\ket{\psi_1}\) and \(\ket{\psi_2}\) differ only by a
    global phase.
\end{proof}

It remains only to make the identity test explicit. Write
\begin{equation}
    \delta_i=\omega_8^{k_i}D(s_i,M_i),
\end{equation}
and let \(B\) be the matrix with columns \(v^{(1)},\ldots,v^{(r)}\). Its phase on a
basis vector is the Clifford phase polynomial
\begin{equation}
    F_i(x)
    =k_i+2s_i\cdot x
      +4\sum_{j<\ell}(M_i)_{j\ell}x_jx_\ell
      \pmod 8.
\end{equation}
By \Cref{cor:phase-polynomial-affine-pullback,thm:clifford-phase-polynomial},
the restriction \(c\mapsto F_i(x_1\oplus Bc)\) has Boolean degree at most
two, with its linear and quadratic coefficients divisible by \(2\) and \(4\),
respectively. It is therefore identically zero if and only if it vanishes at
\begin{equation}
    c=0,
    \qquad c=e_j,
    \qquad c=e_j\oplus e_\ell
    \quad (1\leq j<\ell\leq r).
    \label{eq:tableau-phase-test-points}
\end{equation}
Thus each diagonal quotient requires only \(O(r^2)\) phase evaluations. The
support comparison, construction of the lifts, and all phase tests are
polynomial-time operations. The result is independent of the chosen lifts:
two lifts of the same \(v^{(i)}\) differ by a diagonal stabilizer, which acts
identically on the corresponding support.

\subsection{Measurement update rules for DCP tableaux}
\label{appendix:tableau-update-measurements}

\begin{dcpzmeasurementrestatement}[Updating under $Z$-measurements]
\dcpZMeasurementUpdateStatement
\end{dcpzmeasurementrestatement}

\begin{proof}
    Let us first prove (a), starting with the outcome probabilities.
    Write $z=Z(\vec v)$, and choose $t\in G$ that anticommutes with $z$.
    Since $t$ fixes $\ket{\psi}$ and satisfies $t^\dagger zt=-z$, we have
    \begin{align}
        \bra{\psi}z\ket{\psi}
        =\bra{\psi}t^\dagger zt\ket{\psi}
        =-\bra{\psi}z\ket{\psi}.
    \end{align}
    Thus $\bra{\psi}z\ket{\psi}=0$.
    Hence, using the projector
    $\Pi_\varepsilon=(I+\varepsilon z)/2$, we have
    \begin{align}
        \Pr(\varepsilon)
        =\bra{\psi}\Pi_\varepsilon\ket{\psi}
        =\frac12.
    \end{align}

    It remains to show that $\Fix(H_\varepsilon)=\C\Pi_\varepsilon\ket{\psi}$.
    For the $\supseteq$ inclusion, every element of $K$ commutes with $\Pi_\varepsilon$ and fixes $\ket{\psi}$, while $(\varepsilon z)\Pi_\varepsilon=\Pi_\varepsilon$.
    Thus, $H_\varepsilon$ fixes $\Pi_\varepsilon\ket{\psi}$, and $\C\Pi_\varepsilon\ket{\psi}\subseteq\Fix(H_\varepsilon)$.

    We now prove the reverse inclusion.
    Let $\ket{\varphi}\in\Fix(H_\varepsilon)$.
    Since $t$ anticommutes with $z$, both $t^2$ and $t^{-1}kt$ commute with $z$ for every $k\in K$, and therefore belong to $K$.
    Consequently, $(I+t)\ket{\varphi}$ is fixed by both $K$ and $t$.
    Moreover, $G=\langle K,t\rangle$.
    Indeed, every element $g\in G$ either commutes or anticommutes with $z$ by \cref{lem:z-commute}.
    In the first case, $g\in K$; in the second, $t^{-1}g\in K$, and therefore $g\in tK$ and $G=\langle K,t\rangle$.
    Hence,
    \begin{align}
        (I+t)\ket{\varphi}\in\Fix(G)=\C\ket{\psi}.
    \end{align}
    Moreover, $z\ket{\varphi}=\varepsilon\ket{\varphi}$ and $zt=-tz$, so
    \begin{align}
        \ket{\varphi}
        =\Pi_\varepsilon(I+t)\ket{\varphi}
        \in\C\Pi_\varepsilon\ket{\psi}.
    \end{align}
    This proves the reverse inclusion.
    Since $\Pi_\varepsilon\ket{\psi}\neq 0$, the fixed space is one-dimensional.

    We now verify that the stated generator update produces $K$.
    By \cref{prop:dcp-generators}, it suffices to obtain generators of $\Delta(K)$ and lifts of a basis of $X(K)$.
    We will show that the diagonal generators can be retained, and that the stated row operations produce the required lifts.

    First, every diagonal element of $G$ commutes with $z$, so $\Delta(K)=\Delta(G)$.
    Thus the original diagonal generators already generate $\Delta(K)$.

    It remains to construct lifts of a basis of $X(K)$.
    By \cref{lem:z-commute}, an element of $G$ with $X$-part $\vec{x}$ commutes with $z=Z(\vec v)$ exactly when $\langle\vec{x},\vec v\rangle=0$.
    Therefore,
    \begin{align}
        X(K)=X(G)\cap\vec v^\perp.
    \end{align}
    Let $g_1,\ldots,g_r$ be the original lifts, whose $X$-parts $\vec{x}^{(1)},\ldots,\vec{x}^{(r)}$ form a basis of $X(G)$.
    Choose an anticommuting lift $g_j$, so that $\langle\vec{x}^{(j)},\vec v\rangle=1$.
    For each $i\neq j$, define
    \begin{align}
        h_i=
        \begin{cases}
            g_i, & \langle\vec{x}^{(i)},\vec v\rangle=0,\\
            g_i g_j, & \langle\vec{x}^{(i)},\vec v\rangle=1.
        \end{cases}
    \end{align}
    Each $h_i$ commutes with $z$, and its $X$-part is
    \begin{align}
        \vec{y}^{(i)}
        =\vec{x}^{(i)}+\langle\vec{x}^{(i)},\vec v\rangle\vec{x}^{(j)}.
    \end{align}
    We must show that the vectors $\vec{y}^{(i)}$, for $i\neq j$, form a basis of $X(K)$.
    We first show that the $\vec{y}^{(i)}$ span $X(K)$.
    Let $\vec{x}\in X(K)$ and expand it in the original basis:
    \begin{align}
        \vec{x}=\sum_{i=1}^r a_i\vec{x}^{(i)}.
    \end{align}
    Since $\langle\vec{x},\vec v\rangle=0$ and $\langle\vec{x}^{(j)},\vec v\rangle=1$, we obtain
    \begin{align}
        0=a_j+\sum_{i\neq j}a_i\langle\vec{x}^{(i)},\vec v\rangle,
    \end{align}
    so $a_j=\sum_{i\neq j}a_i\langle\vec{x}^{(i)},\vec v\rangle$, where all arithmetic is over $\mathbb F_2$.
    Substituting this expression gives
    \begin{align}
        \vec{x}
        &=\sum_{i\neq j}a_i
        \bigl(\vec{x}^{(i)}+\langle\vec{x}^{(i)},\vec v\rangle\vec{x}^{(j)}\bigr)
        =\sum_{i\neq j}a_i\vec{y}^{(i)}.
    \end{align}
    Thus the $\vec{y}^{(i)}$ span $X(K)$.

    To show linear independence, suppose that $\sum_{i\neq j}b_i\vec{y}^{(i)}=0$.
    Expanding in the original basis gives
    \begin{align}
        \sum_{i\neq j}b_i\vec{x}^{(i)}
        +\left(\sum_{i\neq j}b_i\langle\vec{x}^{(i)},\vec v\rangle\right)\vec{x}^{(j)}
        =0.
    \end{align}
    Since the $\vec{x}^{(i)}$ are linearly independent, the coefficient $b_i$ of each $\vec{x}^{(i)}$ with $i\neq j$ must vanish.
    Therefore the $\vec{y}^{(i)}$ are linearly independent and form a basis of $X(K)$.

    Consequently, the retained diagonal generators together with the lifts $h_i$ generate $K$.
    Adding $\varepsilon z$ therefore gives generators of $H_\varepsilon=\langle K,\varepsilon z\rangle$, as claimed.

    We now prove (b). Suppose that $z$ commutes with every element of $G$.
    Then $z$ preserves $\Fix(G)$: if $g\in G$ and $\ket{\varphi}\in\Fix(G)$, then
    $gz\ket{\varphi}=z\ket{\varphi}$.  Since
    $\Fix(G)=\C\ket{\psi}$, the state $\ket{\psi}$ is an eigenvector of $z$, and we have $(\varepsilon z)\ket{\psi}=\ket{\psi}$, and hence $\varepsilon z\in\mathcal{D}(\psi)$.  The post-measurement state is
    $\ket{\psi}$ itself, so its DCP-stabilizer group $G$ is unchanged.
\end{proof}

\bibliography{main}

\end{document}